\documentclass[sn-mathphys,Numbered]{sn-jnl}

\usepackage{graphicx}%
\usepackage{multirow}%
\usepackage{amsmath,amssymb,amsfonts}%
\usepackage{amsthm}%
\usepackage{mathrsfs}%
\usepackage[title]{appendix}%
\usepackage{xcolor}%
\usepackage{textcomp}%
\usepackage{manyfoot}%
\usepackage{booktabs}%
\usepackage{algorithm}%
\usepackage{algorithmicx}%
\usepackage{algpseudocode}%
\usepackage{listings}%
\usepackage{amsbsy}
\usepackage{mathtools}
\usepackage{braket}
\usepackage{xcolor}
\usepackage{bm}
\usepackage[caption=false]{subfig}
\usepackage{multirow}
\usepackage{eczoo}
\usepackage{stmaryrd}
\usepackage{bbold}

\newcommand{\e}{\mathrm{e}}

\newcommand{\bs}[1]{{\boldsymbol{#1}}}
\newcommand{\im}{{\rm im}}
\newcommand{\Hom}{{\rm Hom}}
\newcommand{\highlight}[2]{\colorbox{#1!50}{$\displaystyle#2$}}

\newcommand{\id}{\mathbb{1}}
\newcommand{\Diag}[1]{{\rm Diag}\left(#1\right)}

\theoremstyle{thmstyleone}%
\newtheorem{thm}{Theorem}
\newtheorem{prop}{Proposition}
\newtheorem{lem}{Lemma}

\theoremstyle{thmstyletwo}%
\newtheorem*{remark}{Remarks}%
\newtheorem{examp}{Example}

\theoremstyle{thmstylethree}%
\newtheorem{defi}{Definition}

\begin{document}

\title[Homological Quantum Rotor Codes]{Homological Quantum Rotor Codes: Logical Qubits from Torsion}


\author*[1]{\fnm{Christophe} \sur{Vuillot}}\email{christophe.vuillot@inria.fr}

\author[2]{\fnm{Alessandro} \sur{Ciani}}

\author[3,4]{\fnm{Barbara M.} \sur{Terhal}}

\affil*[1]{\orgname{Université de Lorraine, CNRS, Inria, LORIA}, \orgaddress{\postcode{F-54000}, \city{Nancy}, \country{France}}}

\affil[2]{\orgdiv{Institute for Quantum Computing Analytics (PGI-12)}, \orgname{Forschungszentrum J\"ulich}, \orgaddress{\postcode{52425}, \city{J\"ulich}, \country{Germany}}}

\affil[3]{\orgdiv{QuTech}, \orgname{Delft University of Technology}, \orgaddress{\street{Lorentzweg 1}, \postcode{2628 CJ}, \city{Delft}, \country{The Netherlands}}}
\affil[4]{\orgdiv{DIAM, EEMCS Faculty}, \orgname{Delft University of Technology}, \orgaddress{\street{Van Mourik Broekmanweg 6}, \postcode{2628 XE}, \city{Delft}, \country{The Netherlands}}}


\abstract{We formally define homological quantum rotor codes which use multiple quantum rotors to encode logical information.
These codes generalize homological or CSS quantum codes for qubits or qudits, as well as linear oscillator codes which encode logical oscillators.
Unlike for qubits or oscillators, homological quantum rotor codes allow one to encode both logical rotors and logical qudits in the same block of code, depending on the homology of the underlying chain complex. In particular, a code based on the chain complex obtained from tessellating the real projective plane or a M\"{o}bius strip encodes a qubit.
We discuss the distance scaling for such codes which can be more subtle than in the qubit case due to the concept of logical operator spreading by continuous stabilizer phase-shifts.
We give constructions of homological quantum rotor codes based on 2D and 3D manifolds as well as products of chain complexes.
Superconducting devices being composed of islands with integer Cooper pair charges could form a natural hardware platform for realizing these codes: we show that the $0$-$\pi$-qubit as well as Kitaev's current-mirror qubit --also known as the M\"{o}bius strip qubit-- are indeed small examples of such codes and discuss possible extensions. }




\maketitle

\section{Introduction}

Quantum computation is most conveniently defined as quantum circuits acting on a system composed of elementary finite dimensional subsystems such as qubits or qudits.
However, for various quantum computing platforms, the underlying hardware can consist of both continuous and countably-infinite (discrete) degrees of freedom.
For such platforms, one thus considers how to encode a qubit in a well-chosen subspace, and how the possibly continuous nature of the errors affects the robustness of the encoded qubit.
This paradigm has been explored for bosonic encodings \cite{TCV:bosonic,albert:bosonic}, that is, one makes a choice for a qubit(s) subspace inside one or multiple oscillator spaces, such as the GKP code \cite{GKP}. Other work has formalized encodings of a qubit into a rotating body (such as that of a molecule) \cite{ACP:rotor} or many qubits into a single (planar) rotor \cite{kalev+}.
In these works the discreteness of the encoding is arrived at, similar to the GKP code, by selecting stabilizer operators which constrain the angular momentum variable in a modular fashion.

In this paper, we take a rather different approach and introduce the formalism of homological quantum rotor codes. These are codes defined for a set of $n$ quantum rotors and form an extension of qudit homological codes and {\em linear} oscillator codes (defined e.g. in \cite{VAHWPT}) which do not make use of modular constraints to get non-trivial encodings. As for prior work, we note that the authors in \cite{bermejo-vegaNormalizerCircuitsGottesmanKnill2016} consider mainly stabilizer rotor states whereas here we consider stabilizer subspaces to encode information. 
Quantum rotor versions of the toric code and Haah's code have been previously proposed \cite{haahGeneralization2017, albertGeneralPhaseSpaces2017} and small rotor codes were studied as \(U(1)\)-covariant codes \cite{haydenErrorCorrectionQuantum2021, faistContinuousSymmetriesApproximate2020}: all of these are included in our formalism. The novelty of our formulation is that our codes can also encode finite dimensional logical systems such as qubits via the torsion of the underlying chain complex.

In the first half of the paper, after some preliminaries (Section \ref{sec:prelim}), we define homological quantum rotor codes (Section \ref{sec:def-codes}) and present mathematical constructions of such codes (Section \ref{sec:constructions}). 
In the second half of the paper (Section \ref{sec:circuit-QED}) we show how two well-known protected superconducting qubits, namely the $0$-$\pi$ qubit \cite{brooksProtectedGatesSuperconducting2013,dempsterUnderstandingDegenerateGround2014, gyenisExperimentalRealizationProtected2021} and Kitaev's current mirror qubit \cite{kitaevProtectedQubitBased2006, weissSpectrumCoherenceProperties2019, kitaev2007noncommuting}, can be viewed as examples of small rotor codes. The code perspective is useful as it allows one to understand the mechanism behind its native noise protection in terms of the code distance. By the continuous nature of the phase-shift errors leading to {\em logical operator spreading}, we discuss how the $Z$ distance of a rotor code is more subtle than for qubit codes, see Section \ref{sec:spread}. We hope that our work stimulates further research into the effect of homology on encodings and the robustness of the encoded qubit(s), be it via active or passive quantum error correction.

\section{Notations and Preliminaries}
 \label{sec:prelim}
	We denote the group of integers as $\mathbb{Z}$ and the phase group or circle group as $\mathbb{T} = [0,2\pi) = \mathbb{R}/2\pi\mathbb{Z}$.
	We denote the cyclic group $\mathbb{Z}/d\mathbb{Z}$ as $\mathbb{Z}_d$ and the subgroup of $\mathbb{T}$ isomorphic to $\mathbb{Z}_d$ as $\mathbb{Z}^*_d$, that is,
	\begin{equation}
	    \mathbb{Z}^*_d = \frac{2\pi}{d}\mathbb{Z}_d = \left\{\frac{2\pi}{d}k\in \mathbb{T}\;\middle\vert\; k\in\mathbb{Z}_d\right\}.\label{eq:Zstar}
	\end{equation}

     These groups are Abelian and we denote their operations with \(+\).
	The groups $\mathbb{Z}^n$ or $\mathbb{T}^n$ do not have the structure of vector spaces.
	They do have the structure of \(\mathbb{Z}\)-modules since multiplication of an element of \(\mathbb{Z}^n\) or \(\mathbb{T}^n\) by an integer is naturally defined by repeated addition.
	The homomorphisms of \(\mathbb{Z}\)-modules or \(\mathbb{Z}\)-linear maps are well defined using integer matrices and matrix multiplication.
	As such we will abuse the terminology and refer to elements of \(\mathbb{Z}^n\), \(\mathbb{T}^n\) as integer vectors and phase vectors respectively.	
	We always use bold font to denote these vectors $\bs{m}\in \mathbb{Z}^n$ and $\bs{\phi} \in \mathbb{T}^n$.
    Note that there is no unambiguous definition for a multiplication operation over $\mathbb{T}$ so we can never write $\phi\phi^\prime$ for some \((\phi,\phi^\prime)\in\mathbb{T}^2\).
	
 \subsection{Quantum Rotors}
	
	We consider a quantum system called a quantum rotor
	\cite{bermejo-vegaNormalizerCircuitsGottesmanKnill2016, albertGeneralPhaseSpaces2017}.
	The Hilbert space, $\mathcal{H}_{\mathbb{Z}}$, where this system lives has a countably-infinite basis indexed by $\mathbb{Z}$, 
	\[\forall\ell\in\mathbb{Z},\quad \ket{\ell}\in\mathcal{H}_{\mathbb{Z}}.\]
	The states in $\mathcal{H}_{\mathbb{Z}}$ are therefore sequences of complex numbers which are square-integrable, so in $\ell^2(\mathbb{Z})$, and normalized,
	\[\ket{\psi} = \sum_{\ell\in\mathbb{Z}}\alpha_\ell\ket{\ell}, \qquad\sum_{\ell\in\mathbb{Z}}\vert\alpha_\ell\vert^2 = 1.\]
These states could represent the states of definite angular momentum in a mechanical setting, hence the name rotors (or planar rotors). Using the discrete-time Fourier transform\footnote{The name ``discrete-time'' here comes from Fourier analysis of functions sampled periodically at discrete times (indexed by $\mathbb{Z}$).} we can also represent the states as square-integrable functions on the circle $\mathbb{T}$:
	\[\ket{\psi} = \sum_{\ell\in\mathbb{Z}}\alpha_\ell\ket{\ell} = \int_{\theta\in\mathbb{T}}\!{\rm d}\theta\,\psi(\theta)\ket{\theta},\qquad\int_{\theta\in\mathbb{T}}\!{\rm d}\theta\,\vert\psi(\theta)\vert^2 = 1,\]
	with the function $\psi(\theta)$ given by
 	\begin{equation}
	    \psi(\theta) = \braket{\theta\vert\psi} = \frac{1}{\sqrt{2\pi}} \sum_{\ell\in\mathbb{Z}}\e^{i\ell \theta}\alpha_\ell,
	\end{equation}
 or formally
\begin{align}
\ket{\ell}=\frac{1}{\sqrt{2\pi}}\int_{\mathbb{T}} d\theta e^{i\ell \theta} \ket{\theta}.
\label{eq:four}
\end{align}
	Even though they are not physical states in the Hilbert space (as they do not correspond to square-integrable functions), we define the un-normalized phase states, $\ket{\theta}$, as
	\begin{equation}\forall\theta\in\mathbb{T},\quad \ket{\theta} = \frac{1}{\sqrt{2\pi}}\sum_{\ell\in\mathbb{Z}}\e^{-i\ell\theta}\ket{\ell},\; \bra{\theta'} \theta \rangle=\delta(\theta'-\theta).
		\label{eq:FT}
	\end{equation}

	\subsection{Generalized Pauli Operators}
	
	 The generalized Pauli operators, $X(m)_{m\in\mathbb{Z}}$ and $Z(\phi)_{\phi\in\mathbb{T}}$, for a single rotor are defined as 
	\begin{align}
	    \forall m \in \mathbb{Z}, \quad & X(m)= e^{i m \hat{\theta}}, \\
	    \forall \phi \in \mathbb{T}, \quad & Z(\phi) = e^{i \phi \hat{\ell}}, \label{eq:phasedef}
	\end{align}
	where the angular momentum operator $\hat{\ell}$ can be defined via its action on the phase and angular momentum basis: 
	\begin{align}
	    \forall \ell \in \mathbb{Z}, \quad &\hat{\ell} \ket{\ell}  =  \ell \ket{\ell}, \;\;
	    \langle \theta| \hat{\ell} \ket{\psi} = -i \frac{\partial}{\partial \theta} \psi(\theta).
	\end{align}
 We note that ``the phase operator $\hat{\theta}$'' in Eq.~\eqref{eq:phasedef} is only a convenient notation as it should only occur in \(2\pi\)-periodic functionals, see for instance \cite{hallQuantumTheoryMathematicians2013}.
	The action of the generalized Pauli operators on the states $\ket{\theta}$ and the angular momentum eigenstates $\ket{\ell}$ reads
\begin{align}
	    X(m)\ket{\theta} &=  \e^{i m\theta}\ket{\theta},  & Z(\phi)\ket{\theta} & = \ket{\theta - \phi \pmod {2 \pi}},  \\
	   X(m) \ket{\ell}&=\ket{\ell+m}, &  Z(\phi)\ket{\ell} & = e^{i\phi\ell} \ket{\ell},
	    	    \label{eq:defZtheta}
	\end{align}
	    	which can be verified through the Fourier transform in Eq.~\eqref{eq:FT}.
	By direct computation we have the following properties
	\begin{align}
	\id = X(0) &= Z(0),\label{eq:XZid}\\
	X(m_1)X(m_2) &= X(m_1 + m_2),\label{eq:Xadd}\\
	Z(\phi_1)Z(\phi_2) &= Z(\phi_1+\phi_2),\label{eq:Zadd}
	\end{align}
and the commutation relation
\begin{equation}
    	X(m)Z(\phi) = \e^{-i\phi m} Z(\phi)X(m).
    	\label{eq:paulicommutation}
\end{equation}
    When we consider systems of $n$ quantum rotors, we write $X_j(m)$ and $Z_j(\phi)$ for the generalized Pauli operators acting on the $j$th rotor.
	For a multi-rotor Pauli operator, using the bold vector notation for row vectors $\bs{m}$ and $\bs{\phi}$, we have
	\begin{align}
	\bs{m}\in\mathbb{Z}^n,\;X(\bs{m}) &= \prod_{j=1}^nX_j(m_j),\label{eq:multiXop}\\ \bs{\phi}\in\mathbb{T}^n,\;Z(\bs{\phi}) &= \prod_{j=1}^nZ_j(\phi_j).\label{eq:multiZop}
	\end{align}
	The commutation relations for multi-rotor operators are straightforwardly computed from Eq.~\eqref{eq:paulicommutation}
	\begin{equation}
	\forall(\bs{m},\bs{\phi})\in\mathbb{Z}^{n}\times\mathbb{T}^n,\;X(\bs{m})Z(\bs{\phi}) = \e^{-i\bs{m} \cdot\bs{\phi}}Z(\bs{\phi})X(\bs{m}).\label{eq:commutationmulti}
	\end{equation}
     Finally, it is not hard to show that the generalized Pauli operators form an operator basis for the rotor space. 
	\section{Definition of Homological Quantum Rotor Codes}
	\label{sec:def-codes}
	In this section we define homological quantum rotor codes.
    We describe their codestates, logical codespace and logical operators.
    For this we explain their connection to the homology and co-homology of chain-complexes.
    We also propose a generic noise model and define the related notion of code distance for it as well as ways to bound this distance.
    Finally we explain how to define a Hamiltonian whose groundspace coincides with the codespace of a homological quantum rotor code.
    
\subsection{Definition and Relation to Homology}
\label{sec:logicalsandhomology}

	We consider defining a subspace, the codespace or logical subspace, of the  Hilbert space of \(n\) quantum rotors, in a \eczoo[CSS]{css} fashion \cite{calderbankGoodQuantumErrorcorrecting1996, steaneMultipleparticleInterferenceQuantum1997}.
	\begin{defi}[Homological Quantum Rotor Code, \(\mathcal{C}^{\rm rot}(H_X, H_Z)\)]
    Let $H_X$ and $H_Z$ be two integer matrices of size $r_x\times n$ and $r_z\times n$ respectively, such that
    \begin{equation}
    H_X H_Z^T = 0.\label{eq:CSScond}
\end{equation}
We define the following group of operators, $\mathcal{S}$, and call it the stabilizer group:
\begin{equation}
    \mathcal{S} = \left\langle Z(\bs{\varphi}H_Z) X(\bs{s}H_X) \;\middle\vert\;\forall \bs{\varphi}\in\mathbb{T}^{r_z},\,\forall \bs{s}\in\mathbb{Z}^{r_x}\right\rangle.
    \label{eq:stabilizergroup}
\end{equation}
We then define the corresponding homological quantum rotor code, $\mathcal{C}^{\rm rot}(H_X,H_Z)$, on $n$ quantum rotors as the subspace stabilized by \(\mathcal{S}\):
\begin{equation}
    \mathcal{C}^{\rm rot}(H_X,H_Z) = \big\{\ket{\psi}\;\big\vert\; \forall P\in\mathcal{S},\,P\ket{\psi} = \ket{\psi}\big\}.
\end{equation}
\label{def:rotorcode}
\end{defi}

The matrices $H_X$ and $H_Z$ are used to generate the $X$ and $Z$ parts of the stabilizer group $\mathcal{S}$.
Stabilizers of $X$- or $Z$-type are labeled by integer vectors, $\bs{s}\in\mathbb{Z}^{r_x}$, or phase vectors, $\bs{\varphi}\in\mathbb{T}^{r_z}$, and are generated and denoted as follows:
	\begin{align}
	    S_{X}(\bs{s}) &= X(\bs{s}H_X),\label{eq:Xstab}\\
	    S_{Z}(\bs{\varphi}) &= Z(\bs{\varphi}H_Z).\label{eq:Zstab}
	\end{align}
	
	Condition~\eqref{eq:CSScond} implies that the stabilizers all commute and the code is therefore well defined.
	Indeed, one can check that
	\begin{align}
		S_Z(\bs{\varphi})S_X(\bs{s}) &= Z\left (\bs{\varphi}H_Z\right )X\left (\bs{s}H_X\right ) \notag \\
		=\exp\left (i\bs{\varphi} H_Z H_X^T\bs{s}^T\right )S_X(\bs{s})S_Z(\bs{\varphi})
		&= S_X(\bs{s})S_Z(\bs{\varphi}), 
	\end{align}
	where we have used Eq.~\eqref{eq:commutationmulti} and Eq.~\eqref{eq:CSScond}. We will also refer to the generators of the stabilizer group given by the rows of \(H_X\) and \(H_Z\).
For this we take a basis of integer vectors $(\bs{s}_j)_i=\delta_{ij}$ and define
	\begin{align}
	    S_j^X \equiv S_X(\bs{s}_j)=X(\bs{h}_j^X)=e^{i \bs{h}_j^X \cdot \bs{\hat{\theta}}},
     \end{align}
     so we denote the $j$th row of $H_X$ as the integer vector $\bs{h}_j^X$.  The operators $S_j^X$ for $j=1,\ldots, r_x$ generate the $X$ part of the stabilizers. For the $Z$ part of the stabilizer, we can take $\varphi\bs{s}_j$, $\varphi \in \mathbb{T}$ and we define the continuous set of generators
	\begin{align}
	    S_j^Z(\varphi)\equiv S_Z(\varphi\bs{s}_j)=Z(\varphi \bs{h}_j^Z)=e^{i \varphi \bs{h}_j^Z \cdot \bs{\hat{\ell}}},
	    \label{eq:genZ}
	\end{align}
 so we denote the $j$th row of $H_Z$ as the integer vector $\bs{h}_j^Z$.\\
 
\begin{remark}
We do not consider stabilizer groups generated by finite subgroups of $\mathbb{T}$, such as for the GKP code and generalizations \cite{GKP, Conrad2022, nohEncodingOscillatorMany2020, ACP:rotor} defined on an oscillator space. Such stabilizers would correspond to modular measurements of angular momentum, while here we only need angular momentum measurements, see Section \ref{sec:meas-stab}. 
A second remark is that the same integer matrices used to define a quantum rotor code can be used to define a related qudit code which we then denote as \(\mathcal{C}^d(H_X,H_Z)\), see Appendix~\ref{app:quditcode}, Definition~\ref{def:quditcode} for the exact definition.
A third remark is that one can view the quantum rotor space as a subspace of an oscillator space with quadrature operators $\hat{p}$ and $\hat{q}$, with $[\hat{q},\hat{p}]=i\id$, obeying a `stabilizer' constraint $e^{i 2\pi \hat{p}}=\mathbb{1}$, enforcing that $\hat{p} \rightarrow \hat{\ell}$ has integer eigenvalues. From this perspective, the homological quantum rotor codes that we introduce can be viewed as a subclass of `general multi-mode GKP codes' which can encode both discrete as well as continuous information, see e.g. Appendix A in \cite{VAHWPT}.
Realizing homological quantum rotor codes as particular instances of multi-mode GKP codes would be as challenging as realizing any general multi-mode GKP code.
By contrast, when one has access to physical quantum systems natively behaving as rotors, the constraints to impose to realize a homological quantum rotor code are simpler since they are non-modular.
Still, this connection gives a way to design multi-mode GKP codes encoding rotors or qudits.\\
\end{remark}

Stabilizers also allow to define the syndrome of error states, i.e. the information that can be learned about errors, in the usual way.
To see this, pick a code state, \(\ket{\psi}\in\mathcal{C}^{\rm rot}(H_X,H_Z)\).
Consider error states, \(X(\bs{e})\ket{\psi}\) and \(Z(\bs\nu)\ket{\psi}\), and consider the stabilizers \(S_Z(\bs\varphi)\) and \(S_X(\bs{s})\):
\begin{align}
    X(\bs{e})\ket{\psi} &= X(\bs{e})S_Z(\bs{\varphi})\ket{\psi} = \e^{-i\bs{\varphi}H_Z\bs{e}^T}S_Z(\bs{\varphi})\left[X(\bs{e})\ket{\psi}\right],\\
     Z(\bs{\nu})\ket{\psi} &= Z(\bs{\nu})S_X(\bs{s})\ket{\psi} = \e^{i\bs{s}H_X\bs{\nu}^T}S_X(\bs{s})\left[Z(\bs{\nu})\ket{\psi}\right].
\end{align}
The error states \(X(\bs{e})\ket{\psi}\) and \(Z(\bs{\nu})\ket{\psi}\) are therefore eigenstates of \(S_Z(\bs{\varphi})\)  and \(S_X(\bs{s})\) respectively, with eigenvalue \(\e^{i\bs{\varphi}H_Z\bs{e}^T}\) and \(\e^{-i\bs{s}H_X\bs{\nu}^T}\).
The values \(H_Z\bs{e}^T\in\mathbb{Z}^{r_z}\) and \(H_X\bs{\nu}^T\in\mathbb{T}^{r_x}\) is what we call the syndrome.\\

To gain insight into the structure of these rotor codes, errors and logical operators, we make the standard connection between a quantum CSS code and the homology and cohomology of a chain complex.
The fact that \(X\)-type and \(Z\)-type operators have different underlying groups in the rotor case (unlike the qubit, qudit or oscillator case) makes this connection a bit more subtle but also clarifies it.

To construct a chain complex we are going to take the map from the \(X\) part of the stabilizer group, $\mathbb{Z}^{r_x}$, to the full group of \(X\)-Pauli operators, $\mathbb{Z}^n$, as well as the map from the \(X\)-Pauli operators to their syndromes \(\mathbb{Z}^{r_z}\).
We denote the map from the stabilizer group to the operator group by $\partial_2:\mathbb{Z}^{r_x}\rightarrow\mathbb{Z}^{n}$, given in matrix form by $H_X\in\mathbb{Z}^{r_x\times n}$ (and which is applied to a row vector in $\mathbb{Z}^{r_x}$ from the right).
	We denote the syndrome map by $\partial_1:\mathbb{Z}^n\rightarrow\mathbb{Z}^{r_z}$ which is given as a matrix by $H_Z^T\in\mathbb{Z}^{n\times r_z}$.
	The stabilizers are designed so that they all commute, that is to say, they have trivial syndrome, so that we have $\partial_1 \circ \partial_2 = 0$. 
 This is the defining property of a chain complex, hence we have a chain complex, $\mathcal{C}$,  with integer coefficients:
	\begin{equation}
	\begin{matrix}
	\mathcal{C}: & C_2 & \xrightarrow[H_X]{~~\partial_2~~} & C_1 & \xrightarrow[H_Z^T]{~~\partial_1~~} & C_0\\
	& \shortparallel && \shortparallel & &\shortparallel\\
	&\mathbb{Z}^{r_x} & & \mathbb{Z}^n & & \mathbb{Z}^{r_z}\\
	& \shortparallel && \shortparallel & &\shortparallel\\
	&\text{stabilizers} && \text{operators} && \text{syndrome}
	\end{matrix}.\label{eq:rotorchaincomplex}
	\end{equation}
Given a chain complex $\mathcal{C}$, the maps \(\partial_1\) and $\partial_2$ are called {\em boundary} maps and automatically obey $\partial_1 \circ \partial_2=0$ (the boundary of the boundary is zero).
Thus, given a chain complex over $\mathbb{Z}$ we can construct a rotor code or given a rotor code one can associate with it a chain complex. In order to find codes it is useful to consider chain complexes obtained from tessellations of manifolds as we do in Section \ref{sec:tes}.

Viewed as a chain complex $\mathcal{C}$, the group of logical $X$ operators corresponds to the first homology group over $\mathbb{Z}$, i.e. 
	\begin{equation}
	H_1(\mathcal{C},\mathbb{Z}) = \ker\partial_1/\im\partial_2 = \ker\left(H_Z^T\right)/\im\left(H_X\right) = \mathcal{L}_X.
	\end{equation}
	One can also understand the logical operators by constructing the code states, namely the $+1$-eigenstates of the stabilizers. The $Z$-type stabilizers constrain a code state $\ket{\overline{\psi}} = \sum_{\bs{\ell} \in \mathbb{Z}^n} \alpha_{\bs{\ell}}\ket{\bs{\ell}}$ in the following manner
	\begin{align}
	\forall\bs{\varphi} \in \mathbb{T}^{r_z},\;	\ket{\overline{\psi}} &= Z(\bs{\varphi}H_Z)\ket{\overline{\psi}}\nonumber\\
	\Rightarrow
		\sum_{\bs{\ell}\in\mathbb{Z}^n}\alpha_\bs{\ell}\ket{\bs{\ell}} &= \sum_{\bs{\ell}\in\mathbb{Z}^n}\e^{i\bs{\varphi} H_Z \bs{\ell}^T}\alpha_\bs{\ell}\ket{\bs{\ell}}\nonumber\\
		\Rightarrow \forall \bs{\ell},\; \alpha_\bs{\ell}\neq 0 &\Rightarrow H_Z \bs{\ell}^T =(\bs{\ell}H_Z^T)^T= 0.
	\end{align}
	This means that code states can only have support on angular momentum states which are in the left kernel of $H_Z^T$, which we denote as $\ker\left(H_Z^T\right)$. The $X$-type stabilizers constrain $\ket{\overline{\psi}}$ as follows:
	\begin{align}
	\forall \bs{s} \in \mathbb{Z}^{r_x},\;	\ket{\overline{\psi}} &= X(\bs{s}H_X)\ket{\overline{\psi}}\nonumber\\
	\Rightarrow
	\sum_{\bs{\ell}\in\mathbb{Z}^n}\alpha_\bs{\ell}\ket{\bs{\ell}} &= \sum_{\bs{\ell}\in\mathbb{Z}^n}\alpha_\bs{\ell}\ket{\bs{\ell}+\bs{s}H_X}\nonumber\\
	\Rightarrow \forall \bs{\ell},\forall \bs{s},\; \alpha_\bs{\ell}&=\alpha_{\bs{\ell}-\bs{s}H_X}.
	\end{align}
	This means that in a code state, all states of angular momentum differing by an element in the (left) image of $H_X$, $\im\left(H_X\right)$, must have the same amplitude.
	 Since $\im\left(H_X\right)\subseteq \ker\left(H_Z^T\right)$, we can split the latter into cosets with respect to this subgroup.
	 Logical operators of $X$-type should then move between these cosets, i.e. they are elements of $\ker\left(H_Z^T\right)$ but they are not in $\im\left(H_X\right)$. We therefore obtain the group of $X$-type logical operators  $\mathcal{L}_X$ as the first homology group of the chain complex. \\

	
  To characterize $Z$-type operators we go to some cochain complex which is dual to the initial one. 
  Concretely for a general chain complex, 
  \begin{equation}\mathcal{C}:\;\ldots\xrightarrow{\partial_{j+1}} C_j\xrightarrow{\partial_j}C_{j-1}\xrightarrow{\partial_{j-1}}\ldots,
  \end{equation}
  we define the dual complex with dual spaces as
	\begin{equation}
	C_j^* = \Hom(C_j, \mathbb{T}).\label{eq:cochains}
	\end{equation}
	Here $\Hom(A,B)$ denotes the (continuous) group homomorphisms from $A$ to $B$, i.e \begin{equation}
     \Hom(A,B) = \Big\{f:A\rightarrow B\;\Big\vert\; \forall (a_1,a_2)\in A^2,\,f(a_1+_Aa_2) = f(a_1) +_B f(a_2)\Big\},\label{eq:Hom}
 \end{equation}
where $+_A, +_B$ denote the group operation on $A$ and $B$, respectively.
    The choice of \(\mathbb{T}\) in Eq.~\eqref{eq:cochains} is specific to our setting but other choices can be made\footnote{In the literature of integer chain complexes it is more frequent to take the dual with respect to integer coefficients, that is to say, to consider \(\Hom(C_j,\mathbb{Z})\) instead of \(\Hom(C_j, \mathbb{T})\).
    This could be interpreted in our case as interchanging \(X\) and \(Z\), but this exchange is not trivial for rotors as there is no corresponding unitary operation.
    Furthermore, in the case of qubits we have \(\Hom(\mathbb{Z}_2,\mathbb{T}) \simeq \Hom(\mathbb{Z}_2,\mathbb{Z}_2) \simeq \mathbb{Z}_2\).
    Similarly for oscillators we have \(\Hom(\mathbb{R},\mathbb{T}) \simeq \Hom(\mathbb{R},\mathbb{R}) \simeq \mathbb{R}\).
    Hence the correct way to take the dual is obscured as all the groups are the same.}.
    With our choice, the elements of the dual space are the characters of the group.
	The dual boundary map is given by
	\begin{equation}
	\begin{array}{c c c c}
	\partial_j^* :& C_{j-1}^* &\longrightarrow& C_j^*\\
				& f &\mapsto& f\circ\partial_j
	\end{array}.\label{eq:coboundarymap}
	\end{equation}
	One can check that ${\rm Hom}(\mathbb{Z}^n,\mathbb{T}) \simeq \mathbb{T}^n$ and that ${\rm Hom}(\mathbb{T}^n,\mathbb{T}) \simeq \mathbb{Z}^n$ .\\
    
	Specializing to our specific length-3 chain complex in Eq.~\eqref{eq:rotorchaincomplex}, we can also check, using Eq.~\eqref{eq:coboundarymap}, that $\partial_2^*$ and $\partial_1^*$ are given as matrices by taking the transpose of $H_X$ and $H_Z$.
    The cochain complex is therefore structured as follows
		\begin{equation}
	\begin{matrix}
	\mathcal{C}^*: & C_2^* & \xleftarrow[H_X^T]{~~\partial_2^*~~} & C_1^* & \xleftarrow[H_Z]{~~\partial_1^*~~} & C_0^*\\
	& \shortparallel && \shortparallel & &\shortparallel\\
	&\mathbb{T}^{r_x} & & \mathbb{T}^n & & \mathbb{T}^{r_z}\\
	& \shortparallel && \shortparallel & &\shortparallel\\
	& \text{syndrome} && \text{operators} && \text{stabilizers}
	\end{matrix}.\label{eq:rotorcochaincomplex}
	\end{equation}
    Note also that the role of stablizer and syndrome are exchanged by the dualization.\\
	
    The logical operators expressed as homology and cohomology representatives behave in the expected way to form a quantum system in the logical subspace.
    More formally we have the following theorem.
   
	\begin{thm}\label{thm:cohom}
	Let $\mathcal{C}$ be a chain complex of free Abelian groups.
	Then its $n$th level cohomology group over $\mathbb{T}$ coefficients is isomorphic to the character group of its $n$th level homology group, that is to say
 \begin{equation}H^n(\mathcal{C},\mathbb{T}) \simeq \Hom(H_n(\mathcal{C}), \mathbb{T}).
	\end{equation}
	\end{thm}
    A proof can be found for instance in \cite[Chap. VIII]{hurewiczDimensionTheoryPMS41941} or in more general terms in \cite{brownPontrjaginDualityGeneralized1976}.
  	The cohomology group is given in the usual way
	\begin{equation}
	    H^1(\mathcal{C}, \mathbb{T}) = \ker{\partial_2^*}/\im \partial_1^* =  \ker\left(H_X^T\right)/\im\left(H_Z\right) = \mathcal{L}_Z.
	\end{equation}
	Note that the applications $\partial_2^*$ and $\partial_1^*$ are specified using integer matrices but represent maps from phases to phases.
    In particular computing \(\ker\) is slightly different from standard linear algebra with real matrices as integer multiples of \(2\pi\) are equivalent to zero.

	This description of the logical operators is consistent with inspecting the action of a $Z$-type generalized Pauli operator on code states given in Eq.~\eqref{eq:codestatebasis}.
	In order to apply the same phase for every element of a given coset, the vector needs to be in ${\rm \rm ker}(H_X^T)$ and two vectors differing by an element of ${\rm im}(H_Z)$ will have the same action.

	 \subsection{Encoded Information}
    \label{sec:codestates}
	
Now the question is: what do rotor codes encode? 
	For a homological or CSS code on qubits, the code will encode some number $k$ of logical qubits and this is captured by the (co)homology groups for the chain complex $\mathcal{C}$ over $\mathbb{Z}_2$. As argued, this is the same for our codes, namely the (co)homology of the chain complex determines the logical information. However, the nature of the (co)homology groups depends both on the chain complex $\mathcal{C}$ and the \emph{coefficient group} that is used, i.e. $\mathbb{Z}_p$, $\mathbb{Z}$ or $\mathbb{R}$, see Table \ref{tab:co-hom-general}. 

Before we spell this out in mathematical detail, let us consider a simple example which captures the main idea. We have 4 rotors and the check matrices
	\begin{align}
		H_X = \begin{pmatrix}
			+1  & -1  & 0 & 0  \\
			0  & 0 & -1 & +1  \\
			-1& -1 & +1 & +1
		\end{pmatrix}, \;
		H_Z = \begin{pmatrix}
			+1 & +1 & +1 & +1\\
            -1 & -1 & -1 & -1\end{pmatrix}.
			\label{eq:smallcode}
		\end{align}
If this were a qubit code, then $\pm 1$ entries in the check matrices would be equivalent and hence the checks would be $XXII, IIXX, ZZZZ$ which are the checks of the smallest error-detection surface code encoding 1 logical qubit. If this were an oscillator code, then $\pm 1$s matter, and there is only 1 vector which is orthogonal to the rows of $H_X$ (in ${\rm ker}(H_X^T)$), but this vector is a row of $H_Z$ (in ${\rm im}(H_Z)$), hence no logical information is encoded. The rotor code version occupies a place in between: we do not seek a vector which is orthogonal to the rows of $H_X$. Rather, we note that the operator \(\overline{Z} = Z(0,0,\pi,\pi)=\e^{i\pi(\hat\ell_3+\hat\ell_4)}\) is not generated by any product of $Z(\varphi {\bf h}^Z_j)$ for any $\varphi$, but $\overline{Z}$ still commutes with all $X({\bf h}^X_j)$, since the commutator can be $e^{2\pi k i}\mathbb{1}=\mathbb{1}$ for any integer $k$.
This is precisely the niche that rotor codes encoding qudits explore and for codes based on tessellations of manifolds, this opportunity is captured by the manifold having torsion.
  We will come back to this example and similar ones in Section \ref{sec:rpqubit}.

 \begin{table}[htb]
 \caption{Homology and cohomology groups for a chain complex $\mathcal{C}$ over coefficient groups $G=\mathbb{Z}$, $\mathbb{T}$, $\mathbb{Z}_p$ for prime $p$ and \(\mathbb{R}\).
 All entries are expressed in terms of the free and torsion part of the homology groups over $\mathbb{Z}$ (higlighted in gray). The $n$th homology group $H_n(\mathcal{C},\mathbb{Z})$ is separated in its free part $F_n$ (some number of copies of $\mathbb{Z}$) and its torsion part $T_n$ (equal to some number of cyclic groups $\mathbb{Z}_d$ for different $d$). The $(\cdot)^*$ means taking the dual, see Eq.~\eqref{eq:cochains}. The notation $\mathbb{Z}_p(G)$ designates the group containing a $\mathbb{Z}_p$ summand for every $\mathbb{Z}$ or $\mathbb{Z}_{p^k}$ summand in $G$, see \cite[Corollary 3A.6]{hatcherAlgebraicTopology2002}.
 The notation \(\mathbb{R}(G)\) designates replacing each \(\mathbb{Z}\) by a \(\mathbb{R}\).}\label{tab:co-hom-general}%
        
      \renewcommand{\arraystretch}{1.6}
		\begin{tabular}{c | c c c }
			 & \multicolumn{3}{c}{\highlight{green}{\text{Homology}}} \\
			& \multicolumn{3}{c}{$C_2 \xrightarrow{~~\partial_2~~} C_1 \xrightarrow{~~\partial_1~~} C_0$}\\\hline
			 $G$ & $H_2$ & $H_1$ & $H_0$ \\\hline
			$\mathbb{Z}$ & $\highlight{gray}{F_2\oplus T_2}$ & $\highlight{gray}{F_1\oplus T_1}$ & $\highlight{gray}{F_0\oplus T_0}$ \\
			$\mathbb{T}$ &  $\left(F_2\oplus T_1\right)^*$ & $\left(F_1\oplus T_0\right)^*$ & $\left(F_0\oplus T_{-1}\right)^*$ \\\hline
			$\mathbb{Z}_p$ & $\mathbb{Z}_p\left(F_2\oplus T_2 \oplus T_1\right)$ & $\mathbb{Z}_p\left(F_1\oplus T_1 \oplus T_0\right)$ & $\mathbb{Z}_p\left(F_0\oplus T_0 \oplus T_{-1}\right)$ \\
			$\mathbb{R}$ & $\mathbb{R}\left(F_2\right)$ & $\mathbb{R}\left(F_1\right)$ & $\mathbb{R}\left(F_0\right)$ \\\hline
   
   			 & \multicolumn{3}{c}{\highlight{cyan}{\text{Cohomology}}}\\
			& \multicolumn{3}{c}{$C_2^* \xleftarrow{~~\partial_2^*~~} C_1^* \xleftarrow{~~\partial_1^*~~} C_0^*$}\\\hline
			 $G$ & $H^2$ & $H^1$ & $H^0$\\\hline
			$\mathbb{Z}$ & $F_2\oplus T_1$ & $F_1\oplus T_0$ & $F_0\oplus T_{-1}$\\
			$\mathbb{T}$ & $\left(F_2\oplus T_2\right)^*$ & $\left(F_1\oplus T_1\right)^*$ & $\left(F_0\oplus T_{0}\right)^*$\\\hline
			$\mathbb{Z}_p$ & $\mathbb{Z}_p\left(F_2\oplus T_2 \oplus T_1\right)$ & $\mathbb{Z}_p\left(F_1\oplus T_1 \oplus T_0\right)$ & $\mathbb{Z}_p\left(F_0\oplus T_0 \oplus T_{-1}\right)$\\
			$\mathbb{R}$ & $\mathbb{R}\left(F_2\right)$ & $\mathbb{R}\left(F_1\right)$ & $\mathbb{R}\left(F_0\right)$\\
		\end{tabular}
	\end{table}
  
	In all generality the homology group over $\mathbb{Z}$ is decomposed into a so-called free part, $F$, and a so-called {\em torsion} part, $T$, with
 \begin{align}
     H_1(\mathcal{C},\mathbb{Z}) &= F\oplus T,\label{eq:decompLX}
     \\
     F &\simeq \mathbb{Z}^{k^\prime},\\
     T &\simeq \mathbb{Z}_{d_1}\oplus\cdots\oplus\mathbb{Z}_{d_{k^{\prime\prime}}},
     \end{align}
     for some integers \(k^\prime\)\footnote{\(k^\prime\) is often called the Betti number}, \(k^{\prime\prime}\) and \emph{torsion orders} \(d_1\),\ldots, \(d_{k^{\prime\prime}}\).  The homology group over $\mathbb{Z}$ and its generators can be obtained using the Smith normal form of $H_X$ and $H_Z^{T}$ \cite{munkresElementsAlgebraicTopology2019a, comphomology}. In Ref.~\cite{githublink} we provide software to obtain the generators of $ H_1(\mathcal{C},\mathbb{Z})$ as well as its decomposition into free and torsion parts, for our examples. The code is based on the open source software ``Sage". We accompanied the code with detailed explanations on how to use the Smith normal form to obtain the generators of $ H_1(\mathcal{C},\mathbb{Z})$.
  
    When studying the homology of chain complexes over a field, such as \(\mathbb{F}_p\) for prime \(p\) or \(\mathbb{R}\), there can be no torsion.
    The torsion comes from the fact that \(\mathbb{Z}\) is not a field (but a ring) and as such does not have multiplicative inverses.
    Therefore it can happen that there exists some element, \(\bs{w}\in C_1\) which is a boundary only when multiplied by some 
    integer \(d\), that is to say
    \begin{equation}
        \exists \bs{s}\in C_2,\; \partial_2(\bs{s}) = d\bs{w},\; \bs{w}\not\in \im(\partial_2).\label{eq:weakboundary}
    \end{equation}
    Such elements, \(\bs{w}\), are called \emph{weak boundaries} in \cite{munkresElementsAlgebraicTopology2019a}.
    They are homologically non-trivial but become trivial when multiplied by \(d\) and so are elements of order \(d\) in the homology group.
    
	The homology group, \(H_1(\mathcal{C},\mathbb{Z})\), as in Eq.~\eqref{eq:decompLX} corresponds to a code which encodes $k^\prime$ logical rotors and \(k^{\prime\prime}\) qudits of respective dimensions $d_1,\ldots,d_{k^{\prime\prime}}$.
	We denote as $k$ the number of independent logical systems which are encoded so that
	\begin{equation}
	    k = k^\prime + k^{\prime\prime}.
	\end{equation}

	We can structure the generating set of $\mathcal{L}_X$ according to this decomposition. Namely, there are $k^\prime$ integer vectors generating the free part, which correspond to the logical $X$ operator of the $k^\prime$ logical rotors.
	We can stack them into a $k^\prime\times n$ integer matrix $L_X^r$.
	There are also $k^{\prime\prime}$ integer vectors generating the torsion part which form the rows of a matrix $L_X^d$, which correspond to the $X$ logical operators of the logical qudits.
	We denote all these vectors as $\bs{l}^{X}_{i}$, and we set them as rows in a $k\times n$ integer matrix $L_X$: 
	\begin{equation}
	    L_X = \begin{pmatrix}L_X^{r}\\[1em]L_X^{d}\end{pmatrix}.
	\end{equation}
	A generating set of $X$-type logical operators can be then written as follows
	\begin{equation}
	 \bs{m}\in \mathbb{Z}^{k^\prime}\oplus\left(\mathbb{Z}_{d_1}\oplus\cdots\oplus\mathbb{Z}_{d_{k^{\prime\prime}}}\right)\equiv G_{\rm logical}^X   
    ,\quad \overline{X}(\bs{m}) = X(\bs{m}L_X).
	\end{equation}
	Generating all $X$-type logical operators is done as follows
	\begin{equation}
	     \bs{m}\in G_{\rm logical}^X,\quad\forall\bs{s}\in\mathbb{Z}^{r_x},\quad \overline{X}(\bs{m}) = X(\bs{m}L_X + \bs{s}H_X).
      \label{eq:integerchoices}
	\end{equation}
    Note that we do not strictly need to restrict the torsion part of \(\bs{m}\) to be in \(\mathbb{Z}_{d_j}\) and can pick \(\bs{m}\in\mathbb{Z}^k\).
    Indeed picking \(\bs{m}\) in \(\mathbb{Z}^k\) will also generate all logical operators but some will be stabilizer equivalent through Eq.~\eqref{eq:weakboundary}.
		Coming back to the code states, this characterization allows us to write down a basis of code states as equal superpositions over cosets and label these (unphysical) basis states by $\bs{m}$, namely
	\begin{align}
	    \ket{\overline{\bs{m}}} &\propto \sum_{\bs{g}\in\mathbb{Z}^{r_x}} \ket{\bs{\ell}=\bs{g}H_X+\bs{m}L_X},\label{eq:codestatebasis}
	\end{align}
 Hence we can write the code states as
	\begin{align}
		\ket{\overline{\psi}} &= \sum_{\bs{m}\in\mathbb{Z}^k}\alpha_{\bs{m}}\ket{\overline{\bs{m}}}.\label{eq:codestates}
	\end{align}

	Theorem~\ref{thm:cohom} guarantees that $H^1(\mathcal{C}, \mathbb{T})$ has the same structure as the one given for $H_1(\mathcal{C}, \mathbb{Z})$ in Eq.~\eqref{eq:decompLX}, that is
	\begin{equation}
	    H^1(\mathcal{C},\mathbb{T}) \simeq \mathbb{T}^{k^\prime}\oplus\left(\mathbb{Z}^*_{d_1}\oplus\cdots\oplus\mathbb{Z}^*_{d_{k^{\prime\prime}}}\right),\label{eq:decompLZ}
	\end{equation}
	with $\mathbb{Z}^*_d$ in Eq.~\eqref{eq:Zstar}.
	Representatives of the logical operators are generated using a $k\times n$ integer matrix $L_Z$:
	\begin{equation}
	    \bs{\phi}\in \mathbb{T}^{k^\prime}\oplus\left(\mathbb{Z}^*_{d_1}\oplus\cdots\oplus\mathbb{Z}^*_{d_{k^{\prime\prime}}}\right)\equiv G_{\rm logical}^Z,\quad \overline{Z}(\bs{\phi}) = Z(\bs{\phi}L_Z),
     \label{eq:phasechoices}
	\end{equation}
Generating all $Z$-type logical operators is done as follows
	\begin{equation}
	     \bs{\phi}\in G_{\rm logical}^Z,\;\forall \bs{\nu}\in\mathbb{T}^{r_z},\; \overline{Z}(\bs{\phi}) = Z(\bs{\phi}L_Z + \bs{\nu}H_Z).
	\end{equation}
	Note that the vector $\bs{\phi}$ contains unconstrained phases (for the logical rotors) as well as constrained ones (for the encoded qudits).
	Contrarily to the logical \(X\) operators, we cannot here relax the restriction on the phases since doing so would produce operators not commuting with the \(X\) stabilizers.
    We can again split the logical generators in two
	\begin{equation}
	\bs{\phi} = \left(\bs{\phi}^r, \;\; \bs{\phi}^d \right).
	\end{equation}
	
	We can also split the matrix $L_Z$ in two parts: a $k^\prime\times n$ integer matrix $L_Z^r$ generating the rotor part and a $k^{\prime\prime}\times n$ integer matrix $L_Z^d$ generating the qudit part.
	\begin{equation}
	    L_Z = \begin{pmatrix}L_Z^r\\[1em]L_Z^d\end{pmatrix}.
	\end{equation}
	
	The rows of the matrix $L_Z$ with integer entries can be denoted as $\bs{l}_i^Z$.
	Theorem~\ref{thm:cohom} also guarantees us that we can find a pairing of the $X$-type and $Z$-type logical operators such that 
	\begin{equation}
	    L_XL_Z^T = \id.
     \end{equation}

	\subsection{Formal Noise Model and Distance of the Code}
 \label{sec:distance}
	
	We want to characterize the level of protection offered by a rotor code which is usually captured by the distance of the code. The motivation for the definition of a distance is given by the noise model. For a qubit stabilizer code, the distance is the minimal number of physical qubits on which to act to realize any logical operation. For a qudit (or classical dit) code, it is more ambiguous: it could be the minimal number of qudits to act on to realize any logical operation, or some other measure of the minimal total change of any code state to another code state.

	We first discuss a reasonable noise model to which we tie our definition of distance for rotor codes.
	
	Since the group of generalized Pauli operators forms an operator basis, we consider a simple noise model consisting of Pauli noise. More precisely, we assume that each rotor independently undergoes $X$- and $Z$-type errors, of the form $X(m)$ and $Z(\phi)$ respectively, with probabilities of the following generic form
	\begin{align}
	\forall m\in\mathbb{Z},\;\mathbb{P}\left (X(m)\right ) &= A_X\exp\left (-\beta_XV_X(m)\right ),
	\label{eq:errX}\\
	\forall\phi\in\mathbb{T},\;\mathbb{P}\left (Z(\phi)\right ) &= A_Z\exp\left (-\beta_ZV_Z(\phi)\right ).
	\label{eq:errZ}
	\end{align}
	The parameters $A_Z$ and $A_X$ are some normalizing positive constants and $\beta_Z$ and $\beta_X$ are some real and positive strength parameters. We assume that the potential function $V_Z(\phi)$ is unchanged under $\phi \rightarrow 2\pi-\phi$ and monotonically increasing from $V_{Z}(\phi=0)=0$ to $\phi=\pi$. Similarly, the potential function $V_X(m)$ is monotonically increasing with $\vert m\vert$ and $V_X(m=0)=0$.
	One straightforward choice is
	\begin{align}
	V_Z(\phi) &= \sin^2\left(\frac{\phi}{2}\right),& \beta_Z &= \frac{1}{\sigma^{2}},\label{eq:specificerrZ}\\
	V_X(m) &= \vert m\vert,& \beta_X &= -\log p.\label{eq:specificerrX}
	\end{align}
	This choice makes $\mathbb{P}(Z(\phi))$ a normalized von Mises probability distribution characterized by standard deviation $\sigma$, and $p$ is the probability of a ladder jump $\ket{\ell}\rightarrow \ket{\ell\pm 1}$. The normalization constants are given by
	\begin{equation}
	A_Z^{-1} = \exp\left (-\frac{1}{2\sigma^2}\right )2\pi {\rm I}_0\left (\frac{1}{2\sigma^2}\right ),\qquad A_X^{-1} = \frac{1+p}{1-p},
	\end{equation}
	where ${\rm I}_n(x)=\frac{1}{\pi}\int_0^{\pi} d\theta\, e^{x\cos(\theta)} \cos(n\theta)$ is the modified Bessel function of the first kind of order $n$. 
 
The von Mises probability distribution is a natural choice since, unlike Gaussian noise, it respects the periodicity of the $\phi$-variable, and for small $\sigma$ and thus, small values for $\phi$, it can be approximated by a Gaussian distribution.   
 

    For $n$ quantum rotors we consider independent and identically distributed noise given by probability distributions
    \begin{align}
        \forall \bs{m}\in\mathbb{Z}^n,\;\mathbb{P}(X(\bs{m})) &= A_X^n\exp\left(-\beta_XW_X(\bs{m})\right), \\
        \forall\bs{\phi}\in\mathbb{T}^n,\;\mathbb{P}(Z(\bs{\phi})) &= A_Z^n\exp\left(-\beta_ZW_Z(\bs{\phi})\right),
    \end{align}
    where  
    \begin{align}
    W_X(\bs{m}) &= \sum_{j=1}^nV_X(m_j),\label{eq:weightX} \\
        W_Z(\bs{\phi}) &= \sum_{j=1}^nV_Z(\phi_j),\label{eq:weightZ}
    \end{align}
	
	Using these weights, we can introduce $X$ and $Z$ distances for the code. 
 We define the $X$ distance from the definition of the weight in Eq.~\eqref{eq:weightX} and minimize over stabilizer equivalent logical operators:
	\begin{align}
	    	    d_X &= \min_{\bs{m}\in G_{\rm logical}^X,\bs{m}\neq \bs{0}}\, \min_{\bs{s}\in\mathbb{Z}^{r_x}} W_X\left(\bs{m}L_X + \bs{s}H_X\right),
	    \label{eq:distanceX}
	\end{align}
	where $G_{\rm logical}^X$ was defined in Eq.~\eqref{eq:integerchoices}.
	
 As mentioned earlier, the distance in $Z$ is less straightforward when the code encodes some logical rotors. 
	For this reason, if one encodes logical rotors (besides qudits), we can compare the distance of a logical $Z$ operator to the weight of its bare implementation and we denote this distance slightly differently, namely as
	\begin{equation}
	    \delta_Z = \min_{\bs{\phi}\in G_{\rm logical}^Z,\bs{\phi} \neq \bs{0}}\,\min_{\bs{\nu}\in\mathbb{T}^{r_z}}\frac{W_Z\left(\bs{\phi}L_Z+\bs{\nu}H_Z\right)}{W_Z(\bs{\phi})},\label{eq:distanceZ}
	\end{equation}
where $G_{\rm logical}^Z$ was defined in Eq.~\eqref{eq:phasechoices} and $W_Z(\bs{\phi})=\sum_{j=1}^k V_Z(\phi_j)$, the sum of weights on the $k$ unencoded degrees of freedom. For example, if the rotor code encodes a single logical qudit of dimension $d$ we have logical shifts $\phi \in \mathbb{Z}_d$, which implies that $\frac{1}{W_Z(\phi)}$ is always bounded away from 0 by a constant $C$. The constant is noise-model dependent but irrelevant for the distance scaling with the number of physical qubits and hence we can omit the denominator and obtain a similar definition as $d_X$. When the rotor code encodes a logical rotor with logical shifts $\phi\in \mathbb{T}$, the denominator goes to zero when $\phi \rightarrow 0$, but so does the numerator, hence, it is the ratio that matters. In this paper we focus on codes encoding qubits so $\delta_Z$ is like $d_X$. For the codes that we consider, a difference between the $X$ distance and the $Z$ distance is that a lower probability logical $Z$ can be obtained by spreading the support of the logical operator by stabilizers; we discuss this in Section \ref{sec:spread}.
 
 \subsubsection{\(X\) Distance Bound from Qudit Versions of the Code}
	We can compare the $X$ distance of the rotor code which encodes $k$ logical rotors to the distance of the corresponding qudit code for any qudit dimension $l$. Let us denote as $\mathcal{C}^l(H_X, H_Z)$ the quantum code obtained by replacing each physical rotor by a qudit of dimension $l$ defined using generalized Pauli operators for qudits, see the definition in Appendix \ref{app:quditcode}.
 We denote the $X$ distance of the qudit code as $d^l_X$ defined using the Hamming weight of logical operators.

	\begin{thm}[Lower bound on $d_X$]
 \label{thm:boundXdist}
	Given a rotor code $\mathcal{C}^{\rm rot}(H_X,H_Z)$ encoding $k$ degrees of freedom.
	Let $M_{\rm min}$ be the set of non-trivial logical operators with a representative of minimal weight:
	\begin{equation}
	    M_{\rm min} = \left\{\bs{m}\in G^X_{\rm logical}, \bs{m}\neq \bs{0}\middle\vert \min_{\bs{s}\in\mathbb{Z}^{r_x}}W_X(\bs{m}L_X+\bs{s}H_X)=d_X\right\}.
	\end{equation}
	Let $L$ be the set of qudit dimensions, $l=2,3,\ldots$, such that there exists a logical operator of minimal weight which is non-trivial in $\mathcal{C}^l(H_X,H_Z)$:
	\begin{equation}
	    L = \left\{l\in\mathbb{N}^{\geq2}\middle\vert \exists \bs{m}\in M_{\rm min},\,\forall\bs{s}\in\mathbb{Z}^{r_x},\,\bs{m}L_X\neq \bs{s}H_X\pmod l\right\}.
	\end{equation}
	Then $d_X$ is lower bounded as follows:
	\begin{equation}
	    d_X\geq \max_{l\in L}d_X^l.
	\end{equation}
	\end{thm}
	\begin{proof}
	    
		Pick $\bs{m}_{\rm min}$ and $\bs{s}_{\rm min}$ forming a minimum weight logical operator and consider the obtained vector 
	\begin{equation}
	    \bs{v}_{\rm min} = \bs{m}_{\rm min}L_X + \bs{s}_{\rm min} H_X,\qquad W_X(\bs{v}_{\rm min}) = d_X.
	\end{equation}
	Taking $\bs{v}_{\rm min}\pmod l$ yields a valid logical operator for the corresponding qudit code.
	If $\bs{v}_{\rm min}\pmod l$ corresponds to a non-trivial logical operator then its weight is lower bounded by the distance of the qudit code, meaning that
	\begin{equation}
	d_X = W_X(\bs{v}_{\rm min}) \geq W_X(\bs{v}_{\rm min}\pmod l)  \geq W_H(\bs{v}_{\rm min}\pmod l)  \geq d_X^l.
	\end{equation}
 where $W_H()$ is the Hamming weight of a vector $\bs{x}$, counting the non-zero entries $W_H(\bs{x}) = \left\vert\left\{x_j\vert x_j\neq 0\right\}\right\vert$.
	This means that we can also take the maximum of the $d^l_X$ over $l$ for which there is a minimum weight logical operator in the rotor code which is non-trivial in the qudit code.
	\end{proof}
	Note that if there is no torsion, i.e. there are only rotors in the codespace then $L=\mathbb{N}^{\geq 2}$, see Appendix~\ref{app:quditcode}.
	
	\subsubsection{$Z$ Distance and Logical Operator Spreading}
	\label{sec:spread}
	
	It is not clear that the distance measures that we have introduced are achieved by seeking logical operators which have minimal support, as we are used to in the qubit case. This is particularly true for the logical $Z$ operators as we can spread such operator around by adding $\bs{\nu}H_Z$ for a {\em continuous} $\bs{\nu} \in \mathbb{T}^{r_z}$.
	
	This can have non-trivial consequences for the stability and protection of the encoded information, also for the rotor codes which encode a logical qubit. The question whether spread-out logical operators have lower probability than minimal-support logical operators depends on the error model through the weight function $W_Z$. This function has a quadratic dependence on $\phi$ in $V_Z(\phi)$ for small $\phi$ as in Eq.~\eqref{eq:specificerrZ}, thus making spread-out logical $Z$ operators consisting of many small shifts more likely than logical $Z$ operators with minimized support. 
 
 To bound this phenomenon, it turns out that a measure of \emph{disjointness} of the logical $\overline{X}$ operator\footnote{Interestingly, disjointness of logical operators was studied previously \cite{jochym-oconnorDisjointnessStabilizerCodes2018} to understand which logical gates can be realized by Clifford or constant depth operations.} plays a role in minimizing the $Z$ distance as captured by the following Lemma. In this Lemma we focus on codes encoding a single logical degree of freedom, a rotor or a qudit, for simplicity. In addition, we impose some constraints on the support of the logical $\overline{X}$ which is obeyed by all the homological codes for which we wish to apply the Lemma in this paper. 
 
 \begin{lem}
 Let $\mathcal{C}^{\rm rot}(H_X,H_Z)$ be a rotor code encoding one logical degree of freedom ($k=1$).
 Suppose one can find a set \(\Delta_X\subset \mathbb{Z}^n\) of \(N_X\) representatives of the logical operator \(\overline{X}(1)\).
 Suppose furthermore that any of these representatives \(\bs{m}\in\Delta_X\) has only entries in \(\{-1,0,1\}\) and that they all have non-overlapping support pair-wise.
 Define \(D_X\) as the maximum weight among the elements of \(\Delta_X\), i.e $D_X=\max_{\bs{m}\in \Delta_X}|\bs{m}|$.
 Then for sufficiently large distance $d_X$, as defined by Eq.~\eqref{eq:distanceX}, one can lowerbound the weight of any logical $\overline{Z}(\alpha)$, with $\alpha \in G_{\rm logical}^Z$, 
as
  \begin{equation}
        \delta_Z \geq \frac{N_XD_X\sin^2\left(\frac{\alpha}{2D_X}\right)}{\sin^2(\frac{\alpha}{2})}.\label{eq:lowerboundZ}
        \end{equation}   
         \label{lem:Zbound}
 \end{lem}
        When $\mathcal{C}^{\rm rot}(H_X,H_Z)$ encodes a qubit, we have only one logical $\overline{Z}$ with $\alpha=\pi$ and hence Eq.~\eqref{eq:lowerboundZ} becomes
        \begin{equation}
            \delta_Z\geq N_X D_X \sin^2\left(\frac{\pi}{2D_X}\right) \sim \frac{N_X \pi^2}{4 D_X}.
            \label{eq:Zbound-qubit}
        \end{equation}
        When $\mathcal{C}^{\rm rot}(H_X,H_Z)$ encodes a rotor we have $\alpha\in \mathbb{T}$ $(\alpha \neq 0)$ and hence Eq.~\eqref{eq:lowerboundZ} at $\alpha \rightarrow 0$ becomes
        \begin{equation}
            \delta_Z\geq \frac{N_X}{D_X}. 
            \label{eq:Zbound-rotor}
         \end{equation}

\begin{proof}
Let $\overline{Z}(\alpha)$ be realized by some vector of phases, \(\bs{\phi}\in\mathbb{T}^n\), i.e. $\overline{Z}(\alpha)=Z(\bs{\phi})$.
    For all representatives \(\bs{m}\in\Delta_X\) we require
    \begin{align}
    X(\bs{m})Z(\bs{\phi}) = \e^{i\bs{m}\cdot \bs{\phi}}Z(\bs{\phi})X(\bs{m})    \Rightarrow \exists k\in\mathbb{Z},\;\bs{m}\cdot \bs{\phi}&=\alpha+2k\pi.
    \end{align}
    Note that for every index \(j\) where $\bs{m}$ has no support (\(m_j=0\)) the value of \(\phi_j\) is not constrained. Since all $\bs{m} \in \Delta_X$ are non overlapping, one can impose all constraints simultaneously. Now consider minimizing the weight $W_Z(\bs{\phi})$ of $\overline{Z}(\alpha)$ under these constraints, i.e. we have
        \begin{equation}
        \min_{\substack{\bs{\phi}\neq \bs{0}\in\mathbb{T}^n,\, k\in\mathbb{Z}\\\forall \bs{m}\in\Delta_X,\,\bs{m} \cdot \bs{\phi}=\alpha+2k\pi}} W_Z(\bs{\phi}) \geq N_X \min_{\bs{m}\in \Delta_X}\left[\min_{\substack{\bs{\phi}\neq \bs{0}\in\mathbb{T}^n,\, k\in\mathbb{Z}\\ \bs{m}\cdot\bs{\phi}=\alpha+2k\pi}} W_Z(\bs{\phi})\right].\label{eq:minimizationZ}
    \end{equation}
   In Appendix \ref{sec:minimizationZ} we show that when
    \(\bs{m}\) contains only $0,+1,-1$, the minimum in the previous equation for a fixed $\bs{m}$ is attained by {\em spreading} $\bs{\phi}$ evenly over all non-zero entries of \(\bs{m}\). That is, we prove for $|\bs{m}|=n$ (assuming a sufficiently large $n$ and hence a sufficiently large $d_X$) that
  \begin{align}
        W_Z(\bs{\phi}) \geq n \sin^2\left(\frac{\alpha}{2n}\right).
        \label{eq:app-bound}
    \end{align}
Clearly, the bound is smallest when the (sufficiently large) weight $n$ of $\bs{m}$ is maximized, hence equal to $D_X$, leading to Eq.~\eqref{eq:lowerboundZ}, using the definition of $\delta_Z$ in Eq.~\eqref{eq:distanceZ}.
\end{proof}
 
    \subsection{Notation for the Parameters of a Quantum Rotor Code}
    Now that we have defined quantum rotor codes, examined their codespace and defined a notion of distance, we can choose a notation to summarize the main parameters. 
    The main difference with the usual \(\llbracket n,k,d\rrbracket \) notation for qubit codes is that we have to fully specify the homology group and both  \(X\) and \(Z\) distances but it is otherwise very similar.
    \begin{defi}[Parameters of a quantum rotor code]
    Given a quantum rotor code \(\mathcal{C}^{\rm rot}(H_X, H_Z)\) we say it has parameters
    \begin{equation}
        \left\llbracket n, (k, d_1\cdot d_2\cdots d_{k^{\prime}}), (d_X,\delta_Z)\right\rrbracket_{\rm rot},
    \end{equation}
    if it involves \(n\) physical rotors, it encodes \(k\) logical rotors and \(k^\prime\) qudits of respective dimensions \(d_1\), \ldots, \(d_{k^\prime}\) and has \(X\) distance \(d_X\) as given by Eq.~\eqref{eq:distanceX} and \(Z\) distance \(\delta_Z\) as given by Eq.~\eqref{eq:distanceZ}. 
    \end{defi}
    When there are \(m\) qudits of same dimension \(d\) we write \(d^m\) in the sequence of qudit dimensions.
    
	\subsection{Measuring Stabilizers and Hamiltonian of the Code}
	\label{sec:meas-stab}
	
	In active stabilizer quantum error correction, we measure the stabilizer generators of the code in order to infer errors. For qubits, Pauli stabilizers are directly observable and can thus be measured. Here, for the $X$ part of the stabilizer one can construct a set of $r_x$ Hermitian observables for each generator $S_j^X$, i.e.  
	\begin{align}
		O_j^{c,X} &= \frac{S_j^X + {S_j^X}^\dagger}{2} = \cos\left(\bs{h}_j^X\cdot\bs{\hat{\theta}}\right),\qquad
		O_j^{s,X} = \frac{S_j^X - {S_j^X}^\dagger}{2i} = \sin\left(\bs{h}_j^X\cdot\bs{\hat{\theta}}\right).\label{eq:measX}
	\end{align}
	Clearly, an eigenstate of $S_j^X$ with eigenvalue $\e^{i\theta}$ is an eigenstate of $O_j^{c,X}$ with eigenvalue $\cos(\theta)$ and of $O_j^{s,X}$ with eigenvalue $\sin(\theta)$. For the $Z$ part of the stabilizer one can construct a set of $r_z$ observables, namely
	\begin{align}
			O_j^Z &= \bs{h}_j^Z\cdot\hat{\bs{\ell}},\label{eq:measZ}
			\end{align}
such that learning the integer eigenvalue, say $a$, of $O_j^Z$, fixes the eigenvalue of $S_j^Z(\varphi)$ to be $e^{i\varphi a}$ for any $\varphi$. 

    Instead of active error correction in which these observables are approximately measured, we can consider passive error correction and construct a (dimensionless) code Hamiltonian whose groundspace (and excited spaces) is the codespace.
    We thus seek a suitable function of the observables which makes the codespace have the smallest eigenvalue. A natural choice is   
    \begin{equation}
        H_{\rm code} = -\sum_{j=1}^{r_x} O_j^{c,X} + \sum_{j=1}^{r_z} \left(O_j^Z\right)^2 = -\sum_{j=1}^{r_x}\cos\left(\bs{h}_j^X\cdot\bs{\hat{\theta}}\right) + \sum_{j=1}^{r_z} \left(O_j^Z\right)^2.
        \label{eq:codehamiltonian}
    \end{equation}

One can compute the energy of an excitation of type $X(\bs{m})$ or $Z(\bs{\phi})$ on a ground state $\ket{\overline{\psi}}$:
    \begin{equation}
        E(\bs{m}) = \bs{m}^TH_Z^TH_Z\bs{m},\qquad E(\bs{\phi}) = \sum_{j=1}^{r_x}\cos\left(h_j^X\cdot\bs{\phi}\right).
    \end{equation}
One observes that the $Z(\bs{\phi})$ excitations are gapless due to the continuous nature of $\bs{\phi}$. Since actual physical states are only supported on a finite range of $\ell$, and are thus only approximate eigenstates for $S_j^X$, the spectrum in such physical subspace will not be continuous, see e.g. the discussion in Ref.~\cite{RBCD} and Section \ref{sec:protection}.

\section{Constructions of Homological Quantum Rotor Codes}
\label{sec:constructions}
In this section we give general as well as concrete ways to construct homological quantum rotor codes and investigate their parameters.
We start by using tessellations of 2D manifolds, higher dimensional ones and then products of chain complexes.

	\subsection{Homological Quantum Rotor Codes from Tessellations of Manifolds}
 	\label{sec:tes}

	A common way of getting a chain complex is to consider a manifold with a tessellation \cite{FM:plane}. Given a $D$-dimensional manifold and a tessellation of it, one chooses some $i\in\{1,\ldots,D-1\}$ and puts rotors on the $i$-cells. The $(i+1)$-cells and $(i-1)$-cells are used to define the $X$ and $Z$ stabilizers respectively as in Eq.~\eqref{eq:rotorchaincomplex}. 
    
    Note that we will always choose to put the \(X\) stabilizers on the higher dimensional \((i+1)\)-cells and the \(Z\) stabilizers on the lower dimensional \((i-1)\)-cells.
    This is opposite to the usual choice for homological quantum codes, in particular on two-dimensional manifolds. What is important is that in our case exchanging \(X\) and \(Z\) is {\em not} equivalent and the opposite choice is less interesting, in particular in 2D.\\
    
    In this section we consider tessellations of two-dimensional manifolds such as the torus, the projective plane and the M\"{o}bius strip. 
    For instance, the torus using a \(w\times N\) tessellation by square faces simply gives rise to a toric rotor code, --see also \cite{albertGeneralPhaseSpaces2017}, and \cite{VAHWPT} for the oscillator toric code--, encoding two logical rotors, as its homology is \(\mathbb{Z}^2\). We denote this code \(\mathbb{T}_2(w,N)\) and the parameters, in particular the distance, of this code can be computed to be \begin{equation}
    \mathbb{T}_2(w,N): \left\llbracket2wN, (2,0), \left(\min(N,w),\min\left(\frac{N}{w},\frac{w}{N}\right)\right)\right\rrbracket_{\rm rot}.
    \end{equation}
    The $X$ distance bound $\min(N,w)$ is simply the minimal support of a loop along the two directions. The $Z$ distance bound comes from applying Eq.~\eqref{eq:Zbound-rotor} in Lemma \ref{lem:Zbound} twice.
    Once with a set of disjoint representatives of \(\overline{X}_1\) and once of \(\overline{X}_2\).
    One set can be chosen as containing \(N_X=N\) disjoint representatives each of size \(D_X=w\).
    The other set can be chosen as containing \(N_X=w\) disjoint representatives each of size \(D_X=N\).

    \subsubsection{Real Projective Plane Encoding a Qubit}
 \label{sec:rpqubit}
 
	The real projective plane, denoted as $\mathbb{RP}^2$, is an interesting example since its first homology group has a trivial free part (i.e. not encoding any logical rotor) but a non-trivial torsion part.
	Table~\ref{tab:co-hom-projectiveplane} recalls the homology and cohomology groups of the real projective plane for different choices of coefficient groups.

	\begin{table}[ht]
 \caption{Homology and cohomology groups of the real projective plane $\mathbb{RP}^2$ over different coefficient groups \(G\).
	The $\mathbb{Z}$ and $\mathbb{T}$ rows give the possibilities for rotor codes.
	Depending on if one chooses to put $X$-type stabilizers on vertices (highlighted in red) or on faces (highlighted in violet), one gets either an empty codespace or a logical qubit. The $\mathbb{Z}_2$ row gives qubit codes \cite{FM:plane}.
	The $\mathbb{R}$ and $\mathbb{Z}_3$ rows show what would happen for oscillator and qutrit codes which would yield an empty logical codespace in both cases.}\label{tab:co-hom-projectiveplane}%
		\begin{tabular}{c| c | c c c | c c c}
			 & & \multicolumn{3}{c}{\highlight{green}{\text{Homology}}} & \multicolumn{3}{|c}{\highlight{cyan}{\text{Cohomology}}}\\
			& & \multicolumn{3}{c|}{$C_2 \xrightarrow{~~\partial_2~~} C_1 \xrightarrow{~~\partial_1~~} C_0$} & \multicolumn{3}{|c}{$C_2^* \xleftarrow{~~\partial_2^*~~} C_1^* \xleftarrow{~~\partial_1^*~~} C_0^*$}\\\hline
			\text{System} &  \(G\) & $H_2$ & $H_1$ & $H_0$ & $H^2$ & $H^1$ & $H^0$\\\hline
			\multirow{2}{*}{\text{Rotor}} & $\mathbb{Z}$ & $0$ & $\highlight{violet}{\mathbb{Z}_2}$ & $\mathbb{Z}$ & $\mathbb{Z}_2$ & $\highlight{red}{0}$ & $\mathbb{Z}$\\
			& $\mathbb{T}$ &  $\mathbb{Z}_2$ & $\highlight{red}{0}$ & $\mathbb{T}$ & $0$ & $\highlight{violet}{\mathbb{Z}_2}$ & $\mathbb{T}$\\\hline
			\text{Qubit} & $\mathbb{Z}_2$ & $\mathbb{Z}_2$ & $\mathbb{Z}_2$ & $\mathbb{Z}_2$ & $\mathbb{Z}_2$ & $\mathbb{Z}_2$ & $\mathbb{Z}_2$\\
			\text{Qutrit} & $\mathbb{Z}_3$ & $0$ & $0$ & $\mathbb{Z}_3$ & $0$ & $0$ & $\mathbb{Z}_3$\\
			\text{Oscillator} & $\mathbb{R}$ & $0$ & $0$ & $\mathbb{R}$ & $0$ & $0$ & $\mathbb{R}$
		\end{tabular}
	
	\end{table}

	To be concrete, let us construct some codes and see how the non-trivial qubit encoding comes about.
    Given a tessellation of a surface, one associates an arbitrary orientation, clockwise or anti-clockwise, with each two-dimensional face, the elements in $C_2$ in Eq.~\eqref{eq:rotorchaincomplex}. In addition, one associates an arbitrary direction with each edge, the elements in $C_1$. The properties of the code, i.e. what is encoded and what is the code distance, are not dependent on these choices. To construct a row of $H_X$ corresponding to a face in the tessellation, we then place an integer entry $m=m'-m''$ for some edge, when the edge is in the boundary ($\partial$) of the face with the same direction as the face $m'$ times, and when the edge is in the boundary ($\partial$) of the face with the opposite direction as the face $m''$ times. Similarly, to construct a row in $H_Z$, corresponding to a vertex in the tessellation, we place an integer entry $\pm 1$ for some edge, namely $+1$ when the edge is incoming to the vertex, and $-1$ when the edge is outgoing to the vertex (when two vertices with an edge between them are identified in the tessellation, one places both a $+1$ and $-1$ so $0$ in total).

In Fig.~\ref{fig:projectiveplane} we give three tessellations of increasing size of the real projective plane which encode a logical qubit. 
	
	\begin{figure}[ht]
		\centering
		\subfloat[\label{sfig:PP1}]{\includegraphics[width=.25\linewidth]{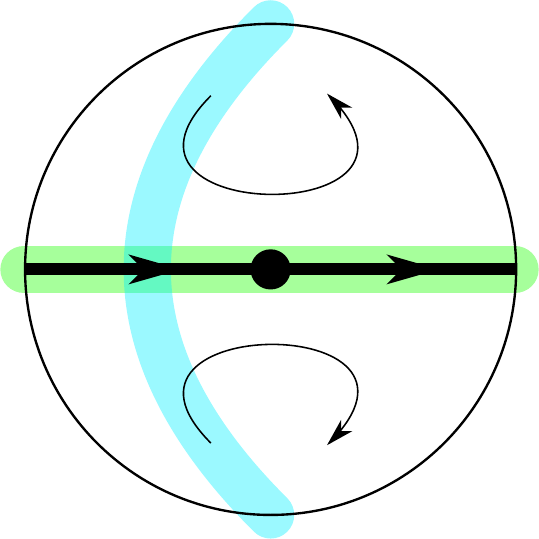}}
		\hfil
		\subfloat[\label{sfig:PP4}]{\includegraphics[width=.25\linewidth]{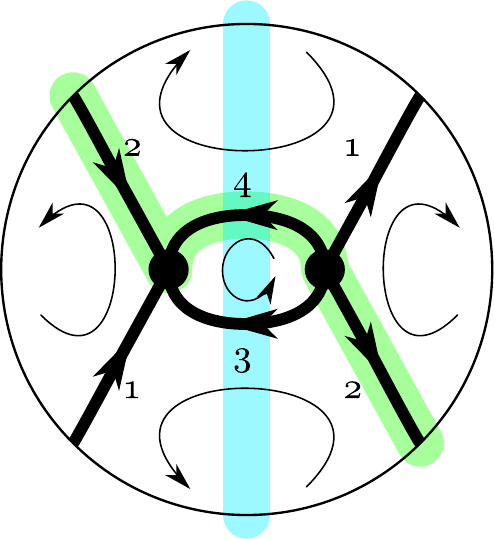}}\hfil
		\subfloat[\label{sfig:PP9}]{\includegraphics[width=.25\linewidth]{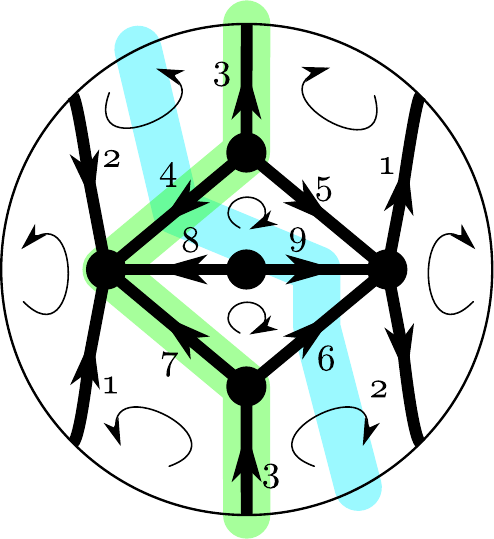}}
		\caption{In \protect\subref{sfig:PP1}, \protect\subref{sfig:PP4} and \protect\subref{sfig:PP9}, points on the outer circle are identified in antipodal pairs.
		\protect\subref{sfig:PP1} The smallest chain complex describing the real projective plane $\mathbb{RP}^2$ with one vertex, one edge and one face.
		The face is attached twice to the edge (from ``both sides'') seeing the same orientation twice.
		The vertex is attached twice to the edge with two opposite (canceling) orientations. 
		\protect\subref{sfig:PP4} A small tessellation of $\mathbb{RP}^2$ with 4 edges and \protect\subref{sfig:PP9} with 9 edges.
		The potential support for the two complementary logical operators are shown in green and blue.
		One can observe that twice the green path can be generated by the sum of all the faces.
		One can also observe that the blue path can host a logical operator, satisfying face constraints, only with coefficients that double to zero, hence they should be $\pi$. These two observations can help developing an intuitive understanding for the (co)homology groups in Table~\ref{tab:co-hom-projectiveplane}.
	}
		\label{fig:projectiveplane}
	\end{figure}
	
	The tessellation in Fig.~\ref{fig:projectiveplane}(a) leads to a stabilizer group generated by the operator $X(2)$, in other words, $H_X=2$ and $H_Z=0$. The logical operators of the encoded qubit are $\overline{X}=X(1)=\e^{i\hat\theta}$ and $\overline{Z}=Z(\pi)=\e^{i\pi\hat\ell}$. The code Hamiltonian, $H=-\cos(2\hat{\theta})$, obtained through Eq.~\eqref{eq:codehamiltonian}, can be viewed as that of a $0-\pi$ qubit: the state $\ket{\overline{+}}=\ket{\theta=0}$ and $\ket{\overline{-}}=\ket{\theta=\pi}$ are related by the $\pi$-phaseshift $\overline{Z}=Z(\pi)$. 
 
 For the tessellation in Figure~\ref{fig:projectiveplane}(b) we have the check matrices, previously given in Eq.~\eqref{eq:smallcode} and the logical operators are \(\overline{X} = X(0,-1,0,1)=\e^{i(\hat\theta_4 - \hat\theta_2)}\) and \(\overline{Z} = Z(0,0,\pi,\pi)=\e^{i\pi(\hat\ell_3+\hat\ell_4)}\). For this code, a spread-out logical, see Section \ref{sec:spread}, is \(\overline{Z} = Z\left(-\frac{\pi}{2},-\frac{\pi}{2},\frac{\pi}{2},\frac{\pi}{2}\right)=\e^{i\frac{\pi}{2}(\hat\ell_3+\hat\ell_4-\hat\ell_1-\hat\ell_2)}\) and a possible code Hamiltonian would be of the form \begin{equation}
		    H=-\cos(\hat{\theta}_1-\hat{\theta}_2)-\cos(\hat{\theta}_3-\hat{\theta}_4)-\cos(\hat{\theta}_1+\hat{\theta}_2-\hat{\theta}_3-\hat{\theta}_4)+(\hat{\ell}_1+\hat{\ell}_2+\hat{\ell}_3+\hat{\ell}_4)^2.
      \label{eq:examH}
		\end{equation}
  This four-rotor projective plane code, \(\mathbb{RP}^2(4)\), has parameters
  \begin{equation}
      \mathbb{RP}^2(4): \left\llbracket4, (0,2), (2,2)\right\rrbracket_{\rm rot}.
  \end{equation}
	
	For the third example code in Fig.~\ref{fig:projectiveplane}(c) we have
	\begin{align}
		H_X = \begin{pmatrix}
			1  & -1  & 0 & 0 & 0 & 0 & 0 & 0 & 0 \\
			-1  & 0  & 1 & 0 & -1 & 0 & 1 & 0 & 0 \\
			0  & 0  & 0 & -1 & 1 & 0 & 0 & 1 & -1 \\
			0  & 0  & 0 & 0 & 0 & -1 & 1 & -1 & 1 \\
			0  & 1  & 1 & -1 & 0 & 1 & 0 & 0 & 0 \\
		\end{pmatrix}, \;
		H_Z = \begin{pmatrix}
			1  & 1  & 0 & 1 & 0 & 0 & 1 & 1 & 0 \\
			0  & 0  & -1 & -1 & -1 & 0 & 0 & 0 & 0 \\
			-1  & -1  & 0 & 0 & 1 & 1 & 0 & 0 & 1 \\
			0  & 0  & 1 & 0 & 0 & -1 & -1 & 0 & 0 \\
			0  & 0  & 0 & 0 & 0 & 0 & 0 & -1 & -1 \\
			\end{pmatrix}.
			\label{eq:bigcode}
		\end{align}
  and we can define
  	\begin{align}
	\bs{l}_X=(0,0,-1,1,0,0,-1,0,0),\;\; \bs{l}_Z=(0,0,0,1,0,1,0,0,1),
    \end{align}
	and the logicals are $\overline{X}=X\left(\bs{l}_X\right) = \e^{i\bs{l}_X\cdot\hat{\bs{\ell}}}$ and $\overline{Z}=Z\left(\bs{l}_Z\right) = \e^{i\pi\bs{l}_Z\cdot\hat{\bs{\theta}}}$.

The parameters of this 9-rotor projective plane code, \(\mathbb{RP}^2(9)\), can be estimated to be
\begin{equation}
\mathbb{RP}^2(9): \left\llbracket9, (0,2), \left(3,\delta_Z\right)\right\rrbracket_{\rm rot},\quad 3\geq\delta_Z\geq\frac{9}{4}.
\end{equation}
The lower bound is obtained by considering a set of three \(X\) logical operators of weight three each.
The support can be \((3,4,7)\), \((1,8,9)\) and \((2,4,5)\) which has a single overlap and equal weights so the minimization technique of Section~\ref{sec:spread} still works.\\

Table \ref{tab:co-hom-projectiveplane} shows what happens when one puts qubits resp. oscillators on the edges for these tessellations, so that one encodes a single qubit resp. no logical information. We can also directly verify here that one cannot interchange $H_X$ and $H_Z$: this is expressed in the homology groups shown in Table \ref{tab:co-hom-projectiveplane}. In terms of freedom to define the code, we note that the $\pm$ signs in the matrices $H_X$ and $H_Z$ depend on the choice of orientation for each face, i.e. the signs in each row in $H_X$ could be flipped. In addition, the choice for orientation for each edge can be flipped which leads to a column in $H_Z$ and $H_X$ being flipped in sign.
Many other small tessellations of the real projective plane exist in the form of projective polyhedra such as the hemicube \cite{leverrierLocalTestabilityQuantum2022} and the tetrahemihexahedron \cite[Table 1]{CCBT:friends}.

 \subsubsection{M\"{o}bius Strip Encoding a Qubit}
\label{sec:thinmoeb}

	The M\"{o}bius strip is known to be the real projective plane without periodic boundaries on one side. One can thus obtain a M\"{o}bius strip by removing a face, say, the face between edges 2,3,4 and 6 in Fig.~\ref{fig:projectiveplane}(c), i.e., the last row of $H_X$ in Eq.~\eqref{eq:bigcode}, and identify the removed face with the logical operators $\overline{X}(m)=e^{im (\hat{\theta}_2-\hat{\theta}_4+\hat{\theta}_3-\hat{\theta}_6)}$ of an encoded rotor. This makes a M\"{o}bius strip with so-called `smooth' boundaries and the logical $\overline{Z}(\phi)$ can attach to such smooth boundaries.
	
	One can also make a M\"{o}bius strip with so-called `rough' boundaries.
    It can be constructed from the real projective plane by removing a vertex (while keeping the edges adjacent to it).
    On a closed manifold a single removed vertex is equivalent to the sum of all the other vertices.
    Hence, no additional logical rotor is associated with the removed vertex.
    Such M\"{o}bius-strip code will encode only a qubit due to the twist in orientation, see Fig.~\ref{fig:moebius}.

\begin{figure}[htb]
    \centering
\centering 
\includegraphics[width=.5\linewidth]{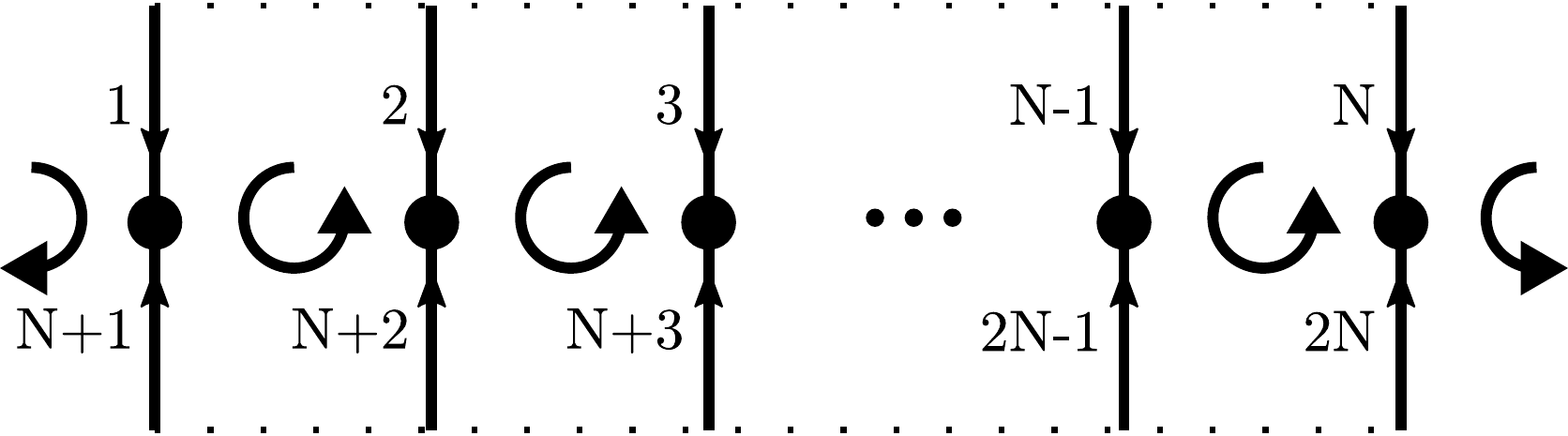}
\caption{Thin tessellation of a M\"{o}bius strip with a `rough' boundary with no rotors on the dashed edges or vertex checks on that boundary.
The edges are numbered from $1$ to $2N$.}
\label{fig:moebius}
\end{figure}

In Figure~\ref{fig:moebius} the strip is very thin: it is only one face wide. 
This rotor code is defined on a set of $n =2N$ rotors (one for each directed edge) where the rotors above each other are labeled by $i$ and $N+i$. The face stabilizers, corresponding to the rows of $H_X$ are \begin{align}
    S_j^X & =  e^{i(\hat{\theta}_j-\hat{\theta}_{N+j}+\hat{\theta}_{N+j+1}-\hat{\theta}_{j+1})},\; j=1,\ldots, N-1,\notag \\
 S_N^X &  = e^{i(\hat{\theta}_N-\hat{\theta}_{2N}-\hat{\theta}_{N+1}+\hat{\theta}_{1})},
 \end{align}
where the last face is {\em twisted}. The vertex stabilizers are 
\[
\forall \varphi,\;S_j^Z(\varphi)=e^{i \varphi (\hat{\ell}_j+\hat{\ell}_{N+j})}.
\]
Let's examine how to construct the logical $\overline{Z}$. Imagine phase shifting the upper rotors by $\epsilon$ each, i.e. we apply $e^{i \epsilon \sum_{i=1}^N \hat{\ell}_i}$. Due to the twisted check $S_N^X$, $\epsilon$ can only be $\pi$. For this thin M\"{o}bius strip, we can spread out the $\overline{Z}$ by moving half its support from the top of the ladder to the bottom. Using Eq.~\eqref{eq:distanceZ} and Eq.~\eqref{eq:specificerrZ} we see that this spread-out logical has distance $\delta_Z=2 N \sin^2(\pi/4)=N$ which is exactly the same as the minimal support logical $\overline{Z}$.

The matching logical $\overline{X}=e^{i (\hat{\theta}_j-\hat{\theta}_{N+j})}$ for any choice of $j$ and it has distance $d_X=2$ using Eq.~\eqref{eq:distanceX} and Eq.~ \eqref{eq:specificerrX}. 
So the parameters of the thin Möbius trip are given
\begin{equation}
\left\llbracket 2N, (0,2), (2,N)\right\rrbracket_{\rm rot}.   
\end{equation}
One can observe that $\overline{X}^k$ for any $k$ also commutes with the stabilizers and is not equal to $\overline{X}$ as an operator, but one can show that $\overline{X}^m$ for even $m$ is in the stabilizer group. To understand this, observe that we can move the support of $\overline{X}$ over the strip while keeping its form the same. This means that $\overline{X}^2$ can be split to different rungs on the ladder, and then we can move and annihilate them at the twist, since the face at the twist has the appropriate opposite signs.

The previous choice of logical $\overline{X}$ and $\overline{Z}$ operators gives the logical $\ket{\overline{0}}$ and logical $\ket{\overline{1}}$ code states for the thin M\"{o}bius strip, which in the angular momentum basis read:

\begin{subequations}
\label{eq:moebiuscodewords}
\begin{align}
\ket{\overline{0}} & = \sum_{\overset{\ell_1, \dots, \ell_N \in \mathbb{Z}}{\sum_{k=1}^N\ell_k = \mathrm{even} } } \ket{\ell_1, \dots, \ell_N, -\ell_1, \dots, -\ell_N}, \\
\ket{\overline{1}} & = \sum_{\overset{\ell_1, \dots, \ell_N \in \mathbb{Z}}{\sum_{k=1}^N\ell_k = \mathrm{odd} } } \ket{\ell_1, \dots, \ell_N, -\ell_1, \dots, -\ell_N}.
\end{align}
\label{eq:codewords-moeb}
\end{subequations}

In Subsec.~\ref{subsec:kitaevmirror} we show how Kitaev's current mirror qubit \cite{kitaevProtectedQubitBased2006, weissSpectrumCoherenceProperties2019} can be interpreted as a physical realization of the thin M\"{o}bius strip code described above.

\begin{figure}[htb]
\centering 
\subfloat[\label{sfig:TM1}]{\includegraphics[width=.3\linewidth]{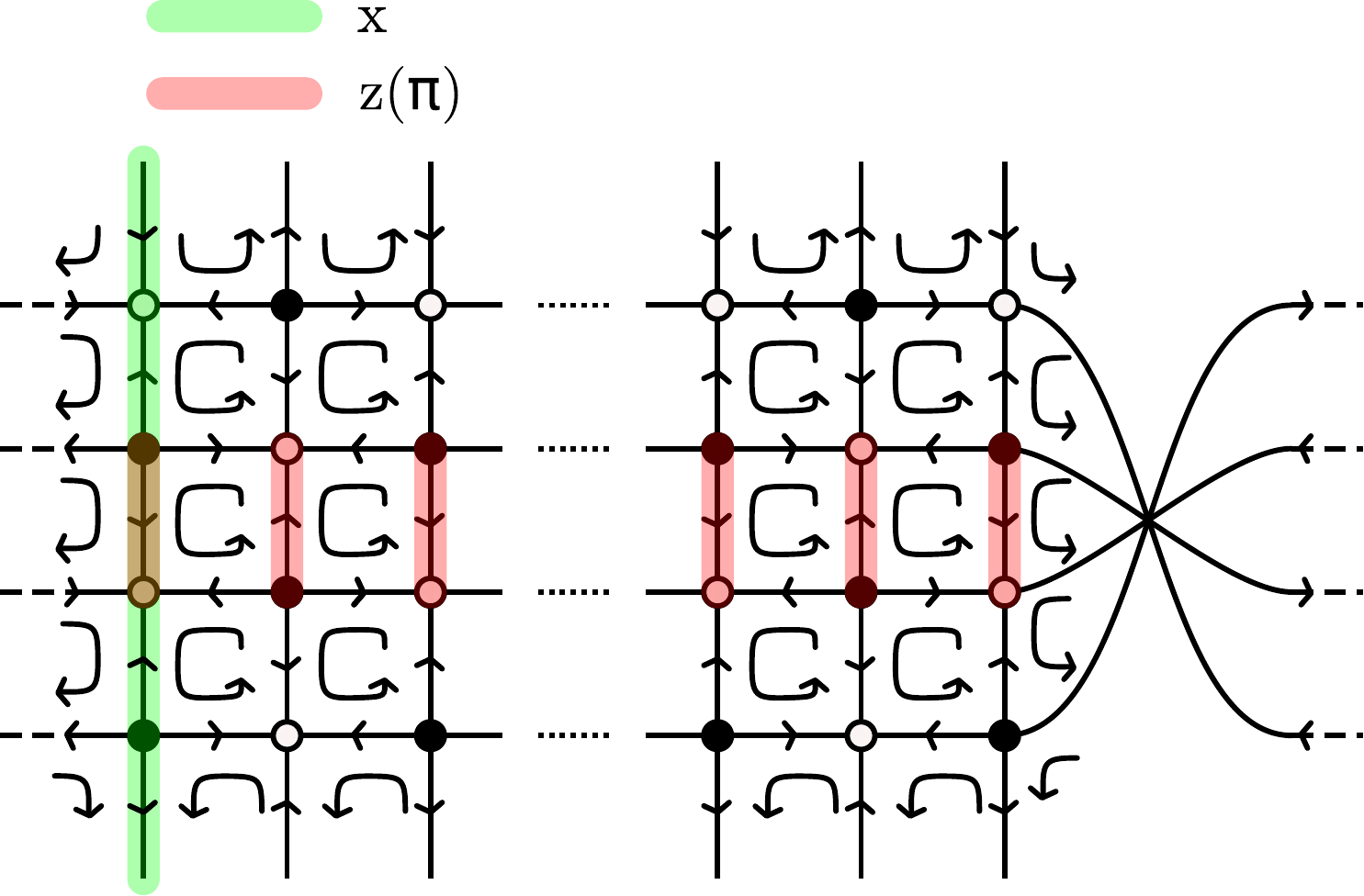}}\hfill
\subfloat[\label{sfig:TM2}]{\includegraphics[width=.3\linewidth]{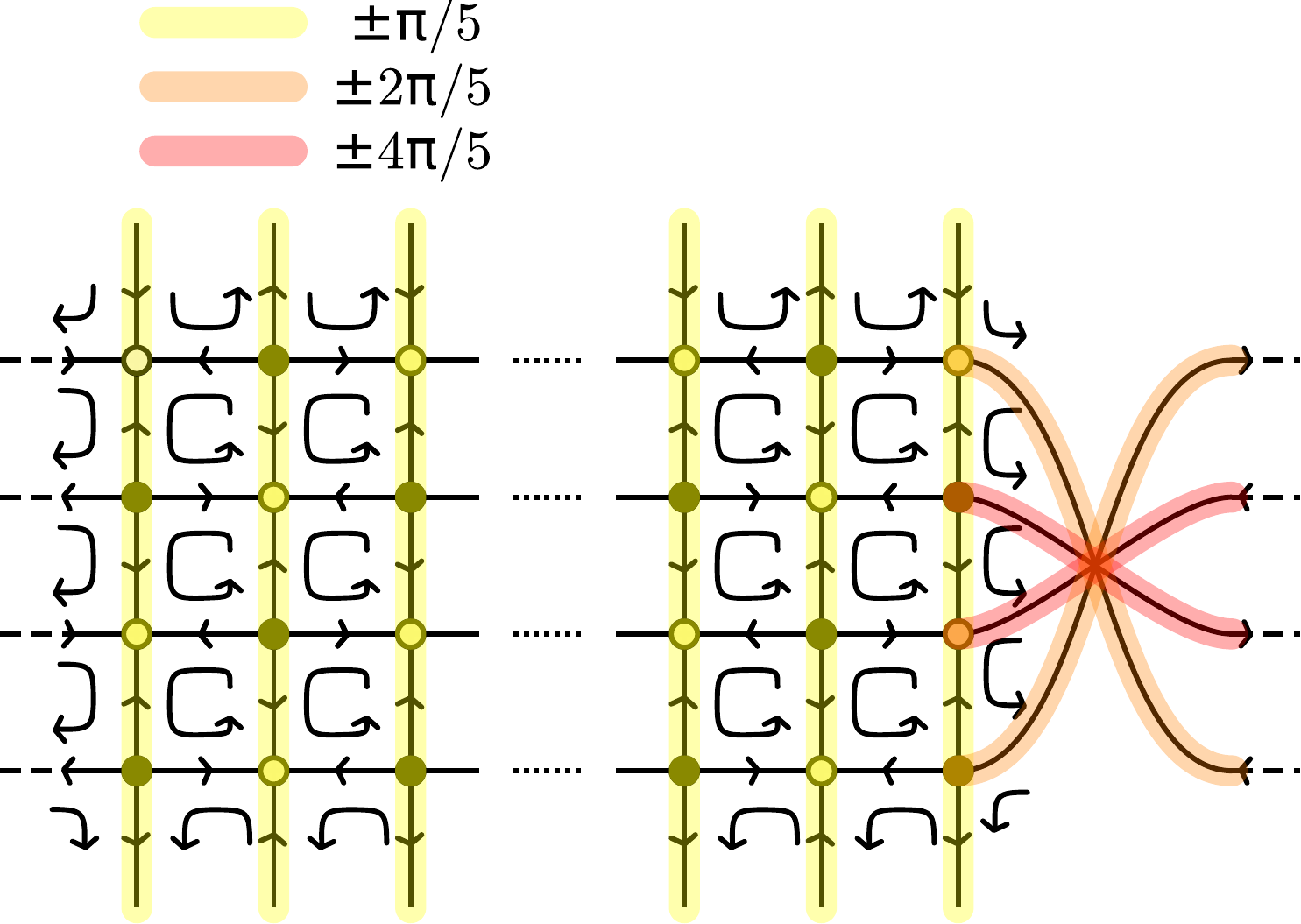}}\hfill
    \subfloat[\label{sfig:TM3}]{\includegraphics[width=.3\linewidth]{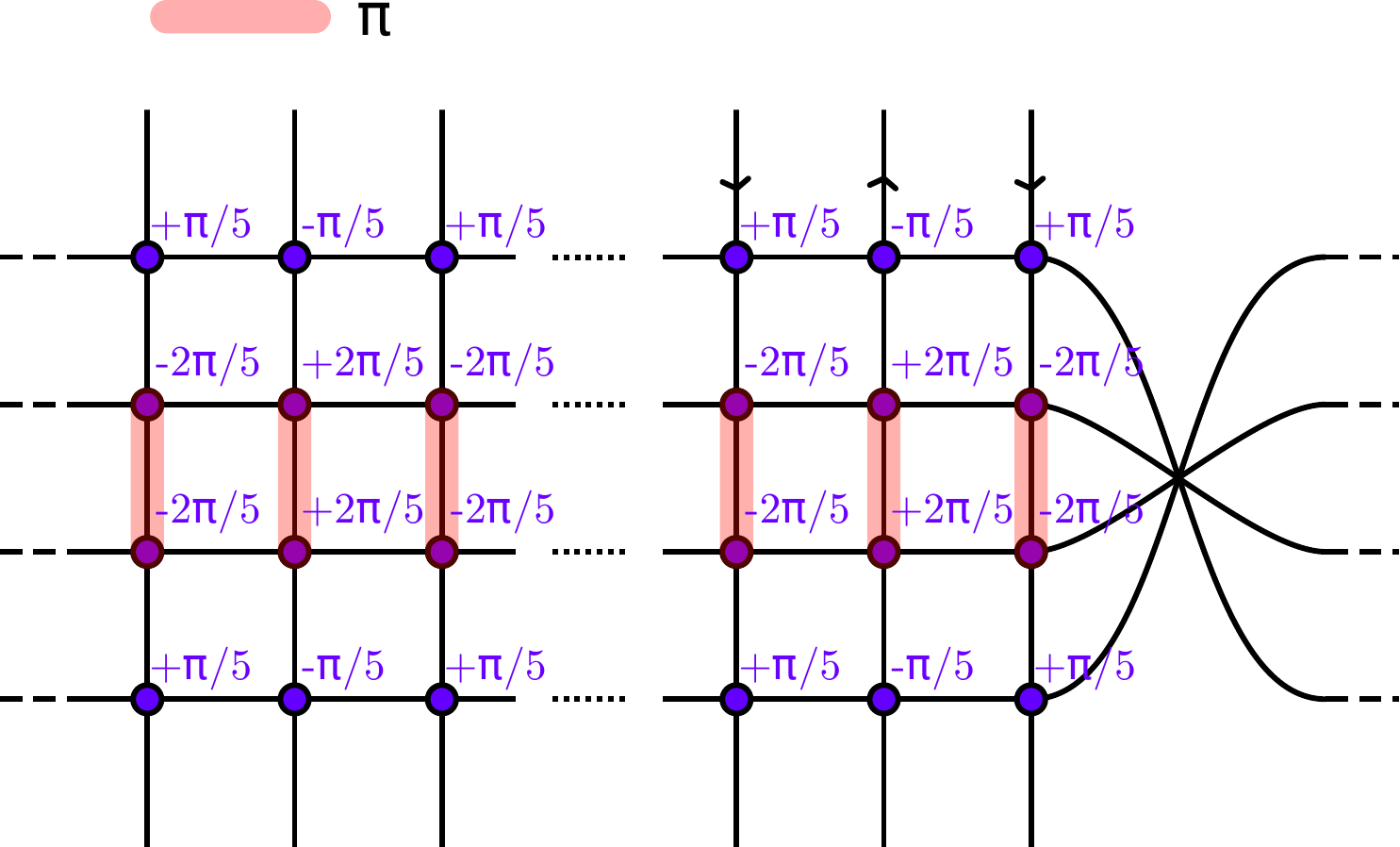}}
\caption{\protect\subref{sfig:TM1} A thicker M\"{o}bius strip with (odd) length $N$ and odd width $w=5$. In red is the support of the logical $\overline{Z}$ and in green is the support of the logical $\overline{X}$. \protect\subref{sfig:TM2} The support of the logical $\overline{Z}$ operator in red in \protect\subref{sfig:TM1} has been shifted to all the rows ($\pm \pi/5$ phase-shifts) and 2 edges with $2\pi/5$ and 2 with \(4\pi/5\) shifts occur at the twist: the shift is done by multiplying with vertex stabilizers with values given in \protect\subref{sfig:TM3}. \protect\subref{sfig:TM3} The value of the vertex stabilizers to choose to perform the spreading-out operation. To obtain the phase shift applied to each edge, one sums the values at its two neighbouring vertices (without changing the signs). Applying these shifts starting from the logical operator with \(\pi\) shifts in the middle row (highlighted in red) yields the logical operator in \protect\subref{sfig:TM2}.}
\label{fig:thick-moeb}
\end{figure}

We can make the M\"{o}bius strip code thicker: for example, we can take a strip of odd length $N$ and odd width $w$, see Fig.~\ref{fig:thick-moeb}. 
This M\"{o}bius-strip code, denoted as \(\mathbb{M}(w,N)\), on $2Nw-N$ rotors is defined by the following stabilizers. The $X$ stabilizer generators are labeled by the faces $\mu$ and the $Z$ stabilizer generators by the vertices $\nu$ of the lattice:
\begin{align}
   \mu=1,\ldots, Nw,\; && S_{\mu}^X  =e^{i \sum_{i \in\partial_2(\mu)} s_{\mu,i} \hat{\theta}_i},&& \mbox{with } s_{\mu,i}=\pm 1,  \notag \\
   \nu=1,\ldots, N(w-1),\; & & \forall \varphi\;  S_{\nu}^Z(\varphi)=e^{i\varphi \sum_{i \in \partial_1^*(\nu)}t_{\nu,i} \hat{\ell}_i},&& \mbox{with } t_{\nu,i}=\pm 1.
    \label{eq:moebstab}
\end{align}

Here the $\pm 1$ values for the $s_{\mu,i}$ and $t_{\nu,i}$ variables depend on orientation. Along the top and the bottom, the $X$ stabilizers act on three rotors, representing rough boundaries. The compact, minimal support, logical $\overline{Z}$ runs along a loop $\gamma^*$ on the dual lattice (in red in Fig.~\ref{sfig:TM1}) and applies a $\pi$-phaseshift on each rotor, i.e. $\overline{Z}=e^{i \pi \sum_{i \in \gamma^*} \hat{\ell}_i}$. We observe that 
the operator $\overline{Z}=e^{i \phi \sum_{i \in \gamma^*} \hat{\ell}_i}$ with $\phi\neq k \pi$, $k\in \mathbb{Z}$ does not commute due to the fact that the twisted face has support on the rotors $a$ and $b$ over which $\gamma^*$ runs of the form $e^{i (\hat{\theta}_a+\hat{\theta}_b)}$ (Note that an alternative commuting support for the logical of the form $e^{i \phi (\hat{\ell}_a-\hat{\ell}_b)}$ would not work due to the fact that it  does not commute with all other faces). The logical $\overline{X}$ runs over a straight line $\gamma$ from top to bottom (in green in Fig.~\ref{sfig:TM1}) and acts as $e^{i (\hat{\theta}_a-\hat{\theta}_b)}$ on two adjacent edges $a$ and $b$, incident to a vertex, in order to commute with the vertex check which has support $e^{i \varphi(\hat{\ell}_a+\hat{\ell}_b)}$ for any $\varphi$.
We can view $\overline{X}$ as applying a sequence of angular momentum or Cooper pair jumps along $\gamma$. The distance of the logical $\overline{X}$ is thus $d_X=w$, as $\gamma$ has to run from top to bottom in order to commute. These logical operators $\overline{Z}$ and $\overline{X}$ overlap on one edge where they anti-commute.

 To determine the distance $\delta_Z$ one considers spread-out logical operators, as discussed in Lemma \ref{lem:Zbound} in \ref{sec:spread}. We observe that one can move the support of the logical $\overline{Z}$ in red in Fig. \ref{sfig:TM1} on all the rows by multiplying by all $S_{\nu}^Z(\pm2\pi/5)$ for vertices alongside the logical operator and $S_{\nu}^Z(\pm\pi/5)$ for vertices closer to the boundary, see Fig.~\ref{sfig:TM3}. This creates a logical $\overline{Z}$ with $\pm \pi/5$-phaseshifts on all vertical edges, and it includes four horizontal edges with $\pm2\pi/5$ and \(\pm4\pi/5\) shifts at the twist, see Fig.~\ref{sfig:TM2}.

We can apply Lemma \ref{lem:Zbound} with $d_X=D_X=w$ and $N_X=N$ since there are $N$ disjoint representatives of the logical $\overline{X}$, all of weight $w$. One thus has
 
 \begin{prop}
 The rotor M\"{o}bius strip code, \(\mathbb{M}(w,N)\), of width $w$ and length $N$ encodes one logical qubit with distance $d_X=w$ and, for sufficiently large $w$, distance $\delta_Z\geq Nw \sin^2(\frac{\pi}{2w})\sim \frac{N \pi^2}{4 w}$.
 \label{prop:moeb-distance}
 \end{prop}

This directly implies that one should choose the length of the strip as $N=w^2$ in order to balance both distances and have them both increase with the number of physical rotors. 
Indeed this would yield the parameters
\begin{equation}
    \mathbb{M}(w,w^2): \left\llbracket 2w^3-w^2, (0,2), \left(w,\Theta(w)\right)\right\rrbracket_{\rm rot}.\label{eq:moebiusparameters}
\end{equation}
This yields a Möbius strip which is in a sense asymptotically 1D as the length increases quadratically faster than the width.

\subsubsection{Cylinder Encoding a Rotor}
\label{sec:cyl}

Instead of a M\"{o}bius strip, one can also choose a normal cylinder with a rough boundary on one side and periodic boundaries on the other side. The form of the checks is as in Eq.~\eqref{eq:moebstab}. In this case one encodes a logical rotor with $\overline{Z}(\phi)=e^{i \phi \sum_{i\in \gamma^*} \hat{\ell}_i}$ and $\overline{X}(m)=e^{i m\sum_{i \in \gamma} \pm \hat{\theta}_i}$ with alternating signs $\pm$ between adjacent edges $i$, see Fig.~\ref{fig:strip_rotor}.  We can again apply Lemma \ref{lem:Zbound} in Section \ref{sec:spread} using $D_X=w$ disjoint representatives for the logical $\overline{X}(1)$, each of weight $w$, and obtain

\begin{prop}
 The cylinder code of width $w$ and length $N$ encodes one logical rotor with $X$ distance $d_X=w$ and, for sufficiently large $w$, $\forall \alpha$, $\delta_Z\geq \frac{N w \sin^2(\frac{\alpha}{2w})}{\sin^2(\alpha/2)} \sim \frac{N}{w}$. 
  \label{prop:surface}
 \end{prop}

 Again, this implies that one should choose the length of the strip $N = w^2$ in order to balance both distances and have them both increase with the number of physical rotors. 
 The shape of the system is then also asymptotically 1D.
 In the next section we show that in higher dimensions we do not have to change the dimensionality of the system to obtain the same distance scaling. 

\begin{figure}[htb]
\centering 
\includegraphics[width=.35\linewidth]{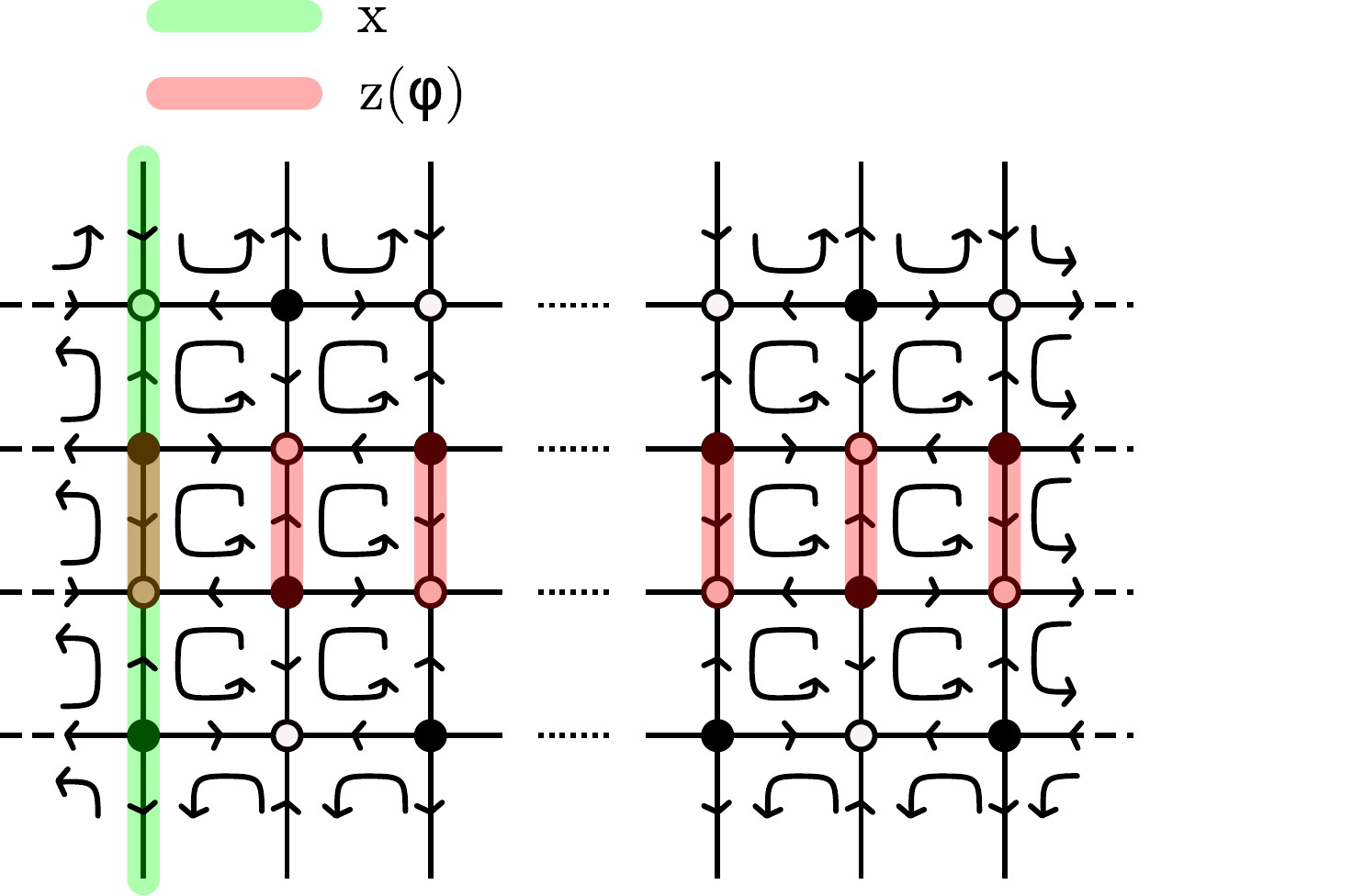}
\caption{(a) A normal strip with (odd) length $N$ and odd width $w=5$. The support of the logical $\overline{Z}$ is in red and the support of the logical $\overline{X}$ is in green. For this strip, one identifies the vertices at the left and right boundary.}
\label{fig:strip_rotor}
\end{figure}

As a curiosity one can also consider a tessellation of the Klein bottle: since the Klein bottle has first homology group $H_1({\rm Klein\, bottle}, \mathbb{Z})=\mathbb{Z} \times \mathbb{Z}_2$, it encodes both a qubit and a rotor with parameters \(\llbracket n, (1,2), (d_X,\delta_Z)\rrbracket_{\rm rot}\). \\

In general, using \(2D\) manifolds, we cannot encode more than a logical qubit alongside some number of logical rotors.
This comes from the relation between orientability and the torsion subgroup of \(H_{D-1}(\mathcal{M},\mathbb{Z})\) at the \(D-1\) level of a connected and closed \(D\)-dimensional manifold.
More precisely, if \(\mathcal{M}\) is orientable then there is no torsion in \(H_{D-1}(\mathcal{M},\mathbb{Z})\) whereas if it is non-orientable the torsion subgroup is \(\mathbb{Z}_2\), see \cite[Corollary 3.28]{hatcherAlgebraicTopology2002}.
We also give an elementary proof of this fact in the case where we have a finite tessellation of \(\mathcal{M}\) in Appendix~\ref{sec:orient=Z2}.
In larger dimensions \(D>2\), the torsion subgroups of \(H_k(\mathcal{M},\mathbb{Z})\) for \(k<D-1\) have no connection to orientability and can be arbitrary and the number of encoded logical qubits or qudits is not restricted.

\subsubsection{Higher Dimensions}
\label{sec:HD}

\begin{figure}[htb]
    \centering
    \subfloat[\label{sfig:RP3}]{\includegraphics[width=.25\linewidth]{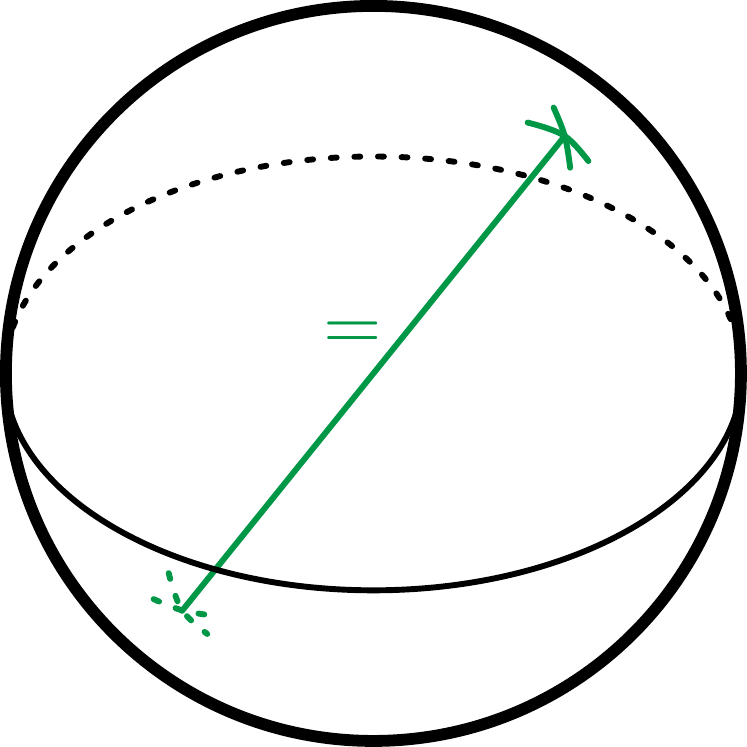}}\hfil
     \subfloat[\label{sfig:RP3*}]{\includegraphics[width=.25\linewidth]{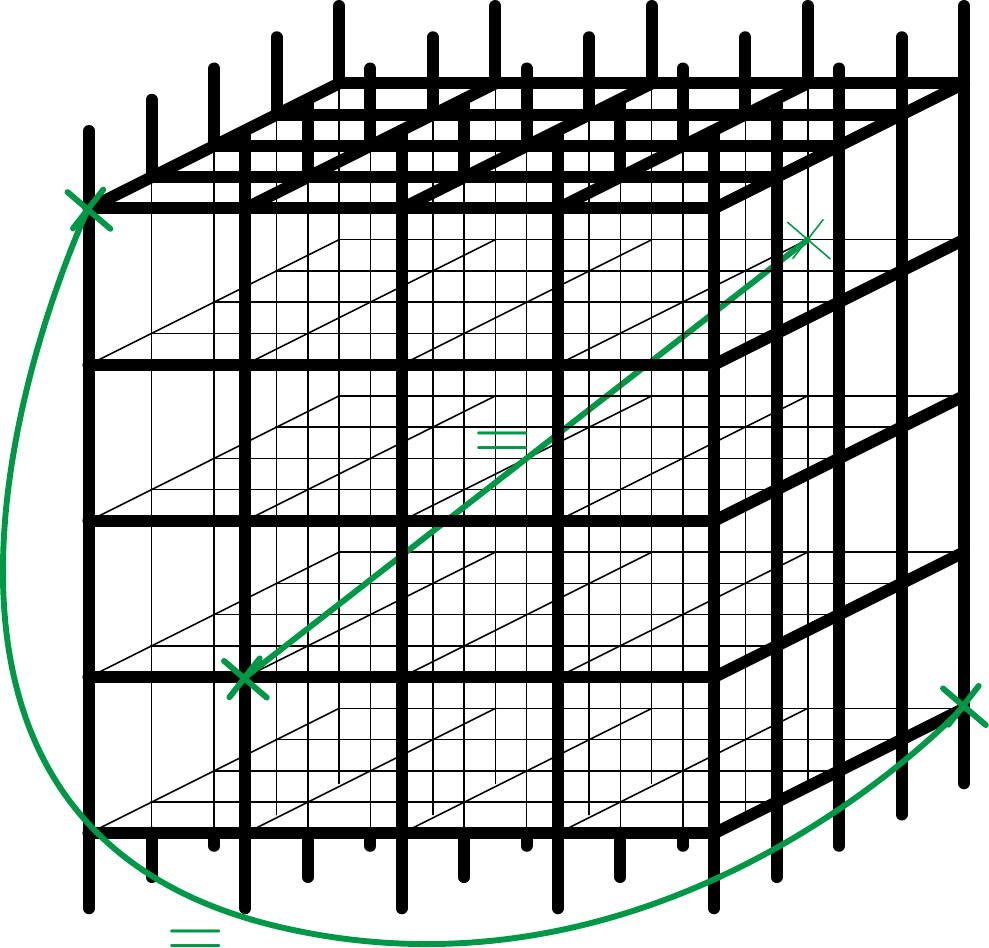}}
    \caption{\protect\subref{sfig:RP3} Representation of \(\mathbb{RP}^3\) as a \(3\)-ball with identified antipodal points on its boundary. \protect\subref{sfig:RP3*} A specific tessellation of \(\mathbb{RP}^3\) from which a point has been removed creating a rough boundary in the shape of a sphere.
    The top rough boundary and the bottom rough boundary in the drawing are connected and form a single boundary.
    They form the double cover of a real projective plane, i.e. a sphere.
    Each section plane is a real projective plane and hosts a \(Z\) logical operator.
    There are \(N^2\) vertical \(X\) logical operators of weight \(N\) each.
    }
    \label{fig:RP3}
\end{figure}

To overcome the restrictions imposed by \(2D\)-manifolds on the number of logical qubits and the \(Z\) distance scaling in, say, Proposition \ref{prop:moeb-distance} and Proposition \ref{prop:surface}, an option is to turn to higher-dimensional manifolds, starting with \(3D\).
Take a \(3\)-torus, say the \(N\times N\times N\) cubic lattice with periodic boundary conditions. We identify the edges with rotors, faces with \(X\) stabilizers and vertices with \(Z\) stabilizers, so that the number of physical rotors is \(3N^3\).
We denote this code by \(\mathbb{T}_3\).
The homology at level 1 is \(\mathbb{Z}^3\), implying that there are 3 logical rotors in the codespace. 
With this choice of dimensions, the $X$-type logical operators are non-trivial closed \(1D\) loops and the $Z$-type operators are $2D$ cuts in $2D$-torus shape.
We can compute the following parameters for the \(3D\)-toric rotor code
\begin{equation}
    \mathbb{T}_3: \left\llbracket 3N^3, (3,0), (N,N)\right\rrbracket_{\rm rot}.\label{eq:T3parameters}
\end{equation}
To get the lower bound on the \(Z\) distance via Lemma \ref{lem:Zbound}, we can exhibit a set of \(N_X = \Theta(N^2)\) disjoint \(X\)-logical operators each of weight \(D_X = N\).
All the parallel lines along the \(x\)-direction for instance.
This yields using Eq.~\eqref{eq:Zbound-rotor}
\begin{equation}
    \delta_Z = \Theta\left(\frac{N_X}{D_X}\right)= \Theta\left(N\right).
\end{equation}
So here we have a genuine growing distance with the system size for both \(X\) and \(Z\) without tweaking its shape. 

To get an example with torsion at level 1 we can turn to the real projective space in \(3D\) denoted as \(\mathbb{RP}^3\).
A way to visualize \(\mathbb{RP}^3\) is to take a \(3\)-ball and identify antipodal points on its boundary surface, see Figure~\ref{sfig:RP3}.
We have that \(H_1\left(\mathbb{RP}^3,\mathbb{Z}\right)=\mathbb{Z}_2\) hence encoding a single logical qubit.

Similarly as how a rough Moebius strip can be obtained from the projective plane by removing a vertex we can obtain a slightly simpler manifold than \(\mathbb{RP}^3\), with a rough boundary in the shape of a \(2\)-sphere by removing one vertex (but keeping the edges attached to it) from \(\mathbb{RP}^3\). To be concrete, we take again a \(N\times N\times N\) chunk of the cubic lattice with rough boundaries at the \(z=0\) and \(z=L\) planes.
Then we connect the four other sides anti-podally, see Figure~\ref{sfig:RP3*}. For short we label this code as \({\mathbb{RP}^{3}}^*\).

This punctured \(\mathbb{RP}^3\) code has parameters
\begin{align}
    {\mathbb{RP}^{3}}^* : \left\llbracket3N^3-N^2, (0,2), (N,N)\right\rrbracket_{\rm rot}.\label{eq:RP3*parameters}
\end{align}
The \(Z\) distance is obtained using the same set of disjoint \(X\) logical operators as in the 3-torus case, i.e \(N^2\) vertical paths from the top rough boundary to the bottom one in Figure~\ref{sfig:RP3*}.
In Section \ref{sec:protection} we will come back to the impact of distance scaling on the protection of the encoded information, although this is partially left to future work.

Many other \(3\)-manifolds with different 1-homology can be found and some have been tabulated \cite{hodgsonSymmetriesIsometriesLength1994}.
For instance the 3-manifold named ``m010(1,3)'' in the ``orientable closed manifold census'' of SnapPy \cite{SnapPy} has its 1-homology equal to \(\mathbb{Z}_{42}\).

\subsection{Constructions from a Product of Chain Complexes}
\label{sec:productconstruction}

For qubit and qudit codes, constructions exist to build CSS quantum codes from products of classical codes.
The first construction of this sort was the hypergraph product construction \cite{TZ:hypergraph} which can be reformulated as a tensor product of chain complexes \cite{audouxTensorProductsCSS2019}.
We show here how to use it to generate quantum rotor codes.

Given two chain complexes, \(\left(\mathcal{C}, \partial^\mathcal{C}\right)\) and \(\left(\mathcal{D}, \partial^\mathcal{D}\right)\) given by \(\mathbb{Z}\)-modules \(C_j\), \(D_j\) and boundary maps \(\partial_j^\mathcal{C}:C_j\rightarrow C_{j-1}\) and \(\partial_j^\mathcal{D}:D_j\rightarrow D_{j-1}\), we define the tensor product \(\left(\mathcal{E}, \partial^\mathcal{E}\right)\),
where we define its \(\mathbb{Z}\)-modules and boundary maps as
\begin{equation}
    E_{k} = \bigoplus_{i+j=k}C_i\otimes D_j,\quad 
    \partial^{\mathcal{E}}_k = \sum_{i+j=k} \partial^\mathcal{C}_i\otimes \id_{D_j} + (-1)^{i}\id_{C_i}\otimes \partial^\mathcal{D}_j.
\end{equation}

The tensor product of groups, \(A\otimes B\), is defined by the following properties
\begin{align}
    \forall (a, a^\prime)\in A^2,\; \forall (b,b^\prime)\in B^2,\; \forall n\in\mathbb{Z},\; &&na\otimes b &= a \otimes nb\\
   && (a+a^\prime)\otimes b &= a\otimes b + a^\prime\otimes b\\
   && a\otimes (b+b^\prime) &=  a\otimes b + a\otimes b^\prime.
\end{align}
Note the following useful identities:
\begin{equation}
    \mathbb{Z}\otimes\mathbb{Z} = \mathbb{Z},\qquad \mathbb{Z}\otimes \mathbb{Z}_d = \mathbb{Z}_d,\qquad \mathbb{Z}_{d_1}\otimes\mathbb{Z}_{d_2} = \mathbb{Z}_{\gcd(d_1,d_2)}.\label{eq:tensor}
\end{equation}
The homology of \(\mathcal{E}\) is readily obtained from the Künneth theorem, see for instance \cite[Theorem 3B.5]{hatcherAlgebraicTopology2002}.
\begin{thm}[Künneth Theorem] Given two chain complexes \(\mathcal{C}\) and \(\mathcal{D}\) such that the \(C_j\) are free, the homology groups of the product are such that
\begin{equation}
    H_k(\mathcal{E}) \simeq \left[\bigoplus_{i+j=k}H_i(\mathcal{C})\otimes H_j(\mathcal{D})\right]\oplus\left[\bigoplus_{i+j = k-1 }{\rm Tor}\left(H_i(\mathcal{C}),H_j(\mathcal{D})\right)\right].\label{eq:kunneth}
\end{equation}
\end{thm}
In order to compute the \({\rm Tor}\) part in our case, we only need to know that it maps pairs of groups to some other group obeying the following rules \cite[Proposition 3A.5]{hatcherAlgebraicTopology2002}
\begin{align}
    {\rm Tor}(A,B) &= {\rm Tor}(B,A), &{\rm Tor}\left(\bigoplus_i A_i, B\right) &= \bigoplus_i{\rm Tor}(A_i,B),\nonumber\\
    {\rm Tor}(A, \mathbb{Z}) &= 0, & {\rm Tor}(\mathbb{Z}_{d_1},\mathbb{Z}_{d_2}) &= \mathbb{Z}_{\gcd(d_1,d_2)}. 
    \label{eq:Tor}
\end{align}

From Eq.~\eqref{eq:kunneth} and Eqs.~\eqref{eq:tensor} and \eqref{eq:Tor}, we see that if we want to construct rotor codes encoding logical rotors we can take the product of chain complexes with free homology and the usual parameter scaling will follow in the same way as with qubit code construction.
If we want a rotor code encoding qubits or qudits, the torsion groups need to agree or we can combine free homology with torsion.

We can thus adopt different construction strategies depending on what we want our logical space to be.
In Appendix~\ref{app:products} we develop in detail three different ways to use this construction to construct rotor codes.
Notably we show how taking the product of a repetition code and a good LDPC code with asymmetric code-lengths yields a code family for encoding logical rotors with parameters
\begin{equation}
    \left\llbracket n, \left(\Theta(\sqrt[3]{n}),0\right) ,\left(\Theta(\sqrt[3]{n}),\Theta(\sqrt[3]{n})\right)\right\rrbracket_{\rm rot}.
\end{equation}
The distance scaling is the same as that of the 3D-toric code, Eq.~\eqref{eq:T3parameters}, or skewed 2D-cylinder code but with a better encoding rate.

We also show that taking a repetition code with a sign twist together with a good LDPC code still with asymmetric code-lengths yields a code family with the same parameter scalings but encoding logical qubits, so described by parameters
 \begin{equation}
        \left\llbracket n, \left(0,2^{\Theta(\sqrt[3]{n})}\right) ,\left(\Theta(\sqrt[3]{n}),\Theta(\sqrt[3]{n})\right)\right\rrbracket_{\rm rot}.
\end{equation}
Again, the distance scaling is the same as for the 3D real projective space code, Eq.~\eqref{eq:RP3*parameters}, or skewed 2D Moebius code, Eq.~\eqref{eq:moebiusparameters}, but with a better encoding rate.
 
 \section{Physical Realizations in Circuit-QED}
	\label{sec:circuit-QED}

Superconducting circuits form a natural platform for the realization of rotor codes, where the superconducting phase $\theta$ of a superconducting island of sufficiently small size can realize the physical rotor degree of freedom, while its conjugate variable $\ell$ represents the (excess) number of Cooper pairs on the island (relative to another island). Thus, we seek to engineer the Hamiltonian or the stabilizer measurements of a rotor code with circuit-QED hardware, namely using Josephson junctions and capacitors between superconducting islands, and possibly inductors, to realize the $X$ and $Z$ terms in the Hamiltonian. To this end, we will use standard procedures for converting an electric circuit to a Hamiltonian \cite{Vool2016, burkardMultilevelQuantumDescription2004, rasmussenSuperconductingCircuitCompanion2021}. This passive approach to quantum error correction has been discussed in circuit-QED systems in Ref.~\cite{Doucot_2012}, and pursued to obtain the Hamiltonian of the GKP code \cite{RBCD, le2019, conrad2021, kolesnikow2023} and of the surface code \cite{sametiSuperconductingQuantumSimulator2017}. In what follows, in order to avoid confusion, we will also denote matrices with bold symbols.

	Naturally, in such superconducting systems, the number of Cooper pairs on an island is confined to a range which is set by the capacitive couplings that the island has to other conducting structures. In particular, we imagine that each island $i$ is capacitively coupled to a common ground-plane via a sufficiently large capacitance $C_{g_i}$ and for simplicity we will take $C_{g_i}=C_g$. This sets an energy scale for the charge fluctuations on each island. In the absence of any further couplings, the Hamiltonian would be $H=4 E_{C_{g}}\sum_i \hat{\ell}_i^2$ with \begin{align}
 E_{C_g}=\frac{e^2}{2C_{g}},
 \end{align}
	where $e$ is the charge of a single electron.
	We can then consider the effect of adding a large capacitor $C$ between two islands $1$ and $2$, with $C\gg C_g$.
  The Hamiltonian of two such islands is
  \begin{equation}
   H=\frac{4 e^2}{2}\sum_{i,j=1,2}\hat{\ell}_i (\bs{C}^{-1})_{ij} \hat{\ell}_j,
\end{equation}
    with capacitance matrix
    \begin{align}
        \bs{C}=\begin{pmatrix}C+C_{g} & -C \\-C & C+C_{g} \end{pmatrix}.
  \end{align}
  The eigenvalues of $\bs{C}^{-1}$ are $1/C_g$ (eigenvector $(1,1)$) and $\frac{1}{2C+C_g}$ (eigenvector $(1,-1)$) hence
  \begin{align}
   H=2 E_{C_g}
   (\hat{\ell}_1+\hat{\ell}_2)^2+\frac{e^2}{2C+C_g}(\hat{\ell}_1-\hat{\ell}_2)^2.
   \label{eq:CC}
      \end{align}
     For large $C \gg C_g$, the second term is small and the first term enforces the constraint $\hat{\ell}_1+\hat{\ell}_2=0$. This shows that a rotor code with weight-2 checks of the form $e^{i (\hat{\ell}_i+\hat{\ell}_j)}$ can be fairly directly realized using capacitors as long as the pair $(i,j)$ is disjoint from other pairs $(k,l)$. In the Sections \ref{subsec:fourphase}, \ref{subsec:zeropi} and Appendix \ref{app:SW4phase} we will refer to this idea as gapping out the `agiton' variable $\frac{1}{2}(\ell_i+\ell_j)$. The code Hamiltonian is realized in the low-energy no-agiton sector with $\ell_i+\ell_j=0$  where the `exciton' variable $\frac{1}{2}(\ell_i-\ell_j)$ can still vary.
     
     The Hamiltonian of a superconducting Josephson junction between two islands $i$ and $j$ is given by $H=-E_J\cos(\hat{\theta}_i-\hat{\theta}_j)$, if we neglect the small capacitance induced by the junction between the nodes. Hence, this naturally represents a code constraint of weight-2. 

To understand the challenge of engineering an electric network which implements all code constraints, consider the following approach. One designs a capacitive network of nodes (all coupled to ground) which is composed of disjoint connected components $V_m$ with $m=1,2,\ldots$, with the nodes in each connected component $\mathrm{V}_m$ connected by some large capacitance $C \gg C_g$. The capacitance matrix $\bs{C}$ will have the smallest eigenvalue of $C_g$ associated with eigenvectors of the form $(1,1,\ldots, 1)$ on the support of any of the connected components. Hence one can obtain a set of capacitive code constraints $(\sum_{i\in \mathrm{V}_m} \hat{\ell}_i)^2$ for non-overlapping sets $\mathrm{V}_m$. However, these capacitive constraints act only on non-overlapping sets of nodes and thus we need a way of `identifying' nodes. However, if we would do this too strongly, then it was a priori incorrect to treat each connected component separately. Instead, we want to identify the node variables only in the subspace where one obeys the constraints $\sum_{i\in \mathrm{V}_m} \hat{\ell}_i=0$ whose violation costs energy $E_{C_g}$. This can be in principle be done using inductors which should act as closed circuits at sufficiently low energies. Furthermore, inside the subspace where the capacitive constraints are obeyed, the pair-wise Josephson junction terms should be treated perturbatively so as to generate face terms which express the joint tunneling of multiple Cooper pairs. In Sections \ref{subsec:zeropi} and \ref{subsec:kitaevmirror} we will show this approach for some particular known examples.

	\subsection{Subsystem Rotor Codes and Bacon-Shor Code}
\label{subsec:BS} 

One can easily generalize the definition of stabilizer rotor codes, Definition \ref{def:rotorcode} to subsystem rotor codes where one starts with a group $\mathcal{G}$ generated by non-commuting generalized Pauli operators $X(\bs{m})$ and $Z(\bs{\phi})$. The reason to study subsystem \cite{SBT:sub} (or Floquet) codes is that one can potentially construct non-trivial codes with gauge checks operators acting only on two rotors. In particular, due to the interest of rotor codes in circuit-QED, we consider codes of a particular restricted form with 2-rotor $X$ checks, which, as Hamiltonians, relate to a Josephson junction between two superconducting islands, i.e. $H \sim -\cos(\hat{\theta}_a-\hat{\theta}_b)$. We prove the following:

\begin{prop}[Josephson-Junction Based Subsystem Codes]
Given $n$ rotors, let \(\mathcal{G}=\left\langle e^{i \bs{h}_j^X \cdot \bs{\hat{\theta}} }, e^{i \varphi \bs{h}_k^Z \cdot \bs{\hat{\ell}} }\;\middle\vert\;\forall j=1,\ldots, r_x, \forall \varphi \in \mathbb{T}, \forall k=1,\ldots, r_z\right\rangle\) where the vector $\bs{h}_j^X$ is of the general restricted form $\bs{h}_j^X=(0,\ldots, 1,0, \ldots,0,-1, \ldots, 0)$. Let $\mathcal{C}(\mathcal{G})$ be the group of generalized Paulis which commute with all elements in $\mathcal{G}$. Let $\mathcal{S}=\mathcal{G} \cap \mathcal{C}(\mathcal{G})$ be the stabilizer center \footnote{In this notation we ignore any phases by which elements in the center can be multiplied due to the non-commutative nature of $\mathcal{G}$, so that $\mathcal{S}$ is a stabilizer group with a $+1$ eigenspace.}. Then either $\mathcal{C}(\mathcal{G})=\mathcal{S}$ (i.e., one encodes nothing) or $\mathcal{C}(\mathcal{G})=\left\langle \mathcal{S}, \overline{X}_i(m), \overline{Z}_i(\phi), \;\middle\vert\;\forall \phi\in \mathbb{T},m\in \mathbb{Z}, i=1,\ldots,k \right\rangle$ (i.e. one encodes some $k > 1$ logical rotors). 
\label{prop:subsys}
\end{prop}

\begin{proof}
 Consider a Pauli $Z(\bs{\phi})$ that should commute with elements of $\mathcal{G}$. In order to commute with some $e^{i \bs{h}_j^X \cdot \bs{\hat{\theta}}}$, with $\bs{h}_j^X$ having nonzero entries on rotor $a$ and $b$, the support of $Z(\bs{\phi})$ on rotor $a$ and $b$ must be of the form $e^{i \phi(\hat{\ell}_a+\hat{\ell}_b)}$ for \emph{any} $\phi \in \mathbb{T}$ (or its Hermitian conjugate). Notice that the case in which the support is of the form $e^{i \pi (\hat{\ell}_a - \hat{\ell}_b)}$ is a special case of $e^{i \phi (\hat{\ell}_a + \hat{\ell}_b)}$, since $e^{i \pi(\hat{\ell}_a+\hat{\ell}_b)}=e^{i \pi(\hat{\ell}_a-\hat{\ell}_b)}e^{i 2\pi 
\hat{\ell}_b}=e^{i \pi(\hat{\ell}_a-\hat{\ell}_b)}$. Now consider the support of the collection of vectors $\bs{h}_j^X$: this could break down in several disjoint connected components $V_m$. 
Then to make some operator $Z(\bs{\phi})$ commute with all $\bs{h}_j^X$, without loss of generality, its support on a connected component $V_m$ must be of the form $e^{i\phi \sum_{i\in V_m} \hat{\ell}_i}$ for some {\em arbitrary} $\phi \in \mathbb{T}$. In subsystem codes logical operators are operators that are in $\mathcal{C}(\mathcal{G})$, but not in $\mathcal{G}$, i.e., not in the stabilizer center $\mathcal{S}$. Thus, we just need to exclude the possibility that an operator $Z(\bs{\phi})$ of the previous form is not in $\mathcal{G}$ for some discrete set of $d$ values of $\phi$ for any $d$. However, since the generators of the $Z$ part of $\mathcal{G}$ are of the form $e^{i \phi \bm{h}_k^Z \cdot \bs{\hat{\ell}}}$, $ \forall \phi \in \mathbb{T}$, this cannot be the case.
\end{proof}

\begin{remark}
    If $\bs{h}_j^X$ can be of the form $\bs{h}_j^X=(0,\ldots, 1,0, \ldots,0,1, \ldots, 0)$, then the no-go result expressed in the Proposition \ref{prop:subsys} would not hold. Such constraint would correspond to a Hamiltonian with terms $-\cos(\hat{\theta}_a+\hat{\theta}_b)$ which would model a coherent increase or decrease of the number of Cooper pairs on both islands $a$ and $b$ by 1 (instead of the tunneling of a Cooper pair through the junction) which is not immediately physical.\\
\end{remark}

The following example demonstrates that even though we have defined (stabilizer) rotor codes in Definition \ref{def:rotorcode} with checks for continuous values $\bs{\varphi} \in \mathbb{T}^{r_z}$ in Eq.~\eqref{eq:stabilizergroup} (no modular constraints), a subsystem rotor code can have a stabilizer subgroup in which $Z$-checks only appear for discrete values of $\bs{\varphi}$.
\begin{examp}
Consider an even-length chain of $n$ rotors with \(\mathcal{G}=\left\langle e^{i(\hat{\theta}_j-\hat{\theta}_{j+1})}, e^{i\varphi(\hat{\ell}_j-\hat{\ell}_{j+1})}\;\middle\vert\;\forall j=1,2,\ldots,n-1,\forall \varphi\in \mathbb{T}\right\rangle\). The stabilizer subgroup is $\mathcal{S}=\mathcal{G} \cap \mathcal{C}(\mathcal{G})= e^{i\pi \sum_{j=1}^{n/2} (\hat{\ell}_{2j-1}-\hat{\ell}_{2j})}$ as this is the only element which is in $\mathcal{G}$ and which commutes with all elements in $\mathcal{G}$. The subsytem code encodes one logical rotor, namely $\mathcal{C}(\mathcal{G})\backslash \mathcal{G}$ is generated by $\overline{Z}(\phi)=e^{i \frac{\phi}{n} \sum_{j=1}^n \hat{\ell}_j}$ and $\overline{X}(m)=e^{i m \sum_{j=1}^n \hat{\theta}_j}$ obeying Eq.~\eqref{eq:paulicommutation}. \\
\label{exam}
\end{examp}

The point of Proposition \ref{prop:subsys} is that it demonstrates that one cannot capture the difference between a M\"{o}bius strip and a cylinder using such restricted (Josephson-junction based) weight-2 checks. 

\subsubsection{Thin Rotor Bacon-Shor code}
An example of a subsystem code is the rotor Bacon-Shor code, of which we consider a thin example. 
For such thin Bacon-Shor code, in analogy with the qubit Bacon-Shor code \cite{aliferisSubsystemFaultTolerance2007}, the group $\mathcal{G}$ is generated by
\begin{align}
     G_j^X&=e^{i(\hat{\theta}_j-\hat{\theta}_{j+1})},& j&=1,\ldots, N-1, \\
     \tilde{G}_{j}^X&=e^{i(\hat{\theta}_{N+j}-\hat{\theta}_{N+j+1})},&j&=1,\ldots, N-1,  \\
G_j^Z(\phi)&=e^{i\varphi(\hat{\ell}_j+\hat{\ell}_{N+j})},& j&=1,\ldots, N, \forall \varphi \in \mathbb{T}.
\end{align}
The logical rotor operators can be taken as $\overline{Z}(\phi)=e^{i \phi \sum_{i=1}^N \hat{\ell}_i}$ and  $\overline{X}(m)=e^{i m (\hat{\theta}_k-\hat{\theta}_{N+k})}$, where we can choose $k$ to be any $k$ in $\{1, \dots, N\}$, forming a row and column which intersect on a single rotor as in the standard Bacon-Shor code \cite{aliferisSubsystemFaultTolerance2007}. The `double column' operators and `double row' operators generate the stabilizer group $\mathcal{S}=\mathcal{G} \cap \mathcal{C}(\mathcal{G})$:
\begin{align}
S_j^X&=G_j^X (\tilde{G}_{j+1}^X)^{-1}=e^{i(\hat{\theta}_j-\hat{\theta}_{j+1}-\hat{\theta}_{N+j}+\hat{\theta}_{N+j+1})},\\
S^Z(\varphi)&=\prod_{j=1}^NG_j^Z(\varphi) = e^{i\varphi\sum_{j=1}^{2N} \hat{\ell}_k}, \forall \varphi.
\label{eq:stab-BS}
\end{align}
The $X$ distance of the code is $d_X=2$ since $\overline{X}(m)$ acts on at least two rotors and the minimum in the definition of the $X$ distance in Eq.~\eqref{eq:distanceX} is achieved at $m=1$.
For the logical $Z$ distance we consider $\delta_Z$ in Eq.~\eqref{eq:distanceZ}. One notices that one can make a slightly spread-out logical $\overline{Z}(\phi)=e^{i\frac{\phi}{2}(\sum_{j=1}^N \hat{\ell}_j-\sum_{j=N+1}^{2N}\hat{\ell}_j)}$, which extends over two rows, and $\delta_Z\sim N$ as the expression in Eq.~\eqref{eq:distanceZ} is linear in $N$ times a constant $C(\phi)$ which depends on the choice of $\phi$ but which is bounded away from 0 for all $\phi$.

The targeted (dimensionless) Hamiltonian associated with this code is of the form
\begin{equation}
    H_{\rm thin-BS}=-\sum_{j=1}^{N-1}\left[\cos(\hat{\theta}_j-\hat{\theta}_{j+1})+ \cos(\hat{\theta}_{N+j}-\hat{\theta}_{N +j+1})\right]+\sum_{j=1}^N (\hat{\ell}_j+\ell_{N+j})^2.
    \label{eq:HthinBS}
\end{equation}
In principle $H_{\rm thin-BS}$ has a spectrum in which each eigenlevel is infinitely-degenerate, i.e. in each degenerate eigenspace, we can build a rotor basis $\ket{\overline{\ell}}$ with $\overline{X}(m) \ket{\overline{\ell}}=\ket{\overline{\ell+m}}$. In practice, the Hamiltonian is approximately realized by the circuit in Fig.~\ref{fig:BS}, where each node is connected to a ground plane via the capacitance $C_g\ll C$. The Hamiltonian, omitting the Josephson capacitances and assuming that the Josephson junctions are equal, of the network then equals:
\begin{multline}
    H_{{\rm thin-BS,circuit}}=-E_J\sum_{j=1}^{N-1}\left[ \cos(\hat{\theta}_j-\hat{\theta}_{j+1})+ \cos(\hat{\theta}_{N + j}-\hat{\theta}_{N+j+1})\right]+\\ 2E_{C_g}\sum_{j=1}^N (\hat{\ell}_j+\hat{\ell}_{N+j})^2+ 
    \sum_{j=1}^N\frac{ e^2}{2C+C_g}(\hat{\ell}_j-\hat{\ell}_{N+j})^2,
\end{multline}
where the large capacitance $C \gg C_g$ suppresses the last term.

\begin{figure}[htb]
\centering 
\includegraphics[height=3cm]{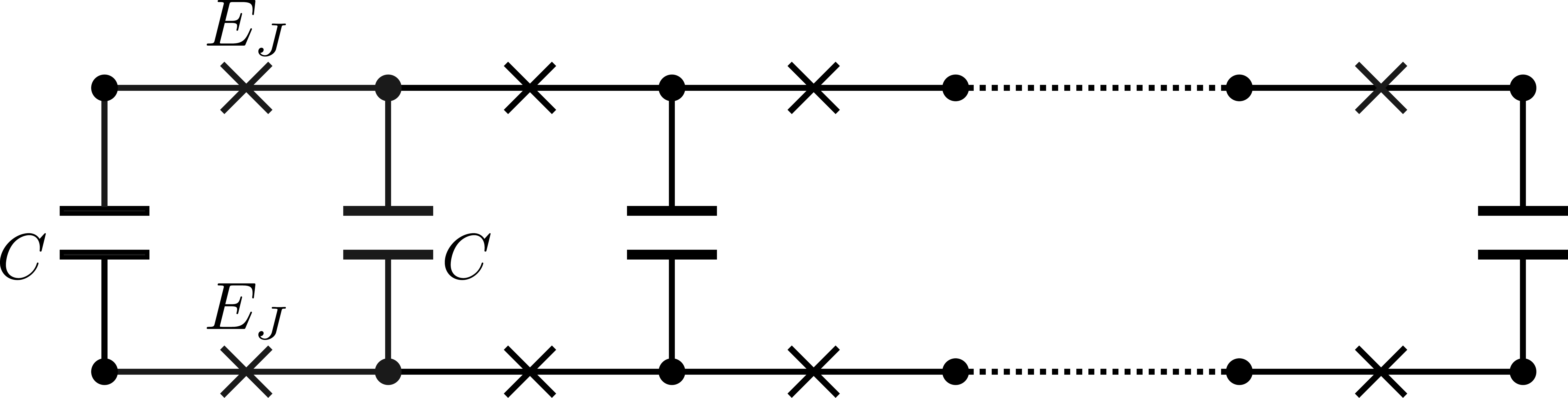}
\caption{Electric circuit for the thin Bacon-Shor code: capacitances to ground for each node are not shown.}
\label{fig:BS}
\end{figure}

	It is clear that, in general, the constraints imposed by parity check matrices $H_X$ and $H_Z$ do not immediately translate to a circuit-QED Hamiltonian constructed from capacitors and Josephson junctions. However, two protected superconducting qubits exist, namely the $0$-$\pi$ qubit \cite{brooksProtectedGatesSuperconducting2013} and the current-mirror qubit by Kitaev \cite{kitaevProtectedQubitBased2006} which can be identified as rotor codes based on tessellating $\mathbb{RP}^2$ and a M\"{o}bius strip respectively, as will be shown in the next sections.

\subsection{The Four-Rotor Circuit}
\label{subsec:fourphase}
In this section, we introduce a fundamental circuit, shown in Fig.~\ref{fig:4phase}, that will be the building block for obtaining a rotor code face term with four rotors (with alternating signs) when we work in the regime $C \gg C_g$. The Hamiltonian of the circuit can also be interpreted as the approximate code Hamiltonian associated with a four-rotor Bacon-Shor code discussed in Subsec.~\ref{subsec:BS} encoding a logical rotor. One can provide analytical expressions of its spectrum, in full analogy with the Cooper-pair box spectrum \cite{cottetThesis, kochChargeinsensitiveQubitDesign2007}, as we will see. When we treat the Josephson junction terms perturbatively, we call this circuit element the four-phase gadget realizing an effective four-rotor face term.
 
\begin{figure}
    \centering
    \includegraphics[height=6cm]{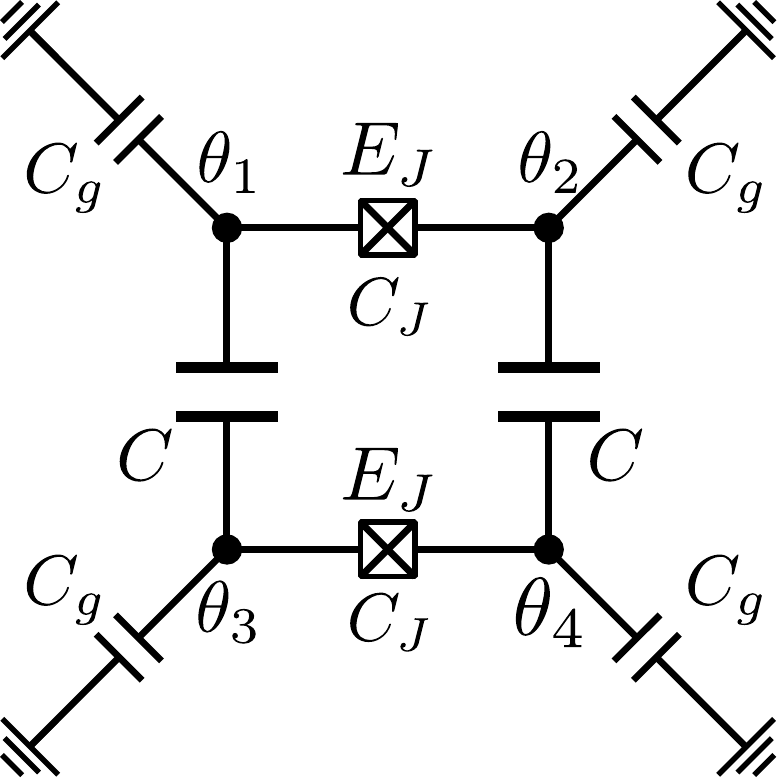}
    \caption{Electric circuit of the four-phase gadget (showing capacitances to ground).}
    \label{fig:4phase}
\end{figure}

For the four-rotor Bacon-Shor code, the group $\mathcal{G}$ is generated by
 \begin{align}
         \mathcal{G}=\left\langle e^{i \varphi(\hat{\ell}_1+\hat{\ell}_3)},e^{i \varphi'(\hat{\ell}_2+\hat{\ell}_4)}, e^{i(\hat{\theta}_2-\hat{\theta}_1)}, e^{i (\hat{\theta}_3-\hat{\theta}_4)}\;\middle\vert\;\varphi,\varphi'\in \mathbb{T}\right\rangle.
 \end{align}
We define the gauge rotor logicals in $\mathcal{G}$ as $\overline{X}_g(m)=e^{i m (\hat{\theta}_2-\hat{\theta}_1)}$ and $\overline{Z}_g(\phi)=e^{i \phi(\hat{\ell}_2+\hat{\ell}_4)}$. The stabilizer is generated by
    \begin{equation}
    \mathcal{S}=\left\langle S^X(1)\equiv e^{i( \hat{\theta}_1-\hat{\theta}_2-\hat{\theta}_3+\hat{\theta}_4)},
    S^Z(\phi)\equiv e^{i\varphi(\hat{\ell}_1+\hat{\ell}_2+\hat{\ell}_3+\hat{\ell}_4)}\;\middle\vert\; \forall \varphi \in \mathbb{T}\right\rangle.
    \end{equation}
    The encoded logical rotor has logical operators \[\overline{Z}_l(\phi)=e^{i \phi(\hat{\ell}_1+\hat{\ell}_2)},\; \overline{X}_l(m)=e^{i m(\hat{\theta}_2-\hat{\theta}_4)}.\]
     
   A basis for the four-rotor space is thus $\ket{\ell_l,\ell_g,\phi_x\in \mathbb{T},s_z\in \mathbb{Z}}$ where $\phi_x$ is the eigenvalue of 
   the $S^X(1)$ check, i.e. $S^X(1)=e^{i\phi_x}$,
   and $s_z=\ell_1+\ell_2+\ell_3+\ell_4$ is the syndrome of the $S^Z$-check. Here $\ket{\ell_g}$ is defined by $\overline{X}_g(m)\ket{\ell_g}=\ket{(\ell+m)_g}$, $\overline{Z}_g(\phi)\ket{\ell_g}=e^{i \phi \ell}\ket{\ell_g}$.
   The targeted (dimensionless) Hamiltonian, a special case of Eq.~\eqref{eq:HthinBS}, is
   \begin{align}
    H_{\rm 4-BS} =(\hat{\ell}_2+\hat{\ell}_4)^2+(\hat{\ell}_1+\hat{\ell}_3)^2-\cos(\hat{\theta}_1-\hat{\theta}_2)-\cos(\hat{\theta}_3-\hat{\theta}_4).
    \label{eq:4-BS}
    \end{align}

Looking at the circuit in Fig.~\ref{fig:4phase} and setting $C_J=0$ and neglecting capacitive terms $\sim 1/C$, we have the circuit-QED Hamiltonian

\begin{equation}
H_{C_J=0} \approx 2 E_{C_g} \bigl(\hat{\ell}_1 + \hat{\ell}_3 \bigr)^2 + 2 E_{C_g} \bigl(\hat{\ell}_2 + \hat{\ell}_4 \bigr)^2 -E_J \cos(\hat{\theta}_1 - \hat{\theta}_2) - E_J \cos(\hat{\theta}_3 - \hat{\theta}_4).
\end{equation}

For simplicity, we now assume $E_J=2 E_{C_g}$ so that we can remove the energy scales. We can also include the effect of the Josephson capacitances $C_J$ in this description, by using a first-order approximation in $C_J/C_g$ in a Taylor expansion of $\bs{C}^{-1}$. Applying this to the dimensionless Hamiltonian, we get an extension of Eq.~\eqref{eq:4-BS}, namely
\begin{align}
    H = (\hat{\ell}_2+\hat{\ell}_4)^2+(\hat{\ell}_1+\hat{\ell}_3)^2+\epsilon(\hat{\ell}_1+\hat{\ell}_3)(\hat{\ell}_2+\hat{\ell}_4)-\cos(\hat{\theta}_1-\hat{\theta}_2)-\cos(\hat{\theta}_3-\hat{\theta}_4),
    \label{eq:4-BS-ext}
    \end{align}
where $\epsilon=4 C_J/C_g$.
All terms in $H$ commute with $\mathcal{S}$ and the logical operators $\overline{Z}_l(\phi), \overline{X}_l(m)$, and thus $H$ has a degenerate spectrum with respect to the logical rotor. Each eigenlevel can additionally be labelled by the quantum numbers $\phi_x$ and $s_z$. It remains to consider the spectrum of $H$ with respect to the gauge logical rotor: this spectrum is identical to that of a Cooper-pair box in the presence of some off-set charge set by $s_z$ and $\epsilon$, and a flux-tunable Josephson junction with flux set by $\phi_x$. To derive this explicitly we use
\begin{align}
        \bra{\ell_1,\ell'_g,\phi_x,s_z}(\hat{\ell}_1+\hat{\ell}_3)^2 \ket{\ell_1,\ell_g, \phi_x, s_z}=\bra{\ell_1,\ell'_g,\phi_x,s_z}(s_z-\hat{\ell}_2-\hat{\ell}_4)^2 \ket{\ell_1,\ell_g, \phi_x, s_z},
    \end{align}
   and introduce the operators $e^{i\hat{\theta}_g}=e^{i(\hat{\theta}_1-\hat{\theta}_2)}$ and $\hat{\ell}_g=\hat{\ell}_2+\hat{\ell}_4$. Using a transformation $e^{i\hat{\theta}_g'}=e^{i(\hat{\theta}_g-\frac{\phi_x}{2})}$ and dropping primes and the subscript label $g$, Eq.~\eqref{eq:4-BS-ext} can be rewritten as
\begin{align}
H=(2-\epsilon)\left(-i\frac{\partial}{\partial \theta}-\frac{s_z}{2}\right)^2+\frac{s_z^2}{2}(1+\epsilon/2)-2\cos(\phi_x/2) \cos(\hat{\theta}).
     \label{eq:heff-g-ext}
     \end{align}
     In the sector labeled by $s_z$ and $\phi_x$, this Hamiltonian has a spectrum $E_n(s_z,\phi_x)$ and eigenfunctions $\psi_n(\theta)$ which are determined by the solutions of the Mathieu equation. The analysis of the Cooper-pair box, which has the transmon qubit as a particular case \cite{kochChargeinsensitiveQubitDesign2007}, is precisely done in this fashion \cite{cottetThesis}. For the Cooper-pair box, $s_z/2$ models an off-set charge and $\phi_x/2$ models the effect of an external flux in a flux-tunable Josephson junction. Unlike the Cooper-pair box, the spectrum also has a $s_z$-dependent shift and the groundstate sits in the sector with smallest $s_z=0$, see Fig.~\ref{fig:bs_levels}.
To convert to the standard form of the Mathieu equation, one defines $g_n(\theta)=e^{-2 is_z \theta} \psi_n(2\theta)$, (with $\psi(\theta+2\pi)=\psi(\theta)$, thus translating into a different boundary condition for $g_n(\theta)$), setting $\epsilon=0$ for simplicity:
\begin{align}
\left[-\frac{\partial^2}{\partial \theta^2}-4\cos(\phi_x/2)\cos(2\theta)+2s_z^2-2E_n(s_z,\phi_x)\right] g_n(\theta) =0
\end{align}
It is known that the spectrum of the effective Hamiltonian in Eq.~\eqref{eq:heff-g-ext}, depends only on the parity of the integer $s_z$ (besides the $s_z$-dependent shift) \cite{cottetThesis, kochChargeinsensitiveQubitDesign2007}. We show the first three energy levels for even and odd $s_z$ as a function of $\phi_x$ in Fig.~\ref{fig:bs_levels}.

\begin{figure}
    \centering
    \includegraphics{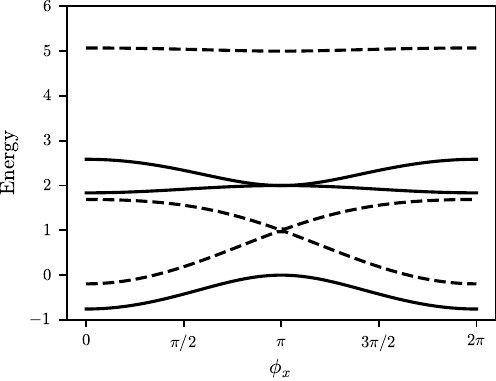}
    \caption{First three energy bands of the four-rotor Bacon-Shor code as a function of $\phi_x$ for even (solid) and odd (dashed) $s_z$ with $\epsilon=0$. The ground states are associated with even $s_z$ at $\phi_x=0$ and $\phi_x= 2 \pi$. When $\phi_x=\pi$, the Josephson term is zero, and the degeneracies between the first and second band at odd $s_z$, and second and third band at even $s_z$, are true degeneracies associated with the charging term.}
    \label{fig:bs_levels}
\end{figure}

The encoded logical rotor is protected when events which change the logical rotor state are (energetically) suppressed: such events are small fluctuations of phase and Cooper-pair jumps on each rotor, such that two rotors are affected as characterized by the weight-2 logical operators. Some of these events could happen in conjunction with the excitation of the gauge logical rotor or changes in stabilizer quantum number, in particular when these additional events cost little energy. If one encodes into the logical rotor at $s_z=0$ at $\phi_x=2\pi$, one observes from the spectrum in Fig.~\ref{fig:bs_levels} that there is a gap towards the second excited state (at $s_z=0$ and $\phi_x=2\pi$), and one is at the minimum for continuous variations in $\phi_x$, so this seems a good working point. In addition, there is a gap to the state $s_z=1$ at $\phi_x=2\pi$, suppressing processes which change $s_z$.

One could generalize this analysis of the four-rotor Bacon-Shor code to the $2N$-rotor Bacon-Shor code where there will be one eigenvalue $s_z=\sum_{j=1}\ell_j$ and $N-1$ eigenvalues $e^{i\phi_{x,j}}$ of the stabilizers $S_j^X$ in Eq.~\eqref{eq:stab-BS} as quantum numbers, besides $N-1$ gauge logical rotors and 1 encoded logical rotor. Again the Hamiltonian, corresponding to the circuit in Fig.~\ref{fig:BS} in the regime $C \gg C_g$ will only act on the gauge logical rotors ($\hat{\ell}_{g_j}=\hat{\ell}_j+\hat{\ell}_{j+N}$) as 
\[
H=\sum_{j=1}^{N-1}\hat{\ell}_{g_j}^2+(s_z-\sum_j \hat{\ell}_{g_j})^2-2\sum_{j=1}^{N-1}\cos(\phi_{x,j})\cos(\hat{\theta}_j-\hat{\theta}_{j+1}).
\]
We note that the degeneracy of the spectrum with respect to $\overline{Z}(\phi)=e^{i\phi \sum_{j=1}^N \hat{\ell}_j}$ for all $\phi$ is only lifted when connecting the last nodes on this strip in twisted fashion as is done in the M\"obius strip qubit, discussed in Section \ref{subsec:kitaevmirror}.

\subsubsection{The Four-Phase Gadget}

For the four-phase gadget, we study a particular regime of the circuit in Fig.~\ref{fig:4phase}: namely, the effect of the Josephson junctions is treated perturbatively (while the dependence on $C_J$ is kept in analytical form). In the language of the previous analysis this corresponds to considering the Cooper-pair box spectrum when $E_J \ll E_C$ and applying second-order perturbation theory with respect to the Josephson term which connects states of different charge. 

The capacitance matrix of the circuit in Fig.~\ref{fig:4phase} in terms of the nodes variables $\ell_k, k=1,\ldots, 4$ reads

\begin{equation}
\bm{C} = \begin{pmatrix}
    C_g + C + C_J & -C_J & -C & 0 \\
    - C_J & C_g + C + C_J & 0 & -C \\
    -C & 0 & C_g + C + C_J & -C_J \\
    0 & -C & -C_J & C_g + C + C_J
\end{pmatrix}.
\end{equation}

It is convenient to first introduce another set of variables, namely the left and right exciton $\theta_{L, R e}$ and agiton $\theta_{L, R a}$ variables defined as
\begin{equation}
\label{eq:excag}
\begin{pmatrix}
    \theta_{L e} \\ 
    \theta_{R e} \\ 
    \theta_{L a} \\
    \theta_{R a}
\end{pmatrix} = \underbrace{\begin{pmatrix}
    1 & 0 & -1 & 0 \\
    0 & 1 & 0 & -1 \\
    1 & 0 & 1 & 0 \\
    0 & 1 & 0 & 1
\end{pmatrix}}_{\bm{M}}
\begin{pmatrix}
    \theta_1 \\ 
    \theta_2 \\ 
    \theta_3 \\
    \theta_4
\end{pmatrix} \;\;,
    \begin{pmatrix}
    \ell_{L e} \\ 
    \ell_{R e} \\ 
    \ell_{L a} \\
    \ell_{R a}
\end{pmatrix}  = 
     \underbrace{\begin{pmatrix}
    \frac{1}{2} & 0 & -\frac{1}{2} & 0 \\
    0 & \frac{1}{2} & 0 & -\frac{1}{2} \\
    \frac{1}{2} & 0 & \frac{1}{2} & 0 \\
    0 & \frac{1}{2} & 0 & \frac{1}{2}
\end{pmatrix}}_{(\bm{M}^T)^{-1}}
   \begin{pmatrix}
    \ell_1 \\ 
    \ell_2 \\ 
    \ell_3 \\
    \ell_4
\end{pmatrix}. 
\end{equation}
These variables have been employed in Ref.~\cite{weissSpectrumCoherenceProperties2019} to study Kitaev's current mirror qubit \cite{kitaevProtectedQubitBased2006}. Note that the change of variables in Eq.~\eqref{eq:excag} does not represent a valid rotor change of variables as described in Appendix \ref{app:rotorchange} which would require the matrices $\bm{M}$ and $(\bm{M}^T)^{-1}$ to be integer matrices. This manifests itself, as one can see from Eq.~\eqref{eq:excag}, in that the exciton and agiton charges on the left $L$ (or right $R$) have either both integer or both half-integer eigenvalues.

The choice for exciton and agiton variables is motivated by block-diagonalizing the capacitance matrix $\bs{C}$, i.e. we define
$$
\bm{\tilde{C}} =(\bm{M}^T)^{-1} \bm{C} \bm{M}^{-1}=  \frac{1}{2}\begin{pmatrix}    
2C+C_g + C_J & -C_J & 0 & 0 \\
-C_J & 2C+C_g + C_J & 0 & 0 \\
0 & 0 & C_g + C_J & -C_J \\
0 & 0 & -C_J & C_g + C_J
\end{pmatrix}.
$$
Let the transformed charging energy matrix be
\begin{equation}
\bm{\tilde{E}}_C = \frac{e^2}{2} \bm{\tilde{C} }^{-1} = \begin{pmatrix}
\bm{E}_C^{(e)} & \bm{0} \\
\bm{0} & \bm{E}_{C}^{(a)}
\end{pmatrix},
\end{equation}
with the exciton and agiton sub-matrices given by
\begin{subequations}
\begin{equation}
\bm{E}_{C}^{(e)} = \frac{e^2}{(2C+C_g)(2C+C_g+2C_J)}\begin{pmatrix}
2C+C_g+C_J & C_J \\ 
C_J & 2C+C_g+C_J
\end{pmatrix},
\end{equation}
\begin{equation}
\bm{E}_{C}^{(a)} = \frac{e^2}{C_g(C_g+2C_J)} \begin{pmatrix}
C_g+C_J & C_J \\ 
C_J & C_g+C_J
\end{pmatrix}.
\end{equation}
\label{eq:ec-coef}
\end{subequations}
The Hamiltonian of the circuit in Fig.~\ref{fig:4phase} then equals
\begin{multline}
\label{eq:exaghamil}
    H_{4-\rm phase} =  4E_{C, 11}^{(e)} \bigl(\ell_{Le}^2 + \ell_{Re}^2 \bigr) + 8 E_{C, 12}^{(e)} \ell_{Le} \ell_{Re} +  4 E_{C, 11}^{(a)} \bigl(\ell_{La}^2 + \ell_{Ra}^2 \bigr) + 8E_{C, 12}^{(a)} \ell_{La} \ell_{Ra}+V
     \\  = 
    \frac{2e^2}{2C+C_g}(\ell_{L e}+\ell_{R e})^2+ \frac{2e^2}{2C+C_g+2C_J}(\ell_{L e}-\ell_{R e})^2+ \\
    \frac{2 e^2}{C_g}(\ell_{L a}+\ell_{R a})^2+ \frac{2 e^2}{C_g+2C_J}(\ell_{L a}-\ell_{R a})^2
    +V,
\end{multline}
using matrix entries $\bm{E}_{C,ij}^{(e/a)}$ and 
\begin{equation}
\label{eq:v4phase}
V=-2 E_J \cos \biggr[\frac{1}{2}\bigl(\theta_{La} - \theta_{Ra}\bigl) \biggr]\cos \biggr[\frac{1}{2}(\theta_{Le}- \theta_{Re}) \biggr].
\end{equation}

We show in Appendix \ref{app:SW4phase}
 that in the limit $C \gg C_J, C_g$ and when $E_J$ is much smaller than the characteristic energy of an agiton excitation, the four-phase gadget gives rise to an effective potential in the zero-agiton subspace given by

\begin{align}
V_{\rm eff}=-E_{J_{\rm eff}}\cos(\hat{\theta}_1+\hat{\theta}_4-\hat{\theta}_2-\hat{\theta}_3),
\label{eq:veff}
\end{align}

where the effective Josephson energy is given by

\begin{equation}
\label{eq:ejeff}
    E_{J_{\rm eff}}\equiv \frac{E_J^2}{4 E_{C,\rm diff}^{(a)}}.
\end{equation}
Here the typical energy of an agiton excitation $E_{C, \mathrm{diff}}^{(a)}$ is approximately given by, see Eq.~\eqref{eq:ecadiff} in Appendix \ref{app:SW4phase}, 

\begin{equation}
\label{eq:ecdiffa}
E_{C, \mathrm{diff}}^{(a)} \approx \frac{e^2}{ C_g + 2 C_J}. 
\end{equation}

For later convenience, we introduce the typical energy of a single exciton as
\begin{equation}
\label{eq:en1exc}
    E_{C}^{(e)} = 4 E_{C, 11}^{(e)} \approx 2 \frac{e^2}{C}.
\end{equation}

In the following subsections, we will see that the perturbative four-phase gadget is a fundamental building block of the $0$-$\pi$ qubit \cite{brooksProtectedGatesSuperconducting2013, gyenisMovingTransmonNoiseprotected2021, dempsterUnderstandingDegenerateGround2014} in Subsec.~\ref{subsec:zeropi} and of Kitaev's current mirror qubit \cite{kitaevProtectedQubitBased2006, weissSpectrumCoherenceProperties2019} in Subsec.~\ref{subsec:kitaevmirror}.

The crucial point here is that while the four-rotor Bacon-Shor encodes a logical rotor, the $0$-$\pi$ and current mirror circuits aim to encode a qubit. Note that if we were to treat the Josephson junctions perturbatively in the four-rotor Bacon-Shor Hamiltonian, we still encode a logical rotor. Hence, in order to encode a qubit with the four-rotor circuit, one uses inductors to `identify nodes' and effectively tessellate a non-orientable surface, as we will see.

	\subsection{The $0$-$\pi$ qubit as small Real Projective Plane Rotor Code}
 \label{subsec:zeropi}
 
	The $0$-$\pi$ qubit Hamiltonian in \cite{brooksProtectedGatesSuperconducting2013,paoloControlCoherenceTime2019, Groszkowski_2018, gyenisExperimentalRealizationProtected2021, dempsterUnderstandingDegenerateGround2014} can be seen as the smallest example of a quantum rotor code obtained by tessellating $\mathbb{RP}^2$. We start with the tessellation in Fig.~\ref{fig:0pia} (left) with two rotors (on the edges). 
The single face uses each edge twice so that the face constraint is $H_X=\left( 2 \,\, -2 \right)$. For later convenience, we label the vertices $1$ and $3$. The two vertex constraints are identical and equal $\hat{\ell}_1+\hat{\ell}_3=0$, or $H_Z=\left( 1 \,\, 1\right)$. The logical $\overline{Z}$ is a $\pi$-phaseshift on a single edge, say $\overline{Z}=Z_1(\pi)=e^{i \pi \hat{\ell}_1}$ and $\overline{X}=e^{i (\hat{\theta}_1-\hat{\theta}_3)}$. The dimensionless Hamiltonian associated with the code is

\begin{align}
    H= (\hat{\ell}_1+\hat{\ell}_3)^2 -\cos[2(\hat{\theta}_1-\hat{\theta}_3)].
    \label{eq:Htar}
\end{align}
	
	 \begin{figure}[ht]
		\centering
  \subfloat[]{\includegraphics[width=\linewidth]{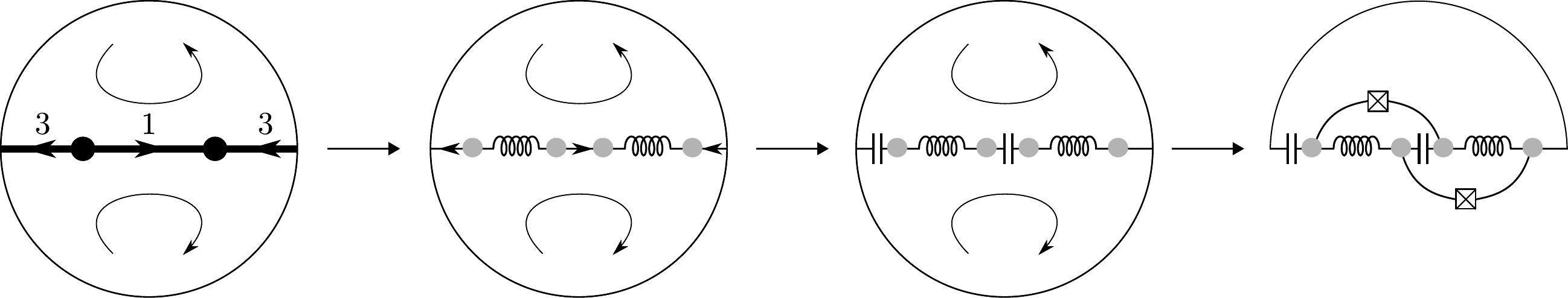}
  \label{fig:0pia}
  } \\
  \subfloat[]{\includegraphics[width=.3\linewidth]{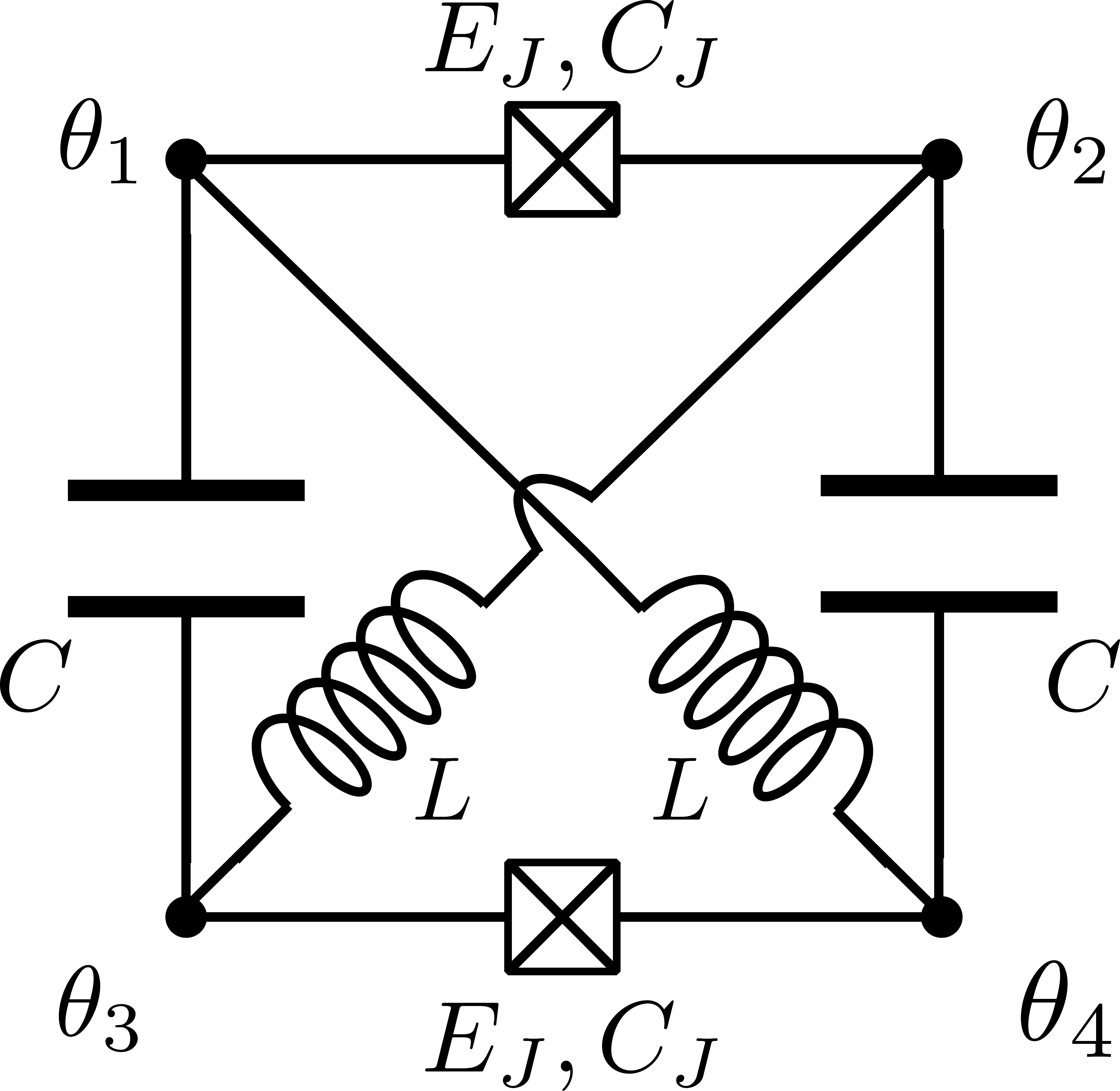}
  \label{fig:0pib}
  }
		\caption{(a) A very simple tiling of the real projective plane with two rotors and a representation of how it can be associated with the $0$-$\pi$ qubit circuit. Each node (on an edge) is split into two nodes (grey dots) and we put an inductor between them. Then two (large) capacitances are added on the edges to implement the capacitive `vertex stabilizers'. Finally, Josephson junctions are added to implement the `face stabilizer' perturbatively leading to the circuit in (b).
		(b) Final circuit of $0$-$\pi$ qubit which can be viewed as a circuit-QED realization of the code in (a). The $0$-$\pi$ qubit is essentially the four-phase gadget in Fig.~\ref{fig:4phase}, with inductances on the diagonal (capacitances $C_g$ to ground are not shown in the picture).}
		\label{fig:0pi}
	\end{figure}
    
    It is clear that the potential energy has two minima, one at $\theta_1-\theta_3=0$ and one at $\theta_1-\theta_3=\pi$.
    In particular, the (unphysical) logical codewords are
	\begin{align}
	\ket{\overline{0}} & \propto \sum_{\ell \in \mathbb{Z}} \ket{2\ell}_1\ket{-2\ell}_3, &
	\ket{\overline{+}} & \propto\int_{\mathbb{T}} d\theta  \ket{\theta}_1\ket{\theta}_3, \nonumber\\
	\ket{\overline{1}} & \propto \sum_{\ell \in \mathbb{Z}} \ket{2\ell+1}_1\ket{-2\ell-1}_3, & \ket{\overline{-}} & \propto \int_{\mathbb{T}} d\theta \ket{\theta}_1\ket{\pi+\theta}_3.
	\label{eq:codewords-simple}
	\end{align}
	
The question is how to obtain the second term in Eq.~\eqref{eq:Htar} using Josephson junctions. A first solution was conceived in Ref.~\cite{brooksProtectedGatesSuperconducting2013}, which introduced the circuit of the $0$-$\pi$ qubit. We note that engineering a $0$-$\pi$ circuit allowing an effective two-Cooper pair tunneling process is also the goal of \cite{smithSuperconductingCircuitProtected2020} which we do not analyze here.
In Fig.~\ref{fig:0pia} we sketch the idea of how the $0$-$\pi$ circuit comes about from a simple tessellation of the real projective plane, using the doubling of each rotor to two rotors coupled by an inductor. Fig.~\ref{fig:0pib} shows the final circuit of the $0$-$\pi$ qubit. We see that the circuit of the $0$-$\pi$ qubit in Fig.~\ref{fig:0pib} is the same as that of the four-phase gadget in Fig.~\ref{fig:4phase}, with the addition of inductances on the diagonal.  More specifically, the system is operated in the regime of $C \gg C_J, C_g$ and when the perturbative analysis in Appendix \ref{app:SW4phase} is valid. 

In what follows, we provide a simplified and intuitive explanation of why the circuit of the $0$-$\pi$ qubit gives rise to the code associated with the Hamiltonian in Eq.~\eqref{eq:Htar}, using the exciton and agiton variables introduced in Subsec.~\ref{subsec:fourphase}. We remark that in the literature a different change of variables is usually used to study the $0$-$\pi$ circuit \cite{dempsterUnderstandingDegenerateGround2014, Groszkowski_2018, gyenisExperimentalRealizationProtected2021}. However, the exciton-agiton picture lets one see the connection to the rotor code more clearly.

The impedance of an inductor in the Fourier domain is $Z_{L}(\omega) = i \omega L$, which means that at high frequencies the circuit behaves as an open circuit, while at low frequencies it functions as a short circuit. In the four-phase gadget circuit, when $C \gg C_J, C_g$ agiton excitations have high energy, i.e., high frequency. Thus, we expect intuitively that if the inductance is large enough, it will behave as an open circuit with respect to the agiton excitations, which are consequently only weakly affected. Mathematically, this requires that the typical energy of an agiton excitation $E_{C, \mathrm{diff}}^{(a)}$ given in Eq.~\eqref{eq:ecdiffa} satisfies 
$E_{C, \mathrm{diff}}^{(a)} \gg E_L$, with $E_L=\frac{\Phi_0^2}{4 \pi^2 L}$ the inductive energy of the inductor.
This means that the perturbative analysis in Appendix \ref{app:SW4phase} is still valid if the previous condition is satisfied and we expect that within the zero-agiton subspace the four-phase gadget effectively gives a potential as in Eq.~\eqref{eq:veff}.

On the other hand, we want the inductor to behave as a short circuit within the low-frequency, zero-agiton subspace. This condition requires that the energy of a typical exciton $E_{C}^{(e)}$ defined in Eq.~\eqref{eq:en1exc} and approximately equal to $\frac{2 e^2}{C}$ is much smaller than $E_L$. In this case, we can identify $\hat{\theta}_1\approx\hat{\theta}_4$, $\hat{\theta}_3\approx\hat{\theta}_2$, and obtain

\begin{equation}
\label{eq:zeropiveff}
V_{\mathrm{eff}} = -E_{J_{\rm eff}} \cos[2(\hat{\theta}_1 - \hat{\theta}_3)].
\end{equation}

	In what follows, we provide an analysis at the Hamiltonian level of the previous argument.  The classical Hamiltonian of the $0$-$\pi$ circuit, using the exciton and agiton variables, is
	\begin{equation}
	  H_{0-\pi}=H_{4-\rm phase}	+\frac{E_L}{4}(\theta_{L a}-\theta_{R a})^2
	+\frac{E_L}{4}(\theta_{R e}+\theta_{L e})^2.
	\end{equation}
 where $H_{4-\rm phase}$ is given in Eq.~\eqref{eq:exaghamil}. 

First, we note that the inductive terms do not couple excitons and agitons and the effect of $\frac{E_L}{4}(\hat{\theta}_{La}-\hat{\theta}_{Ra})^2$ is zero in the no-agitons subspace. This latter term induces a coupling between different agiton variables, which we do not expect to modify the perturbative analysis in Appendix \ref{app:SW4phase}, as long as $E_{C, \mathrm{diff}}^{(a)} \gg E_L$. 

	Next, in the zero-agiton subspace the idea is that the capacitive terms, which have characteristic energy scale $E_{C}^{(e)}$, are small relative to the inductive term.
    Then, the inductive term can enforce $\hat{\theta}_{Le} + \hat{\theta}_{Re} = \hat{\theta}_1- \hat{\theta}_3+ \hat{\theta}_2-\hat{\theta}_4=0$ which leaves $\hat{\theta}_{Le}-\hat{\theta}_{Re} = \hat{\theta}_1+\hat{\theta}_4-\hat{\theta}_2-\hat{\theta}_3 \approx 2(\hat{\theta}_1-\hat{\theta}_3)$, hence plugging this into Eq.~\eqref{eq:veff} we realize the term in Eq.~\eqref{eq:zeropiveff}. We note that this also implies $\hat{\ell}_{Le} = -\hat{\ell}_{Re} = \frac{1}{2}(\hat{\ell}_1 - \hat{\ell}_3)$. 
    Finally, we get the Hamiltonian
    \begin{equation}
    \label{eq:zeropieffexc}
    H_{0-\pi} \approx E_{C}^{(e)} (\hat{\ell}_{1} - \hat{\ell}_3)^2 - 2 E_{J_{\rm eff}} \cos[2 (\hat{\theta}_1 - \hat{\theta}_3)],
    \end{equation}
    where the eigenstates must satisfy the zero-agiton constraint

    \begin{equation}
        \hat{\ell}_1 + \hat{\ell}_3 = 0,
    \end{equation}
which effectively implements the $Z$-type stabilizer associated with the rotor code in Eq.~\eqref{eq:Htar}. Thus, in order for the eigenstates to represent the codewords of the rotor code, we need to neglect the charging term in Eq.~\eqref{eq:zeropieffexc} and require 

\begin{equation}
    E_{C}^{(e)} \ll E_{J_{\mathrm{eff}}}.
\end{equation}
 Let us also examine the wavefunctions obeying these constraints. The zero-agiton subspace restricts the wavefunctions of the four variables $\theta_1,\theta_2,\theta_3,\theta_4$ to be states of the form 
 \[
 \ket{\psi}=\sum_{\ell \in \mathbb{Z},\ell' \in \mathbb{Z}} \alpha_{\ell,\ell'} \ket{\ell}_1 \ket{\ell'}_2 \ket{-\ell}_3 \ket{-\ell'}_4.
 \]
Now consider that this state has to obey the inductive energy constraint, $\hat{\theta}_1- \hat{\theta}_3+ \hat{\theta}_2-\hat{\theta}_4=0$. This requires {\em at least} that the state is an approximate $+1$ eigenstate of $e^{i(\hat{\theta}_1-\hat{\theta}_3+\hat{\theta}_2-\hat{\theta}_4)}$, hence $\alpha_{\ell+1,\ell'+1}=\alpha_{\ell,\ell'}$.

Now to be at a minimum of the effective potential in Eq.~\eqref{eq:veff}, one also wishes to be a $+1$ eigenstate of $e^{i (\hat{\theta}_{L e}-\hat{\theta}_{R e})}=e^{i ( \hat{\theta}_1-\hat{\theta}_3-\hat{\theta}_2+\hat{\theta}_4)}$, which implies $\alpha_{\ell+1,\ell'-1}=\alpha_{\ell,\ell'}$, or $\alpha_{\ell,\ell'}=f(\ell+\ell')$ (only a function of the sum $\ell+\ell'$). 
 Thus, both constraints together imply that $\alpha_{\ell,\ell'}$ is some constant for even $\ell+\ell'$ (and even $\ell-\ell'$), and (another) constant for odd $\ell+\ell'$ (and odd $\ell-\ell'$). Such states can be made from a superposition of orthogonal basis states
\begin{align}
\ket{0} & =\sum_{\ell \in \mathbb{Z},k \in \mathbb{Z}} \ket{\ell}_1\ket{\ell+2k}_2\ket{-\ell}_3\ket{-\ell-2k}_4, \; \notag \\\ket{1} &=\sum_{\ell \in \mathbb{Z},k \in \mathbb{Z}} \ket{\ell}_1\ket{\ell+2k+1}_2\ket{-\ell}_3\ket{-\ell-2k-1}_4.
\end{align}

If we use Eq.~\eqref{eq:four} and Fourier sums, one can also write
\begin{align}
\ket{+} \propto \ket{0}+\ket{1} &=
\int d\theta \int d\theta'  \ket{\theta}_1\ket{\theta'}_2 \ket{\theta}_3 \ket{\theta'}_4,\; \notag \\ \ket{-} \propto \ket{0}-\ket{1}&=\int d\theta \int d\theta'  \ket{\theta}_1 \ket{\theta'}_2\ket{\theta+\pi}_3 \ket{\theta'+\pi}_4.
\label{eq:Xcode}
\end{align}
 We observe that $\ket{\pm}$ are product states between nodes 1-3 and 2-4, i.e. they are of the form $\ket{\pm}=\ket{\overline{\pm}}\ket{\overline{\pm}}$
where $\ket{\overline{\pm}}$ are the codewords in Eq.~\eqref{eq:codewords-simple}. Hence, the $0$-$\pi$ circuit represents the code in `doubled form' and we observe that due to this doubling the logical $\overline{Z}$ is now a $\pi$-phaseshift on two rotors $e^{i \pi(\hat{\ell}_1+\hat{\ell}_3)}$ while the logical $\overline{X}$ is the same as before.

We note that in this analysis we have connected each node to ground via a capacitance $C_g$ while in the usual analysis one of the nodes is grounded itself. That grounding choice would remove one degree of freedom, setting, say, $\theta=0$ in Eq.~\eqref{eq:Xcode}, so we get the code without doubling.
 
	\subsection{Kitaev's Current-Mirror Qubit as thin M\"{o}bius-strip Rotor Code}
 \label{subsec:kitaevmirror}

Kitaev's current mirror qubit \cite{kitaevProtectedQubitBased2006, weissSpectrumCoherenceProperties2019} has a direct interpretation as the quantum rotor code of the thin M\"{o}bius strip discussed in Section \ref{sec:thinmoeb} and Fig.~\ref{fig:moebius}. As stabilizer checks in the Hamiltonian we have
\begin{align}
\label{eq:moebius_xstab}
O_j^{c,X} & =  \cos(\hat{\theta}_j-\hat{\theta}_{N+j}+\hat{\theta}_{N+j+1}-\hat{\theta}_{j+1}), \quad j=1,\ldots,N-1,\; \notag \\
O_N^{c,X} & =  \cos(\hat{\theta}_N-\hat{\theta}_{2N}-\hat{\theta}_{m+1}+\hat{\theta}_{1}), \quad j=N.
\end{align}
The weight-2 vertex $Z$ checks are specified as 
\begin{equation}
\label{eq:moebius_zstab}
O_j^Z=\hat{\ell}_j+\hat{\ell}_{N+j}, \quad j=1,\ldots, N,  
\end{equation}

\begin{figure}[htb]
    \centering
\centering 
\includegraphics[height=3.8cm]{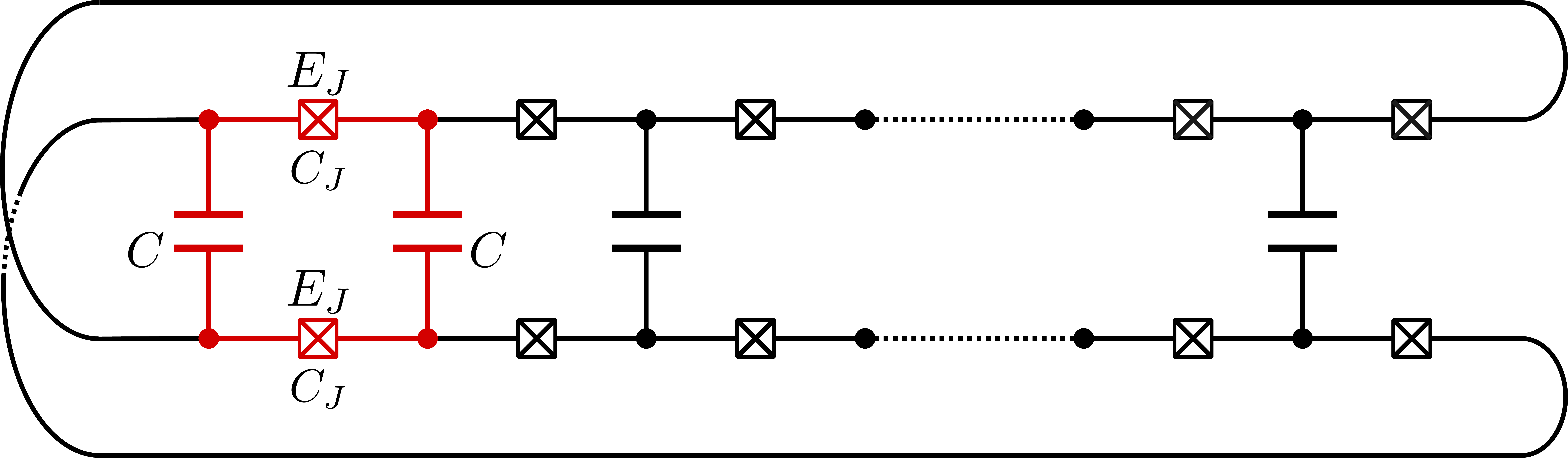}
\caption{Kitaev's current mirror qubit \cite{kitaevProtectedQubitBased2006, weissSpectrumCoherenceProperties2019}. The circuit can be viewed as a sequence of four-phase gadgets discussed in Subsec.~\ref{subsec:fourphase} (highlighted in red) in the limit $C \gg C_J, C_g$ that we analyze in detail in Appendix \ref{app:SW4phase}. Each node is assumed to have a capacitance $C_g$ to ground (not shown in the picture).}
\label{fig:current_mirror_qubit}
\end{figure}

 The key to understand this mapping is the perturbative analysis of the four-phase gadget circuit (in red in Fig.~\ref{fig:current_mirror_qubit}) as performed in \cite{weissSpectrumCoherenceProperties2019} and shown in Appendix~\ref{app:SW4phase}.
We refer the reader to Appendix~\ref{app:SW4phase} for all the details, while here we highlight the main features of the analysis:
\begin{itemize}
    \item In the limit of $C \gg C_J, C_g$ it is energetically unfavourable to have charges on the Josephson junction capacitance $C_J$ and/or on the ground capacitance $C_g$. This translates into the `no-agiton' constraints
        \begin{equation}
    \label{eq:zero_ag_cond}
        \hat{\ell}_{j} + \hat{\ell}_{m + j} = 0, \quad j=1, \dots, N,
    \end{equation}
    which is exactly the stabilizer constraint imposed by the weight-2, $Z$-type stabilizers in Eq.~\eqref{eq:moebius_zstab}. 
\item The Josephson potential associated with the two junctions in the four-phase gadget is treated perturbatively. The perturbative regime corresponds to the condition
\begin{equation}
\label{eq:moebius_param_1}
    \frac{E_J}{E_{C, \mathrm{diff}}^{(a)}} \approx \frac{E_J (C_g + 2 C_J)}{e^2}\ll 1,
\end{equation}
with $E_{C, \mathrm{diff}}^{(a)}$ defined in Eq.~\eqref{eq:ecdiffa}, under the additional assumption that
\begin{equation}
\label{eq:moebius_param_2}
\frac{E_{J, \mathrm{eff}}}{E_{C}^{(e)}} \approx \frac{E_J^2 C (C_g + 2 C_J) }{8 e^4} \gg 1,
\end{equation}
where we used the definition of the effective Josephson energy in Eq.~\eqref{eq:ejeff}, while $E_C^{(e)}$ is given in Eq.~\eqref{eq:en1exc}, respectively. In this regime, 
the four-phase gadgets in Kitaev's current mirror qubit in Fig.~\ref{fig:current_mirror_qubit} contribute terms in the Hamiltonian that are dominated by the purely inductive terms given by (see Eq.~\eqref{eq:h_eff_4phase} in the Appendix)
\begin{align}
\label{eq:moebius_4phase}
-E_{J_{\rm eff}} \cos (\hat{\theta}_j - \hat{\theta}_{N+j} + \hat{\theta}_{N+j + 1} - \hat{\theta}_{j+1}  ), & \quad j=1, \dots, N-1, \notag \\
-E_{J_{\rm eff}} \cos (\hat{\theta}_N - \hat{\theta}_{2N} - \hat{\theta}_{N+ 1} + \hat{\theta}_{1}  ), & \quad j=N,
\end{align}
while the charging terms are smaller in magnitude. The contributions in Eq.~\eqref{eq:moebius_4phase} are exactly the terms of the $X$-type stabilizers in Eq.~\eqref{eq:moebius_xstab}.  

\item Consequently, the (unphysical) codewords of the thin M\"{o}bius strip given in Eq.~\eqref{eq:moebiuscodewords} give an ideal representation of the degenerate ground states of Kitaev's current mirror qubit.
\end{itemize}

We remark that Eq.~\eqref{eq:moebius_param_2} is not in contradiction with Eq.~~\eqref{eq:moebius_param_1}. In fact, we can satisfy both constraints if we take a very large ratio $E_{C, \mathrm{diff}}^{(a)}/E_{C}$, which again means $C \gg C_{J}, C_g$. In practice, we will have residual charging energy terms that can be seen as a perturbation of the ideal rotor code Hamiltonian, and thus of the code subspace. 

As for the spectrum of the M\"obius strip qubit, we observe the following. If one takes the Bacon-Shor strip with its circuit in Fig.~\ref{fig:BS} and adds the Josephson junctions with a twist at the end to get the circuit in Fig.~\ref{fig:current_mirror_qubit}, then the degeneracy with respect to the encoded logical rotor of the Bacon-Shor code is lifted through the additional Josephson junctions. However, these rotor states still constitute many low-lying eigenstates (local minima) as discussed in \cite{weissSpectrumCoherenceProperties2019}.

    \subsection{Distance and Protection}
\label{sec:protection}

 In this subsection, we discuss in what way the distance of the rotor code determines the level of noise protection of the encoded qubit or rotor. It seems intuitive that the larger the $X$ and $Z$ distance, defined in Section \ref{sec:distance}, the more protected the encoded qubit or rotor should be. 
	
As a first comment, we note that
	 we have assumed that a superconducting island is characterized by (an excess) number of Cooper pairs $\ell$, but in practice any superconductor has a distribution of dynamically-active (single electron) quasi-particles, as excitations above the superconducting ground state. Rotor codes do not intrinsically protect against this source of noise, as it breaks the assumption of the rotor model.
 
	 Otherwise, noise in superconducting devices can physically originate from many sources, but can be classified as either charge noise, coupling to $\hat{\ell}$, or flux noise coupling to $e^{i\hat{\theta}}$. One can thus imagine that the code Hamiltonian in Eq.~\eqref{eq:codehamiltonian} is arrived at, imperfectly, as $H_{\rm code}+f(\{e^{i\hat{\theta}_i}\})+g(\{\hat{\ell}_i\})$ where $f()$ and $g()$ are (time-dependent) functions which are 1- or 2-local (or more generally O(1)-local) in the variables $e^{\hat{\theta}_i}$ resp. $\hat{\ell}_i$ (and which may involve other noncoding degrees of freedom).
  
A large $X$ distance $d_X$ suggests protection against flux noise, as the logical $\overline{X}$ corresponds to changing the number of Cooper pairs on a collection of superconducting islands, which, via the $O_j^Z$ charging constraints in the code Hamiltonian in Eq.~\eqref{eq:codehamiltonian}, costs energy. If such processes, modeled by $f(\{e^{i\hat{\theta}_i}\})$, are present in the Hamiltonian itself, then, if they are sufficiently weak to be treated perturbatively, their effect on the spectrum would only be felt in $d_X$-th order. This captures the idea of (zero-temperature) topological order which, for qubit stabilizer codes, has been proved rigorously \cite{BHS:topo}. 
Similarly, one can think about a large $Z$ distance $\delta_Z$ as protection against charge-noise processes, and these processes are modeled by the additional term $g(\{\hat{\ell}_i\})$ in the code Hamiltonian, for example representing fluctuating off-set charges. However, there is no direct argument for a gap in the spectrum with respect to these fluctuations, i.e. the spectrum of $O_j^{c/s,X}$ in Eq.~\eqref{eq:measX} varies continuously, and this continuous freedom is also expressed in the phenomenon of logical operator spreading suppressing the effective distance $\delta_Z$. Thus, we expect that the nature of a zero-temperature topological phase transition for these systems would be fundamentally different than in the qubit case. In particular, one may expect only a zero-temperature Kosterlitz-Thouless type transition for a 2D rotor code, and only a proper memory phase for the 3D rotor code in Section \ref{sec:HD}.

The nature of these phase transitions could be explicitly studied for a model of active error correction in which one (inaccurately) measures the stabilizer observables, and one can map the decoding problem onto a statistical mechanics problem, see e.g. \cite{DKLP, VAHWPT}. Our preliminary investigation shows that decoding the 2D rotor code undergoing continuous phase errors, --in a model of noiseless stabilizer measurements--, maps onto a (disordered) 2D XY model, which does not exhibit spontaneous magnetization due to the Mermin-Wagner theorem and can only undergo a finite-temperature Kosterlitz-Thouless phase-transition. Similarly, the decoding problem for a 2D rotor code undergoing discrete angular momentum errors can be related to a disordered 2D solid-on-solid model \cite{TD:super}. It will be interesting to relate properties of these phase transitions in these disordered models to properties of a possible memory phase. A 3D rotor code would similarly map to (disordered) 3D XY model and a 3D solid-on-solid model for which one would expect a genuine memory phase for sufficiently weak noise. Naturally, decoding the code when stabilizer measurements are inaccurate is the physically more relevant question and can alter the nature of these phases. \\

Let us now further examine the consequences of distance for the Hamiltonians of the $0$-$\pi$ qubit and the M\"obius-strip qubit in particular. For both qubits, $d_X$ is only 2, not growing with system size. If the qubits are operated in the presence of external flux, then the Hamiltonian will contain terms such as $-E_J \cos(\hat{\theta}_a-\hat{\theta}_b+\phi_{\rm ext}(t))\approx -E_J \cos(\hat{\theta}_a-\hat{\theta}_b)+E_J \phi_{\rm ext}(t)\sin(\hat{\theta}_a-\hat{\theta}_b)+O(\phi^2_{{\rm ext}}(t))$, directly proportional to logical $\overline{X}$ of the thin M\"oebius-strip code. Since the dependence on external flux (noise) only depends on a small $E_J$ in these circuits and one can work at a flux sweetspot to remove first-order effects, its effect on $T_1$ and $T_{\phi}$ may be small, but should however get progressively worse with $N$. Hence this qubit does not have flux noise protection beyond standard sweet-spot arguments. 

For the thin M\"obius-strip qubit, $\delta_Z$ scales with $N$ and this results in $T_{\phi}$ increasing with $N$ \cite{weissSpectrumCoherenceProperties2019}. The effect of charge noise on $T_1$ depends on transitions between the logical $\ket{\overline{0}}$ and $\ket{\overline{1}}$ via multiple perturbative steps to higher levels, i.e. an individual noise term in the Hamiltonian has the property that $\bra{\overline{0}} g(\{\hat{\ell}_i\})\ket{\overline{1}}=0$ as $g()$ only affects $O(1)$ rotors. However, the analysis in \cite{weissSpectrumCoherenceProperties2019} suggests that $T_1$ decreases with $N$, but this may be due to the authors defining $T_1$ by any transition to a low-lying excited state (instead of transitions from `a valley of states' around $\ket{\overline{0}}$ to a valley of states $\ket{\overline{1}}$). However, as the authors also note, transitions from $\ket{\overline{0}}$ to $\ket{\overline{1}}$ require stepping through $N$ multiple low-lying (energy scaling as $1/N$) excitations, --these excitations are the logical rotor encodings of the Bacon-Shor code--, and one may expect that $T_1=O(1)$ instead.
   \section{Discussion}
\label{sec:discussion}

    In this paper, we have introduced a formalism of quantum rotor codes and have explored how such codes can encode qubits via torsion of the underlying chain complex. We have discussed physical realizations of these codes in circuit-QED and we have pointed out some challenges and opportunities. Here we like to make a few concluding remarks. 
    
   
   It is not clear that a Hamiltonian encoding is preferred over realizing the code by active error correction, measuring face and vertex checks, or perturbatively, by measuring weight-2 checks by employing weakly-coupled circuits. It might be possible to continuously measure a total charge on two islands ($\hat{\ell}_i+\hat{\ell}_j$) and simultaneously weakly measure the current through a junction proportional to $\sin(\hat{\theta}_a-\hat{\theta}_b)$, to realize the four-rotor Bacon-Shor code, in analogy with continuous monitoring of the four-qubit Bacon-Shor code \cite{atalayaBaconShorCodeContinuous2017}. 

   On the topic of two-body measurements, recently new families of quantum error correction protocols, named Floquet codes, have attracted interest \cite{hastingsDynamicallyGeneratedLogical2021}.
   They use carefully chosen sequences of two-qubit measurements to stabilize a dynamical codespace. 
   The first schemes involved \(Y\) operators and present a periodic Hadamard transformation of the logical information, this hints that they most likely cannot work with quantum rotors.
   Some other schemes use only \(ZZ\)- or \(XX\)-measurements \cite{davydovaFloquetCodesParent2023, kesselringAnyonCondensationColor2022} and could possibly be adapted for quantum rotors.
   While it seems that the adaptation would work well on 2D orientable manifolds, it is less clear if it would work on non-orientable manifolds which is required in order to encode a logical qubit.

For a bosonic code such as the GKP code, a useful theory of approximate codewords has been developed in various papers. It would be of interest to develop numerically-useful approximations to explore rotor codes since the Hilbert space of multiple large rotor spaces grows rather quickly and limits numerics.
    
The Josephson junction has been invoked to contribute a $-\cos(\hat{\theta}_a-\hat{\theta}_b)$ term to the Hamiltonian, but the tunnel junction can also include some higher-order terms of the form $-\cos(k (\hat{\theta}_a-\hat{\theta}_b))$ for $k> 1$, which relate to the tunneling of multiple Cooper pairs \cite{willschObservationJosephsonHarmonics2023, golubovCurrentphaseRelationJosephson2004}. It is not clear whether such higher-order processes could help in engineering the targeted code terms more directly.

 As a separate observation, we note that the encoded qubit proposed in \cite{gladchenkoSuperconductingNanocircuitsTopologically2008} does not seem to be captured by the rotor code formalism. In \cite{gladchenkoSuperconductingNanocircuitsTopologically2008} one uses Josephson junctions and external fluxes $\phi_{\rm ext}$ through loops to generate terms of the form $-\cos(\hat{\theta}_a-\hat{\theta}_b+\phi_{\rm ext})$. Degeneracy in the spectrum comes about from classical frustration in these (commuting) potential terms, allowing to encode a qubit.

Another interesting avenue for further investigations is that of better constructions of quantum codes based on quasi-cyclic codes \cite{kovalevQuantumKroneckerSumproduct2013, panteleevDegenerateQuantumLDPC2021, panteleevQuantumLDPCCodes2022a} or modified products \cite{hastingsFiberBundleCodes2021a, breuckmannBalancedProductQuantum2021a, panteleevAsymptoticallyGoodQuantum2022a,leverrierQuantumTannerCodes2022}.
Notably it would be interesting to see if LDPC rotor code families with linear encoding rate as well as growing distances for both \(d_X\) and \(\delta_Z\) exists. 
The main difficulty we see is to be able to guarantee a lower bound on \(\delta_Z\).
In the examples we presented we had to degrade the other parameters of the codes to be able to prove a lower bound on \(\delta_Z\).

Finally some manifolds are equipped with a Riemannian metric, in such cases, a relation can be established between the distance of the codes based on this manifold and the size of the smallest non-trivial \(p\)-dimensional cycle according to the metric, called the \(p\)-\emph{systole}.
Some \(D\)-dimensional Riemannian manifolds can exhibit so-called \emph{systolic freedom}, meaning that the size of the \(p\)-systole and \(q\)-systole, for \(p+q=D\), multiply to more than the overall volume of the manifold.
Systolic freedom has been used to build the first qubit codes with larger than \(\sqrt{n}\) distance \cite{guthQuantumErrorCorrecting2014}.
Strikingly, this feature is more easily obtained with integer coefficients than \(\mathbb{Z}_2\) coefficients \cite{freedmanSystolicFreedom1999, freedmanZ2systolicFreedomQuantum2002}, and only recently more than logarithmic freedom was obtained over \(\mathbb{Z}_2\) \cite{freedmanBuildingManifoldsQuantum2021}.
An interesting question one can ask is whether systolic freedom over integer coefficients can be used to obtain quantum rotor codes with good parameters and what is the interplay between torsion and systolic freedom.
One complication needing further investigation is that contrary to the direct relation between the \(p\)-systole and the corresponding \(d_X\) the connection between the \(p\)-cosystole and \(\delta_Z\) is less straightforward due to the possibility of operator spreading.
\backmatter

\bmhead{Supplementary information}

Relevant code to reproduce the results of the paper is available at Ref.~\cite{githublink}. 

\bmhead{Acknowledgments}

  We thank Emmanuel Jeandel, Mac Hooper Shaw and David DiVincenzo for discussions.

\section*{Declarations}

\bmhead{Funding and/or Conflicts of interests/Competing interests}
C.V. acknowledges funding from the Plan France 2030 through the project NISQ2LSQ ANR-22-PETQ-0006. 

A.~C. acknowledges funding from the Deutsche
Forschungsgemeinschaft (DFG, German Research Foundation) under Germany’s Excellence Strategy – Cluster
of Excellence Matter and Light for Quantum Computing
(ML4Q) EXC 2004/1 – 390534769 and from the German Federal Ministry of Education and Research in the
funding program “quantum technologies – from basic research to market” (contract number 13N15585).

B.M.T. acknowledges support by QuTech NWO funding 2020-2024 – Part I “Fundamental Research”, project number 601.QT.001-1, financed by the Dutch Research Council (NWO). 

The authors declare no conflict of interest or competing interests.

\begin{appendices}

		\section{Quantum Rotor Change of Variables}
  \label{app:rotorchange}

	For continuous variable degrees of freedom representing, the $2n$ quadratures of, say $n$, oscillators, it is common to perform linear symplectic transformations \cite{VAHWPT} which preserve the commutation relations and which do not change the domain of the variables, namely $\mathbb{R}^{2n}$. For $n$ rotor degrees of freedom one cannot perform the same transformations and obtain a set of independent rotors.  
	Imagine we apply a transformation $\bs{A}$, with $\hat{\ell}'_j=\sum_k \bs{A}_{jk} \hat{\ell}_k$, defining new operators. This in turn defines an operator $Z'(\bs{\phi})=\prod_j e^{i\phi_j \hat{\ell}'_j}=\prod_j e^{i\phi_j \sum_k \bs{A}_{jk} \hat{\ell}_k}=\prod_k e^{i\sum_j \phi_j \bs{A}_{jk}\hat{\ell}_k}=Z(\bs{\phi}')$ with $\bs{\phi}'=\bs{A}^T \bs{\phi}$. We require that $\bs{\phi}'\in \mathbb{T}^n$, implying that $\bs{A}^T$ is a matrix with integer coefficients. Similarly, we can let $\hat{\theta}'_j=\sum_k \bs{B}_{jk} \hat{\theta}_k$, defining $X'(\bs{m})=X(\bs{m'})$ with $\bs{m'}=\bs{B}^T \bs{m}$. We also require that $\bs{m'} \in \mathbb{Z}^n$, which implies that $\bs{B}^T$ is a matrix with integer coefficients. In order to preserve the commutation relation in Eq.~\eqref{eq:commutationmulti}, 
	we require that
	\[
	\forall \bs{m},\bs{\phi}\;\;
	X'(\bs{m}) Z'(\bs{\phi})=e^{-i \bs{m} \bs{\phi}^T} Z'(\bs{\phi}) X'(\bs{m})=e^{-i \bs{m}'\bs{\phi}'^T} Z'(\bs{\phi}) X'(\bs{m}),
	\]
	which implies that $\bs{B} \bs{A}^T=\id$, or $\bs{A}=(\bs{B}^T)^{-1}$. Thus $\bs{A}$ needs to be a unimodular matrix, with integer entries, having determinant equal to $\pm 1$. An example is a Pascal matrix
\[
\bs{A}=\begin{pmatrix} 1 & 0 & 0 \\ 1 & 1 & 0 \\ 1& 2 & 1 \end{pmatrix}, \;\;
\bs{B}=\begin{pmatrix} 1 & -1 & 1 \\ 0 & 1 & -2 \\ 0 & 0 & 1 \end{pmatrix}. \\
\]	

	In the circuit-QED literature other coordinate transformations are used for convenience which are not represented as unimodular matrices. For example, for two rotors one can define the (agiton/exciton) operators 
	\[
	\hat{\theta}^{\pm}=\hat{\theta}_1 \pm \hat{\theta}_2,\; \hat{\ell}^{\pm}=\frac{1}{2}(\hat{\ell}_1\pm \hat{\ell}_2).
	\]
	In such cases it is important to understand that the new operators, even though the correct commutation rules are obeyed, do not represent an {\em independent} set of rotor variables.

	\section{Quantum Qudit Code Corresponding to a Quantum Rotor Code}
 \label{app:quditcode}
 First, recall the definition of the cyclic groups \(\mathbb{Z}_d\) and \(\mathbb{Z}_d^*\) in Eq.~\eqref{eq:Zstar} in Section \ref{sec:prelim}, where the group operation for \(\mathbb{Z}_d\) is addition modulo \(d\) and for \(\mathbb{Z}_d^*\) it is addition modulo \(2\pi\).
 Given a qudit dimension \(d\in\mathbb{N}^{\geq 2}\), the Hilbert space of a qudit, \(\mathcal{H}_{\mathbb{Z}_d}\), is defined by $d$ orthogonal states indexed by \(\mathbb{Z}_d\),
    \begin{equation}
        \forall j\in\mathbb{Z}_d,\quad\ket{j}\in\mathcal{H}_{\mathbb{Z}_d}.
    \end{equation}
    The qudit quantum states are
    \begin{equation}
        \ket{\psi} = \sum_{j=1}^d\alpha_j\ket{j},\quad\sum_{j=1}^d\vert\alpha_j\vert^2 = 1.
    \end{equation}
    The qudit generalized Pauli operators are given by

	\begin{align}
	    \forall k \in \mathbb{Z}_d, \quad & X_d(k)\ket{j}= \ket{j+k\!\pmod{d}}, \\
	    \forall \phi \in \mathbb{Z}_d^*, \quad & Z_d(\phi) \ket{j}= e^{i \phi j}\ket{j}.
	\end{align}
    By direct computation we have the following properties
	\begin{align}
	\id = X_d(0) &= Z_d(0),\label{eq:XZidqudit}\\
	X_d(k_1)X_d(k_2) &= X_d(k_1 + k_2\!\pmod{d}),\label{eq:Xaddqudit}\\
	Z_d(\phi_1)Z_d(\phi_2) &= Z_d(\phi_1+\phi_2\!\pmod{2\pi}),\label{eq:Zaddqudit}
	\end{align}
and the commutation relation
\begin{equation}
    	X_d(k)Z_d(\phi) = \e^{-i\phi k} Z_d(\phi)X_d(k).
    	\label{eq:paulicommutation-d}
\end{equation}    
The multi-qudit Pauli operators are defined by tensor product of single-qudit Pauli operators and labeled by tuples of \(\mathbb{Z}_d\) and \(\mathbb{Z}_d^*\) similarly to the rotor case.
If we have integer matrices \(H_X\) and \(H_Z\) as in Definition~\ref{def:rotorcode} we can define a qudit stabilizer code as follows.
\begin{defi}[Quantum Qudit Code, \(\mathcal{C}^d(H_X,H_Z)\)]
 Let $H_X$ and $H_Z$ be two integer matrices of size $r_x\times n$ and $r_z\times n$ respectively, such that
    \begin{equation}
    H_X H_Z^T = 0.
\end{equation}
We define the following group of operators, $\mathcal{S}$, and call it the stabilizer group:
\begin{equation}
    \mathcal{S}_d = \langle Z_d(\bs{\varphi}H_Z) X_d(\bs{s}H_X) \;\vert\;\forall \bs{\varphi}\in{\mathbb{Z}_d^*}^{r_z},\,\forall \bs{s}\in\mathbb{Z}_d^{r_x}\rangle.
\end{equation}
We then define the corresponding quantum qudit code, $\mathcal{C}^{d}(H_X,H_Z)$, on $n$ quantum rotors as the codespace stabilized by \(\mathcal{S}\):
\begin{equation}
    \mathcal{C}^{d}(H_X,H_Z) = \big\{\ket{\psi}\;\big\vert\; \forall P\in\mathcal{S}_d,\,P\ket{\psi} = \ket{\psi}\big\},
\end{equation}\label{def:quditcode}    
    \end{defi}
In other words: to go from the quantum rotor code to the qudit code we restrict the phases to \(d^{\rm th}\) roots of unity and pick the \(X\) operators modulo \(d\). A general definition of \eczoo[qudit stabilizer codes]{qudit_stabilizer} (beyond the `CSS' codes described here) can be found at the error correction zoo.

    We can relate logical operators from a quantum rotor code and the corresponding qudit code.
    Pick a non-trivial logical $X$ operator for \(\mathcal{C}^{\rm rot}(H_X,H_Z)\), that is to say, let 
    \begin{equation}
        \overline{X}(\bs{m}) = X(\bs{m}L_X + \bs{s}H_X).
    \end{equation}
    The following then holds: 
  \begin{align}
         0&=\left(\bs{m}L_X+\bs{s}H_X\right)H_Z^T \label{eq:logXrot}\\
        \not\exists \bs{s}^\prime,\quad \bs{s}^\prime H_X &= \bs{m}L_X+\bs{s}H_X.\label{eq:non-trivialXrot}
    \end{align}
    Eq.~\eqref{eq:logXrot} still holds if taken modulo \(d\) so \(X_d(\bs{m}L_X+\bs{s}H_X)\) is a valid logical \(X\) operator for \(\mathcal{C}^d(H_X,H_Z)\).
    On the other hand, Eq.~\eqref{eq:non-trivialXrot} might not hold when considering things modulo \(d\).
    This is the case for instance in the projective plane where taking \(d=3\) makes the logical operator become trivial since \(-2=1\!\pmod3\).
    Note that these cases only occur for logical operators in the torsion part of the homology.

    For the \(Z\) logical operators we can get a logical operator from the rotor code to the qudit code if and only if the logical operator is generated with \(\mathbb{Z}_d^*\) phases.
    \begin{equation}
     \bs{\phi},\bs{\nu}\in{\mathbb{Z}_d^*}^k,\quad\overline{Z}(\bs{\phi}) = Z(\bs{\phi}L_Z + \bs{\nu}H_Z).
    \end{equation}
    In this case we have that
    \begin{align}
0&=\left(\bs{\phi}L_Z+\bs{\nu}H_Z\right)H_X^T, \label{eq:logZrot}\\
        \not\exists \bs{\nu}^\prime\in\mathbb{T}^{r_z},\quad \bs{\nu}^\prime H_Z &= \bs{\phi}L_Z+\bs{\nu}H_Z,\label{eq:non-trivialZrot}
    \end{align}
    and Eq.~\eqref{eq:logZrot} and Eq.~\eqref{eq:non-trivialZrot} will remain valid when restricting phases to \(\mathbb{Z}_d^*\) (note that addition is modulo $2\pi$ in Eqs.~\eqref{eq:logZrot} and \eqref{eq:non-trivialZrot}).
    That is to say, some \(Z\) logical operators present in the quantum rotor code can be absent in the qudit code if these logical operators come from the torsion part, from some \(\mathbb{Z}_p\) and \(d\) does not divide \(p\).

    \section{Bounding the \(Z\) Distance from Below}
    \label{sec:minimizationZ}
This appendix proves Eq.~\eqref{eq:app-bound} in the proof of Lemma \ref{lem:Zbound}.    
    Given \(\bs{m}\in\Delta_X\), with $n=|\bs{m}|$, we want to lower bound the following quantity from Eq.~\eqref{eq:minimizationZ}:
    \begin{equation}
        \min_{\substack{\bs{\phi}\in\mathbb{T}^n,\, k\in\mathbb{Z}\\ \bs{m}\cdot\bs{\phi}=\alpha+2k\pi}} W_Z(\bs{\phi}) =  \min_{\substack{\bs{\phi}\in\mathbb{T}^n,\, k\in\mathbb{Z}\\ \bs{m}\cdot\bs{\phi}=\alpha+2k\pi}} \sum_{j=1}^n\sin^2\left(\frac{\phi_j}{2}\right).
    \end{equation}
    We define the Lagrangian function
    \begin{equation}
        \mathcal{L}(\bs{\phi},\lambda) = \sum_{j=1}^n\sin^2\left(\frac{\phi_j}{2}\right) - \lambda\left(\bs{m}\cdot\bs{\phi}-\alpha-2k\pi\right).
    \end{equation}
    We can compute the partial derivatives of \(\mathcal{L}\)
    \begin{align}
        \forall j,\;\frac{\partial\mathcal{L}}{\partial\phi_j} &= \frac{1}{2}\sin(\phi_j) - m_j\lambda,\\
        \frac{\partial\mathcal{L}}{\partial\lambda} &= \alpha+2k\pi - \bs{m}\cdot\bs{\phi}.
    \end{align}
    According to the Lagrange multiplier theorem \cite{book:arfkenweber} there exists a unique \(\lambda^*\) for an optimal solution \(\bs{\phi}^*\) such that
        \begin{equation}
        \left\{\begin{matrix}
            \forall j,\; \sin(\phi_j^*)=2m_j\lambda^*,\\[1em]
      \bs{m}\cdot{\bs{\phi}^*} =\alpha+2k\pi.
        \end{matrix}\right.
        \end{equation}
        For the optimal solution we have, for each \(j\), two options for \(\phi_j^*\):
    \begin{equation}
    \phi_j^* = \left\{\begin{matrix}
        \arcsin\left(2m_j\lambda^*\right),\\
        \pi - \arcsin\left(2m_j\lambda^*\right).
    \end{matrix}\right.
    \end{equation}
    We set a binary vector \(\bs{x}\) with \(x_j=0\) if it is the first case and \(x_j=1\) if it is the second.
    We deduce that
    \begin{align}
        \alpha+2k\pi &= \sum_{j=1}^n m_j\phi_j^*\\
        &= \sum_{j=1}^n (-1)^{x_j}m_j\arcsin\left(2m_j\lambda^*\right) +w\pi \quad (w = W_H(\bs{x})).       
        \end{align}
with Hamming weight $W_H()$.
       From now on we restrict ourselves to the case where \(m_j=\pm 1\).
    
    \begin{align}
            \Rightarrow\;\frac{\alpha+(2k - w)\pi}{n-2w} &=  \arcsin\left(2\lambda^*\right).
    \end{align}
    Let's compute the objective function
    \begin{align}
        W_Z(\bs{\phi}^*)&=\sum_{j=1}^n\sin^2\left(\frac{\phi_j^*}{2}\right) \\
        &= (n-w)\sin^2\left(\frac{\alpha+(2k - w)\pi}{2(n-2w)}\right) +w\cos^2\left(\frac{\alpha+(2k - w)\pi}{2(n-2w)}\right)\label{eq:objcomputesin2cos2}\\
        &=\frac{(n-w)}{2}\left(1-\cos\left(\frac{\alpha+(2k - w)\pi}{n-2w}\right)\right) +\frac{w}{2}\left(1+\cos\left(\frac{\alpha+(2k - w)\pi}{n-2w}\right)\right)\\
        &= \frac{n}{2} + \frac{2w-n}{2}\cos\left(\frac{\alpha+(2k - w)\pi}{n-2w}\right).\label{eq:objcomputesend}
    \end{align}
     We can assume that 
    \begin{equation}
    w < \frac{n}{2},
    \end{equation}
    as the other case behaves symmetrically and the case \(w=n/2\) always yields \(W_Z(\bs{\phi}^*) = n/2\) (see Eq.~\eqref{eq:objcomputesin2cos2}).
    Minimizing over \(w\) and \(k\) will give the lower bound.
    From Eq.~\eqref{eq:objcomputesend} we always have
        \begin{equation}
        W_Z(\bs{\phi}^*) \geq  w .
    \end{equation}
    We consider the asymptotic regime when \(n\rightarrow\infty\).
    If \(w\) grows with \(n\) then \(W_Z(\bs{\phi}^*)\) has to grow as well.
    Therefore choosing \(w\) constant can (and indeed will) yield a smaller value for \(W_Z(\bs{\phi}^*)\).
    Hence we fix \(w\) to a constant from now on.
    We can always put each \(\phi_j^*\) in the range \([-\pi,\pi)\) so that \(k\) is in the range \([-n/2,n/2]\).
    To minimize \(W_Z(\bs{\phi}^*)\) one needs to make 
    \begin{equation}
        \frac{\alpha + (2k-w)\pi}{n-2w} \sim 0,
    \end{equation}
    for which the best choice of \(k\in[-n/2,n/2]\) is \(k=0\).
    In fact \(w=0\) is also the best choice, since
    \begin{align}
        W_Z(\bs{\phi}^*) &= \frac{n}{2} + \frac{2w-n}{2}\cos\left(\frac{\alpha+(2k - w)\pi}{n-2w}\right)\\
        &= \frac{n}{2} + \frac{2w-n}{2}\cos\left(\frac{\alpha- w\pi}{n-2w}\right)\\
        &=w + \frac{1}{4n}\left(\alpha -\pi w\right)^2+o\left(\frac{w^2}{n}\right).
    \end{align}

Therefore, with $w=0$, we have
  \begin{align}
        W_Z(\bs{\phi}^*) &= \frac{n}{2}\left(1-\cos\left(\frac{\alpha}{n}\right)\right) =\frac{\alpha^2}{4n}+o\left(\frac{1}{n}\right).
    \end{align}

    \section{Orientability and Single Logical Qubit at the \((D-1)\) Level}
    \label{sec:orient=Z2}
    We consider the case of a rotor code \(\mathcal{C}^{\rm rot}(H_X,H_Z)\) where the entries in \(H_X\) and \(H_Z\) are taken from \(\{-1, 0, 1\}\) and \(H_X\) has the additional property that each column contains exactly 2 non-zero entry.
    This corresponds to the boundary map from \(D\) to \(D-1\) in the tessellation of a closed \(D\)-dimensional manifold (in \(2D\) edges are adjacent to exactly 2 faces, in \(3D\) faces are adjacent to exactly 2 volumes etc.\ldots).
    We also assume that the bipartite graph obtained by viewing \(H_X\) as an adjacency matrix (ignoring the signs) is connected.
    This corresponds to the \(D\)-dimensional manifold being connected.
    
    We want to consider the possibility that there is some \(\mathbb{Z}_p\) torsion, for an integer \(p\geq2\), at the \((D-1)\)-level in the corresponding chain complex.
    For this it needs to be the case (see Eq.~\eqref{eq:weakboundary}) that there exist \(\bs{s}\in\mathbb{Z}^{r_x}\) such that
    \begin{equation}
        \bs{s}H_X = p\bs{v},
    \end{equation}
        and such that
        \begin{equation}
            \bs{v}\not\in\im(H_X).\label{eq:properweakboundary}
        \end{equation}
    For any entry of \(p\bs{v}\), say \(j\), we have that
    \begin{equation}
        pv_j = \pm s_k \pm s_l,
    \end{equation}
    for some \(k\) and \(l\).
    Since the manifold is connected this implies that all entries of \(\bs{s}\) have the same residue modulo \(p\), say \(r\in\{1,\ldots,p-1\}\) (\(r\) cannot be zero without contradicting Eq.~\eqref{eq:properweakboundary}).
    In turns this means that 
    \begin{equation}
        p\bs{v} = 2r\bs{u} + p\bs{w},
    \end{equation}
    where \(\bs{u}\) is a vector with entries in \(\{-1,0,1\}\) and \(\bs{w}\) some integer vector.
    We conclude from this that \(p\) is necessarily even, so \(p=2q\).
    Hence we have
    \begin{equation}
        q\left(\bs{v}-\bs{w}\right) = r\bs{u}.
    \end{equation}
    This implies in turn that \(q\) divides \(r\), say \(r=tq\).
    Although we have \(tq=r<p=2q\) hence 
    \begin{equation}
    r=q.
    \end{equation}
    Which yields
    \begin{equation}
        q\bs{u}H_X = 2q\bs{v}^\prime,
    \end{equation}
    where \(\bs{v}^\prime\) differs from \(\bs{v}\) by the boundary of the quotient of \(\bs{s}\) by \(p\).
    We see that \(\bs{v}\) has in fact order 2 as the \(q\) simplifies and so we can fix \(q=1\) and \(p=2\).
    This leaves the only the possibility of \(r=1\) and so \(\bs{s}\) has only odd entries.
    
    Finally any \(\mathbb{Z}_2\) torsion generator \(\bs{v}_1\) and \(\bs{v}_2\) are related by a boundary since the corresponding \(\bs{s}_1\) and \(\bs{s}_2\) (having each only odd entries) sum to a vector with even entries so 
    \begin{equation}
    \frac{\bs{s}_1 + \bs{s}_2}{2}H_X = \bs{v}_1 + \bs{v}_2.    
    \end{equation}

    This concludes the proof that any finite tessellation of a connected closed \(D\)-dimensional manifold has either \(\mathbb{Z}_2\) or no torsion in the homology group \(H_{D-1}(\mathcal{M},\mathbb{Z})\) at the \((D-1)\)-level. 
    In particular rotor codes obtained from a \(2D\) manifold will have at most one logical qubit and some number of logical rotors.\\

    We remark that in this picture the manifold is orientable if there is a choice of signs for the \(D\) cells \(\bs{s}\in\{-1,1\}\) such that 
    \begin{equation}
        \bs{s}H_X = \bs{0},
    \end{equation}
    in which case there is no torsion since \(\bs{v}\) is necessarily trivial.

    \section{Products of chain complexes}
    \label{app:products}
In this Appendix we detail the construction of quantum rotor codes from the product of chain complexes given in Section~\ref{sec:productconstruction}.
We explore three settings, the first to encode logical rotors, the next two to encode qudits.

We use this construction on two integer matrices seen as chain complexes of length 2.
Take two arbitrary integer matrices \(\partial^\mathcal{C}\in\mathbb{Z}^{m_C\times n_C}\) and \(\partial^\mathcal{D}\in\mathbb{Z}^{n_D\times m_D}\).
They can be viewed as boundary maps of chain complexes
\begin{equation}
    \begin{matrix}
    \mathcal{C}:\;& \mathbb{Z}^{m_C}&\xrightarrow{\partial^\mathcal{C}}&\mathbb{Z}^{n_C} & &  \\[.6em]
    &H_1(\mathcal{C}) = \ker(\partial^\mathcal{C})&&H_0(\mathcal{C}) = \mathbb{Z}^{n_C}/\im(\partial^\mathcal{C})\\[1em]
    \mathcal{D}:\;& \mathbb{Z}^{n_D}&\xrightarrow{\partial^\mathcal{D}}&\mathbb{Z}^{m_D}\\[.6em]
    &H_1(\mathcal{D}) = \ker(\partial^\mathcal{D})&&H_0(\mathcal{D}) = \mathbb{Z}^{m_D}/\im(\partial^\mathcal{D}),
    \end{matrix}
\end{equation}
where the homology groups are given in the second row.
Taking the product \(\mathcal{C}\otimes\mathcal{D}\) will give a chain complex of length 3 characterized by two matrices \(H_X\) and \(H_Z\) 
\begin{equation}
    \mathcal{C}\otimes\mathcal{D}:\quad \mathbb{Z}^{m_Cn_D}\quad\xrightarrow{H_X}\quad\mathbb{Z}^{n_Cn_D + m_Cm_D}\quad\xrightarrow{H_Z^T}\quad\mathbb{Z}^{n_Cm_D}.
\end{equation}
The matrices \(H_X\) and \(H_Z\) can be written in block form as
\begin{align}
    H_X &= \begin{pmatrix}
        \partial^\mathcal{C}\otimes\id_{n_D} & -\id_{m_C} \otimes \partial^\mathcal{D}
    \end{pmatrix},\\
    H_Z &= \begin{pmatrix}
        \id_{n_C}\otimes {\partial^\mathcal{D}}^T & {\partial^\mathcal{C}}^T \otimes \id_{m_D}
    \end{pmatrix}.
\end{align}
The homology in the middle level, corresponding to the rotor code logical group, can be characterized according to the Künneth formula in Eq.~\eqref{eq:kunneth} as
\begin{align}
    H_1(\mathcal{C}\otimes\mathcal{D}) \simeq &
    \left(\ker(\partial^\mathcal{C})\otimes \mathbb{Z}^{m_D}/\im({\partial^\mathcal{D}})\right)\nonumber\\
    &\oplus \left(\mathbb{Z}^{n_C}/\im(\partial^\mathcal{C})\otimes \ker({\partial^\mathcal{D}})\right)\nonumber\\
    &\oplus{\rm Tor}\left(\mathbb{Z}^{n_C}/\im(\partial^\mathcal{C}), \mathbb{Z}^{m_D}/\im({\partial^\mathcal{D}})\right).\label{eq:homprod}
\end{align}
\subsection{Free+Free Product: Logical Rotors}
By picking, for \(\partial^\mathcal{C}\) and \({\partial^\mathcal{D}}\), matrices which are full-rank, rectangular, with no torsion and with\begin{align}
    k_{C} &= n_{C} - m_C>0,\\
    k_{D} &= n_{D} - m_D>0,
\end{align}
we can ensure that 
\begin{equation}
    \mathbb{Z}^{n_C}/\im(\partial^\mathcal{C})=\mathbb{Z}^{k_C},\qquad\ker({\partial^\mathcal{D}}) = \mathbb{Z}^{k_D},
\end{equation}
and, that all other terms in Eq.~\eqref{eq:homprod} vanish, such that we get, using Eq.~\eqref{eq:tensor}
\begin{equation}
    H_1(\mathcal{C}\otimes\mathcal{D}) = \mathbb{Z}^{n_C}/\im(\partial^\mathcal{C})\otimes \ker({\partial^\mathcal{D}}) = \mathbb{Z}^{k_Ck_D},
\end{equation}
that is to say, a rotor code encoding logical rotors. This configuration is the common one when using the hypergraph product to construct qubit codes. To obtain the \(X\) logical operators we can use the following matrix, \(L_X\), acting only on the first block of rotors:
\begin{equation}
    L_X = \begin{pmatrix}
        E_C \otimes G_D & 0\label{eq:freefreeLX}
    \end{pmatrix},
\end{equation}
where \(G_D\) and \(E_C\) are the generating matrices for \(\ker(\partial^\mathcal{D})\) and \(\mathbb{Z}^{n_C}/\im(\partial^\mathcal{C})\) respectively. It is straightforward to check that \(H_ZL_X^T = 0\) by definition of \(G_D\) and that \(H_X\) cannot generate the rows of \(L_X\) by definition of \(E_C\).
To obtain the \(Z\) logical operators we can use the following matrix, \(L_Z\), also acting only on the first block of rotors: 
\begin{equation}
    L_Z = \begin{pmatrix}
        G_C \otimes E_D & 0
    \end{pmatrix},\label{eq:freefreeLZ}
\end{equation}
where \(G_C\) and \(E_D\) are the generating matrices for \(\ker\left({\partial^\mathcal{C}}^T\right)\) and \(\mathbb{Z}^{n_D}/\im\left({\partial^\mathcal{D}}^T\right)\) respectively.
It is straightforward to check that \(H_XL_Z^T = 0\) by definition of \(G_C\) and that \(H_Z\) cannot generate the rows of \(L_Z\) by definition of \(E_D\).

If \(\partial^{\mathcal{D}}\) is a full-rank parity check matrix of a classical binary code with minimum distance \(d^{\mathcal{D}}\), we get the following lower bound on the \(X\) distance
\begin{equation}
    d_X^{\mathcal{C}\otimes\mathcal{D}}\geq d^\mathcal{D},
\end{equation}
using Theorem~\ref{thm:boundXdist} and the known lower bound of the \(X\) distance of the qubit code corresponding to the hypergraph product obtained by replacing all the rotors by qubits \cite{TZ:hypergraph}.
So we have for parameters
\begin{equation}
    \left\llbracket n_Cn_D+m_Cm_D, (k_Ck_D,0), (\Omega(d^\mathcal{D}),O(d^\mathcal{C}))\right\rrbracket_{\rm rot}.
\end{equation}

The weight of the logical \(Z\) operators in Eq.~\eqref{eq:freefreeLZ} is \(O(d^\mathcal{C})\) but they could be spread around using \(Z\) stabilizers.
To bound the \(Z\) distance using Lemma~\ref{lem:Zbound} one can try to get a large set of disjoint logical \(X\) operators by considering sets of the type
\begin{equation}
    \Delta_X(\bs{e}_i^\mathcal{C}\otimes\bs{g}_j^\mathcal{D}) = \left\{\begin{pmatrix}
        \left(\bs{e}_i\oplus \partial^\mathcal{C}(\bs{s})\right)\otimes\bs{g}_j^\mathcal{D} & 0
    \end{pmatrix}\;\vert\;\bs{s}\in S\subset\mathbb{Z}^{m_C}\right\},\label{eq:candidateDeltaHP}
\end{equation}
where \(\bs{e}_i^\mathcal{C}\) is the \(i\)th row of \(E_C\) and \(\bs{g}_j^\mathcal{D}\) the \(j\)th row of \(G_D\) and the set \(S\) is the largest subset of \(\mathbb{Z}^{m_C}\) which generates elements \(\bs{e}_i\oplus \partial^\mathcal{C}(s)\) that are pair-wise disjoints.
The best to hope for is a set \(S\) of size linear in \(n_C\) for which we would have per Lemma~\ref{lem:Zbound}
\begin{equation}
    \delta_Z = \Omega\left(\frac{n_C}{d^\mathcal{D}}\right).
\end{equation}
This is guaranteed for instance if \(\partial^\mathcal{C}\) is the parity check matrix of the repetition code.
Very similar to the Möbius and cylinder rotor codes in Sections \ref{sec:thinmoeb} and \ref{sec:cyl}, one needs to `skew the shape' in order to have a growing distance \(\delta_Z\).
So choosing for \(\partial^\mathcal{C}\) the standard parity check matrix of the repetition code with size \(n_C = \Theta(n_D^2)\) and a good classical LDPC code for \(\partial^\mathcal{D}\), with \(d^\mathcal{D}=\Theta(n_D)\) and \(k_D=\Theta(n_D)\), one would have as parameters of the \(\mathcal{C}\otimes\mathcal{D}\) code 
\begin{equation}
    \left\llbracket n, \left(\Theta(\sqrt[3]{n}),0\right) ,\left(\Theta(\sqrt[3]{n}),\Theta(\sqrt[3]{n})\right)\right\rrbracket_{\rm rot}.
\end{equation}

To provide a concrete example, we can pick the \([7,4,3]\) Hamming code with parity check matrix, \(H\), generator matrix, \(G\), and generator matrix for the complementary space of the row space of the parity-check matrix, \(E\), given by
\begin{equation}
    H = \begin{pmatrix}
        1 & 1 & 1 & 0 & 0 & 1 & 0\\
        0 & 1 & 1 & 1 & 0 & 0 & 1\\
        1 & 0 & 1 & 1 & 1 & 0 & 0
    \end{pmatrix},\qquad G = \begin{pmatrix}
        1& 0& 0& 0& -1& -1& 0\\
        0& 1& 0& 0& 0& -1& -1\\
        0& 0& 1& 0& -1& -1& -1\\
        0& 0& 0& 1& -1& 0& -1
    \end{pmatrix},\qquad E = \begin{pmatrix}
        1 & 0 & 0 & 0 & 0 & 0 & 0\\
        0 & 0 & 0 & 0 & 1 & 0 & 0\\
        0 & 0 & 0 & 0 & 0 & 1 & 0\\
        0 & 0 & 0 & 0 & 0 & 0 & 1
    \end{pmatrix}.\label{eq:ParityHammingCode}
\end{equation}
We can set
\begin{equation}
    \partial^\mathcal{C} = H,\qquad\partial^\mathcal{D} = H^T,
\end{equation}
and so for the stabilizers
\begin{align}
    H_X &= \begin{pmatrix}
        H\otimes\id_{7} & -\id_{3} \otimes H^T
    \end{pmatrix}\\
    H_Z &= \begin{pmatrix}
        \id_{7}\otimes H & H^T \otimes \id_{3}
    \end{pmatrix}.
\end{align}
This rotor code has \(58\) physical rotors and \(16\) logical rotors.
The minimal \(X\) distance is \(d_X=3\) and the minimal \(Z\) distance, \(\delta_Z\), is such that
\begin{equation}
3\geq \delta_Z(\alpha)\geq \frac{3\sin^2\left(\frac{\alpha}{6}\right)+9\sin^2\left(\frac{\alpha}{18}\right)}{\sin^2(\frac{\alpha}{2})}\xrightarrow[\alpha\rightarrow0]{}\frac{4}{9}.
\end{equation}
This last inequalties are obtained by picking for the upper bound a logical \(Z\) representative
\begin{equation}
    \overline{Z} = Z\left(\alpha\begin{pmatrix}
      \bs{e}_1\otimes\bs{g}_1  & 0
    \end{pmatrix}\right),
\end{equation}
and, for the lower bound, picking \(\Delta_X(\bs{e}_1\otimes\bs{g}_1)\) as in Eq.~\eqref{eq:candidateDeltaHP},
\begin{equation}
    \Delta_X(\bs{e}_1\otimes\bs{g}_1) = \left\{\begin{pmatrix}
        \bs{e}_1\otimes\bs{g}_1 & 0
    \end{pmatrix}, \begin{pmatrix}
        (\bs{e}_1+\bs{h}_1)\otimes\bs{g}_1 & 0
    \end{pmatrix}\right\}.
\end{equation}
To summarize the parameters we have
\begin{equation}
    \left\llbracket58, (16,0), (3,\delta_Z)\right\rrbracket_{\rm rot},\qquad 3\geq\delta_Z\geq\frac{4}{9}.
\end{equation}

\subsection{Torsion+Free Product: Logical Qudits}
In order to get some logical qudits in the product we can do the following:
Pick for \(\partial^\mathcal{C}\) a square matrix (\(n_C=m_C\)), with full-rank still, but also some torsion, that is to say 
\begin{equation}
    \mathbb{Z}^{n_C}/\im(\partial^\mathcal{C})=\mathbb{Z}_{d_1}\oplus\cdots\oplus\mathbb{Z}_{d_{k_C}},\qquad\ker({\partial^\mathcal{D}}) = \mathbb{Z}^{k_D}.\label{eq:torsionfreeH}
\end{equation}
The other terms in Eq.~\eqref{eq:homprod} vanish and we get, using Eq.~\eqref{eq:tensor}
\begin{equation}
    H_1(\mathcal{C}\otimes\mathcal{D}) = \mathbb{Z}^{n_C}/\im(\partial^\mathcal{C})\otimes \ker({\partial^\mathcal{D}}) = \left(\mathbb{Z}_{d_1}\oplus\cdots\oplus\mathbb{Z}_{d_{k_C}}\right)^{k_D}.
\end{equation}
In this case the rotor code encodes logical qudits.
The generating matrix for the logical \(X\) operators is given by the same expressions as Eq.~\eqref{eq:freefreeLX}.
For the \(Z\) logical operators, the expression is similar to Eq.~\eqref{eq:freefreeLZ}
\begin{equation}
    L_Z = \begin{pmatrix}
        G_C^\prime \otimes E_D & 0
    \end{pmatrix},\label{eq:freetorsionLZ}
\end{equation}
but the generating matrix \(G_C^\prime\) has a slightly different interpretation.
It is actually related to the weak boundaries defined in Eq.~\eqref{eq:weakboundary}.
The \(i\)th row of \(G_C^\prime\), \(\bs{g}_i\), is such that
\begin{equation}
    \partial^\mathcal{C}\bs{g}_i^T = d_i\bs{e}^T,\label{eq:weakboundaryC}
\end{equation}
for some integer vector \(\bs{e}\not\in\im(\partial^\mathcal{C})\) and \(d_i\) comes from Eq.~\eqref{eq:torsionfreeH}.
This allows for 
\begin{align}
    H_X\begin{pmatrix}
        \frac{2\pi}{d_i}\bs{g}_i\otimes \bs{e}^D_j & 0
    \end{pmatrix}^T &=\begin{pmatrix}
        \partial^\mathcal{C}\otimes\id_{n_D} & -\id_{m_C} \otimes \partial^\mathcal{D}
    \end{pmatrix}\begin{pmatrix}
        \frac{2\pi}{d_i}\bs{g}_i\otimes \bs{e}^D_j & 0
    \end{pmatrix}^T \\
    &= 2\pi\begin{pmatrix}
        \bs{e}\otimes\bs{e}^D_j & 0
    \end{pmatrix}^T,
\end{align}
where \(\bs{e}_j^D\) is the \(j\)th row of \(E_D\), i.e some element not generated by \(\partial^\mathcal{D}\).

One simple choice for \(\partial^\mathcal{C}\) is derived from the repetition code with a sign twist:
\begin{equation}
    \partial^\mathcal{C} = \begin{pmatrix}
        1 & -1 & 0 & \ddots & 0 & 0\\
        0 & 1 & -1 & \ddots & 0 & 0\\
        0 & 0 & 1 & \ddots & 0 & 0\\
        \ddots & \ddots & \ddots & \ddots & \ddots & \ddots\\
        0 & 0 & 0 & \ddots & 1 & -1\\
        1 & 0 & 0 & \ddots & 0 & 1
    \end{pmatrix}.
\end{equation}
This matrix is full rank and exhibits \(\mathbb{Z}_2\) torsion with \(\mathbb{Z}^{n_C}/\im(\partial^\mathcal{C}) = \mathbb{Z}_2\).
This implies that the logical codespace of the product code contains \(k_D\) qubits.
This also guarantees we can find a set of disjoint logical \(X\) representatives of size linear in \(n_C\) as in the previous section, see Eq.~\eqref{eq:candidateDeltaHP}.
All in all with this choice, a good classical LDPC code for \(\partial^\mathcal{D}\) with \(d^\mathcal{D}=\Theta(n_D)\) and \(k_D=\Theta(n_D)\), and choosing once again a `skewed shape' with \(n_C = \Theta(n_D^2)\), we can get a family of rotor codes with parameters
\begin{equation}
        \left\llbracket n, \left(0,2^{\Theta(\sqrt[3]{n})}\right) ,\left(\Theta(\sqrt[3]{n}),\Theta(\sqrt[3]{n})\right)\right\rrbracket_{\rm rot}.
\end{equation}

An other way to obtain a square matrix, \(\partial^\mathcal{C}\), with more torsion of even order say, is to take the rectangular, full-rank parity check matrix of a binary code, say \(H\in\{0,1\}^{m\times n}\) and define
\begin{equation}
    \partial^\mathcal{C} = H^TH \pmod 2.\label{eq:classicaltorsioncode}
\end{equation}
In some cases such matrix will be full rank over \(\mathbb{R}\) but not over \(\mathbb{F}_2\).
In particular for a codeword of the binary code \(\bs{g}\in\ker(H^T)\) we would have
\begin{equation}
    \partial^\mathcal{C}\bs{g}^T = H^TH\bs{g}^T = 2\bs{e}^T,
\end{equation}
for some integer vector \(\bs{e}\not \in \im(\partial^\mathcal{C})\) which then represents some even order torsion.
This choice can yield better encoding rates but we do not have a general way of bounding the \(\delta_Z\) distance in this case.

For instance, for the \([7,4,3]\) Hamming code given in Eq.~\eqref{eq:ParityHammingCode}, we have
    \begin{equation}
    \partial^\mathcal{C} = H^TH\pmod{2} = \begin{pmatrix}
        0& 1& 0& 1& 1& 1& 0\\
        1& 0& 0& 1& 0& 1& 1\\
        0& 0& 1& 0& 1& 1& 1\\
        1& 1& 0& 0& 1& 0& 1\\
        1& 0& 1& 1& 1& 0& 0\\
        1& 1& 1& 0& 0& 1& 0\\
        0& 1& 1& 1& 0& 0& 1
    \end{pmatrix},\label{eq:SquareHammingCode}
\end{equation}
which is full rank, so \(\ker(\partial^\mathcal{C})= \{\bs{0}\}\).
We have $\mathbb{Z}^{n_C}/\im(\partial^\mathcal{C}) = \mathbb{Z}_2\oplus\mathbb{Z}_2\oplus\mathbb{Z}_2\oplus\mathbb{Z}_4$, and
\begin{equation}
    \qquad G_C^\prime = \begin{pmatrix}-1& 0& 0& 1& 0& 1& -1 \\-1& 0& 0& 0& 1& 1& 0\\-1& 1& 0& 0& 1& 0& -1\\-1& 1& -1& 1& 1& 1& -1\end{pmatrix},\qquad E^\prime_C = \begin{pmatrix}
        1& 0& 0& -1& 0& 0& 0\\
        1& 0& 1& 0& 0& 0& 0\\
        1& -1& 0& 0& 0& 0& 0\\
        1& 0& 0& 0& 0& 0& 0
    \end{pmatrix},
\end{equation}
where \(E^\prime_C\) gives generators for \(\mathbb{Z}^{n_C}/\im(\partial^\mathcal{C})\), the last row of \(E^\prime_C\) is the generator of order \(4\) and \(G_C^\prime\) is related to \(E_C^\prime\) according to Eq.~\eqref{eq:weakboundaryC}.

So choosing \(\partial^\mathcal{C} = H^TH\pmod 2\) and \(\partial^\mathcal{D} = H^T\) yields a code with parameters
\begin{equation}
    \left\llbracket70, (0, 2^{12}\cdot4^4), (3,\delta_Z)\right\rrbracket_{\rm rot},\qquad 3\geq\delta_Z\geq 18\sin^2\left(\frac{\pi}{12}\right).
\end{equation}
 For the lower bound on $\delta_Z$ we pick \(\Delta_X(\bs{e}_1^\prime\otimes\bs{g}_1)\) as in Eq.~\eqref{eq:candidateDeltaHP},
\begin{equation}
    \Delta_X(\bs{e}^\prime_1\otimes\bs{g}_1) = \left\{\begin{pmatrix}
        \bs{e}_1^\prime\otimes\bs{g}_1 & 0
    \end{pmatrix}, \begin{pmatrix}
        (\bs{e}_1^\prime+\bs{\partial}^\mathcal{C}_2)\otimes\bs{g}_1 & 0
    \end{pmatrix}, \begin{pmatrix}
        (\bs{e}_1^\prime+\bs{\partial}^\mathcal{C}_5)\otimes\bs{g}_1 & 0
    \end{pmatrix}\right\}.
\end{equation}

\subsection{Torsion+Torsion Product: Logical Qudits}
Finally, one can take for {\em both} \(\partial^\mathcal{C}\) and \({\partial^\mathcal{D}}\) such full-rank square matrices with some torsion, and then get a logical space through the \({\rm Tor}\) part of Eq.~\eqref{eq:kunneth}.
Since the torsion groups have to agree (see Eq.~\eqref{eq:Tor}) we can pick, for instance, both \(\partial^{\mathcal{C}}\) and \({\partial^\mathcal{D}}\) in the same way as Eq.~\eqref{eq:classicaltorsioncode}, ensuring that they have a common \(\mathbb{Z}_2\) part.
We would have 
\begin{equation}
    \mathbb{Z}^{n_C}/\im(\partial^\mathcal{C})=\mathbb{Z}_{d_1}\oplus\cdots\oplus\mathbb{Z}_{d_{k_C}},\qquad\mathbb{Z}^{m_D}/\im({\partial^\mathcal{D}})=\mathbb{Z}_{p_1}\oplus\cdots\oplus\mathbb{Z}_{p_{k_D}},
\end{equation}
where the \(d_i\) and \(p_j\) are all even integers.
All other terms vanish and we get
\begin{equation}
    H_1(\mathcal{C}\otimes\mathcal{D}) = {\rm Tor}\left(\mathbb{Z}^{n_C}/\im(\partial^\mathcal{C}), \mathbb{Z}^{m_D}/\im({\partial^\mathcal{D}})\right)  = \bigoplus_{i,j}\mathbb{Z}_{\gcd(d_i,p_j)}.\label{eq:torsiontorsionhomology}
\end{equation}
Since the \(d_i\) and \(p_j\) were all chosen to be even, each term of Eq.~\eqref{eq:torsiontorsionhomology} is at least \(\mathbb{Z}_2\) or larger. 
The \(X\) logical operators are now slightly different than Eq.~\eqref{eq:freefreeLX}, 
\begin{equation}
    L_X = \begin{pmatrix}
        \Lambda_C\left(E^\prime_C\otimes G^\prime_D\right) & -\Lambda_D\left(G^\prime_C\otimes E^\prime_D\right)
    \end{pmatrix}.\label{eq:torsiontorsionLX}
\end{equation}
The matrices \(\Lambda_C\) and \(\Lambda_D\) are integer diagonal matrices of size \((k_C\times k_D)^2\) defined as
\begin{equation}
    \Lambda_C = \Diag{\frac{d_i}{\gcd\left(d_i,p_j\right)}},\qquad \Lambda_D = \Diag{\frac{p_j}{\gcd\left(d_i,p_j\right)}}.
\end{equation}
They are such that given the matrices with the torsion orders on the diagonal, \(D_C\) and \(D_D\),
\begin{equation}
    D_C =\Diag{d_1,d_2,\ldots,d_{k_C}},\qquad D_D =\Diag{p_1,p_2,\ldots,p_{k_D}},
\end{equation}
 we have that
\begin{equation}
    \left(\id_{k_C}\otimes D_D\right)\Lambda_C = \left(D_C\otimes\id_{k_D}\right)\Lambda_D = \Diag{\frac{d_ip_j}{\gcd\left(d_i,p_j\right)}}.\label{eq:commontorsiondiag}
\end{equation}
The matrices \(G^\prime_C\) and \(E^\prime_C\), and \(G^\prime_D\) and \(E^\prime_D\), in Eq.~\eqref{eq:torsiontorsionLX}, are related to the weak boundaries.
More precisely, we have 
\begin{align}
    {\partial^\mathcal{C}}^T {G^\prime_C}^T &= {E^\prime_C}^T D_C,\\ 
    {\partial^\mathcal{D}}^T {G^\prime_D}^T &= {E^\prime_D}^T D_D.
\end{align}
With this we can check that
\begin{align}
    H_Z L_X^T &= \begin{pmatrix}
        \id_{n_C}\otimes {\partial^\mathcal{D}}^T & {\partial^\mathcal{C}}^T \otimes \id_{m_D}
    \end{pmatrix} \begin{pmatrix}
        \Lambda_C\left(E^\prime_C\otimes G^\prime_D\right) & -\Lambda_D\left(G^\prime_C\otimes E^\prime_D\right)
    \end{pmatrix}^T\\
    &=\left({E_C^\prime}^T\otimes {E^\prime_D}^TD_D\right)\Lambda_C - \left({E_C^\prime}^TD_C \otimes {E^\prime_D}^T\right)\Lambda_D\\
    &= \left({E_C^\prime}^T\otimes {E^\prime_D}^T\right)\left(\id\otimes D_D\right)\Lambda_C -\left({E_C^\prime}^T\otimes {E^\prime_D}^T\right)\left(D_C\otimes\id\right)\Lambda_D\\
    &=0,
\end{align}
where the last line is obtained from the previous one using Eq.~\eqref{eq:commontorsiondiag}.
One can see that it is crucial to get some common divisor in the torsion of \(\mathcal{C}\) and \(\mathcal{D}\).
If there is no common divisor for some pair \((d_i,p_j)\) then the corresponding entries in \(\Lambda_C\) and \(\Lambda_j\) would be \(d_i\) and \(p_j\) respectively.
It would still follow that the corresponding row in \(L_X\) commutes with the \(Z\) stabilizers but it would also be trivial (because generated by the \(X\) stabilizers) since we have that
\begin{align}
    {\bs{g}_C}_i^\prime\otimes{\bs{g}_D}_j^\prime H_X &={\bs{g}_C}_i^\prime\otimes{\bs{g}_D}_j^\prime \begin{pmatrix}
        \partial^\mathcal{C}\otimes\id_{n_D} & -\id_{m_C} \otimes \partial^\mathcal{D}
    \end{pmatrix}\\
    &=\begin{pmatrix}
    d_i{\bs{e}_C}_i^\prime\otimes{\bs{g}_D}_j^\prime & -p_j{\bs{g}_C}_i^\prime\otimes{\bs{e}_D}_j^\prime 
\end{pmatrix}.
\end{align}
For the \(Z\) logical operators we have
\begin{equation}
    L_Z = \begin{pmatrix}
        U\left(G_C^\prime\otimes E_D^\prime\right) & V\left(E_C^\prime\otimes G_D^\prime\right)
    \end{pmatrix},
\end{equation}
where the matrices \(U\) and \(V\) are diagonal matrices of size \((k_C\times k_D)^2\) recording the smallest Bézout coefficients for all pairs \((d_i, p_j)\), i.e so that \(d_iu_{ij} - p_jv_{ij} = \gcd(d_i,p_j)\), we therefore have that
\begin{equation}
    \left(D_C\otimes\id\right)U - \left(\id\otimes D_D\right)V = \Diag{\gcd(d_i,p_j)}.
\end{equation}
We can check that
\begin{align}
    H_XL_Z^T &= \begin{pmatrix}
        \partial^\mathcal{C}\otimes\id_{n_D} & -\id_{m_C} \otimes \partial^\mathcal{D}
    \end{pmatrix} \begin{pmatrix}
        U\left(G_C^\prime\otimes E_D^\prime\right) & V\left(E_C^\prime\otimes G_D^\prime\right)
    \end{pmatrix}^T\\
    &= \left({E_C^\prime}^T D_C\otimes {E_D^\prime}^T\right)U - \left({E_C^\prime}^T\otimes {E_D^\prime}^T D_D\right)V\\
    &=\left({E_C^\prime}^T\otimes {E_D^\prime}^T \right)\left(\left(D_C\otimes\id\right)U - \left(\id\otimes D_D\right)V\right)\\
    &=\left({E_C^\prime}^T\otimes {E_D^\prime}^T \right)\Diag{\gcd(d_i,p_j)}.
\end{align}
We see that this would be trivial when multiplied by phases in \(\mathbb{Z}^*_{\gcd(d_i,p_j)}\).

As for a concrete example, we can again use the Hamming code and the square matrix in Eq.~\eqref{eq:SquareHammingCode} for both \(\partial^\mathcal{C}\) and \(\partial^\mathcal{D}\).
This example as well as other examples of square parity check matrices of classical codes can be found in \cite{quintavallePartitioningQubitsHypergraph2022} where  they are used in a product for different reasons.
The parameters are given by
\begin{equation}
    \left\llbracket 98, (0,2^{15}\cdot 4), (3,\delta_Z)\right\rrbracket_{\rm rot},\qquad 3\geq\delta_Z.
\end{equation}
Observe that compared to the qubit version, --a qubit code with parameters \(\llbracket98,32,3\rrbracket\) given in \cite{quintavallePartitioningQubitsHypergraph2022}, the quantum rotor code has a bit above half as many qubits.
This is due to the fact that in the qubit case the two blocks do not have to conspire to form logical \(X\) operators but can do so independently.
    
    \section{Schrieffer-Wolff Perturbative Analysis of the Four-Phase Gadget}
	\label{app:SW4phase}

In this Appendix, we derive an effective low-energy Hamiltonian for the four-phase gadget introduced in Subsec.~\ref{subsec:fourphase} and shown in Fig.~\ref{fig:4phase}. We will focus on the regime where $C \gg C_g, C_J$ and when $E_J$ is smaller than the typical energy of an agiton excitation. Our analysis will closely follow the one of the current-mirror qubit in Ref.~\cite{weissSpectrumCoherenceProperties2019}. In order to perform the perturbative analysis, we work with exciton and agiton variables, and our starting point is the quantized version of the classical Hamiltonian shown in Eq.~\eqref{eq:exaghamil}. 

When $C \gg C_g, C_J$ the charging energy of a single exciton in either the left or the right rung is given by $E_C^{(e)}$ given in Eq.~\eqref{eq:en1exc} in the main text, which is much smaller than the energy of a single agiton $4 E_{C,11}^{(a)} = 4 e^2 \frac{C_g + C_J}{C_g(C_g + 2 C_J)}$. We want to obtain an effective low-energy Hamiltonian for the zero-agiton subspace defined as

\begin{equation}
\mathcal{H}_{a=0} = \bigl\{\ket{\psi} \, \lvert \quad \hat{\ell}_{L, Ra} \ket{\psi}  = 0 \bigr\}.
\end{equation}
We will denote the non-zero agiton subspace, perpendicular to $\mathcal{H}_{a=0}$ as $\mathcal{H}_{a\neq0}$.

We carry out the perturbative analysis using a Schrieffer-Wolff transformation following Ref.~\cite{BRAVYI20112793}. In our case, the unperturbed Hamiltonian is the charging term 
\begin{equation}
H_0 = 4 E_{C,11}^{(e)} \bigl(\hat{\ell}_{Le}^2 + \hat{\ell}_{Re}^2 \bigr) + 8 E_{C, 12}^{(e)} \hat{\ell}_{Le} \hat{\ell}_{Re} +  4 E_{C, 11}^{(a)} \bigl(\hat{\ell}_{La}^2 + \ell_{Ra}^2 \bigr) + 8 E_{C, 12}^{(a)} \hat{\ell}_{La} \hat{\ell}_{Ra}
\end{equation}
while the perturbation is given by the Josephson contribution
\begin{equation}
\label{eq:vpert}
V = -2 E_J \cos \biggr[\frac{1}{2}\bigl(\hat{\theta}_{La} - \hat{\theta}_{Ra} \bigr) \biggr]\cos \biggr[\frac{1}{2}\bigl(\hat{\theta}_{Le} - \hat{\theta}_{Re} \bigr) \biggr]. 
\end{equation}
In order to understand the relevant processes, let us consider a state with $m_e \in \mathbb{Z}$ excitons on the left rung and $m_e' \in \mathbb{Z}$ excitons on the right rung. We denote this state as $\ket{m_e, m_e'}$ and it is formally defined as a state in the charge basis such that
\begin{eqnarray}
\ket{m_e, m_e' }  = & \ket{\ell_{Le} = m_e, \ell_{Re} = m_e', \ell_{La} = 0, \ell_{Ra} = 0}  \notag \\
 = &\ket{\ell_1 = m_e, \ell_2 = m_e', \ell_3 = -m_e, \ell_4 = -m_e'}.
\end{eqnarray}
We also wrote the state in terms of the original charge operators because the action of the Josephson potential is more readily understood in terms of these variables. In fact, the Josephson junction on the upper (lower) part of the circuit in Fig.~\ref{fig:4phase} allows the tunneling of a single Cooper-pair from node $1$ ($3$) to node $2$ ($4$), and vice versa. Thus, a transition $\ket{m_e, m_e' }\leftrightarrow \ket{m_e - 1, m_e' + 1}$ within the zero-agiton subspace can be effectively realized via two processes which are mediated by the non-zero agiton states
\begin{subequations}
\label{eq:right_ag_states}
\begin{multline}
\ket{a_+; m_e, m_e' } = \ket{ \ell_{Le} = m_e - \frac{1}{2}, \ell_{Re} = m_e' + \frac{1}{2}, \ell_{La} = -\frac{1}{2}, \ell_{Ra} = \frac{1}{2}} \\ = 
\ket{\ell_1 = m_e -1, \ell_2 = m_e' + 1, \ell_3 = -m_e, \ell_4 = -m_e'},
\end{multline}
\begin{multline}
\ket{a_-; m_e, m_e'} = \ket{\ell_{Le} = m_e - \frac{1}{2}, \ell_{Re} = m_e' + \frac{1}{2}, \ell_{La} = +\frac{1}{2}, \ell_{Ra} = -\frac{1}{2}} \\ = 
\ket{\ell_1 = m_e, \ell_2 = m_e', \ell_3 = -m_e + 1, \ell_4 = -m_e' -1}.
\end{multline}
\end{subequations}
The two processes are 
\begin{subequations}
\begin{equation}
   \ket{m_e, m_e'} \rightarrow \ket{a_+; m_e, m_e'} \rightarrow \ket{m_e-1, m_e' + 1},
\end{equation}
\begin{equation}
\ket{m_e, m_e'} \rightarrow \ket{a_-; m_e, m_e' }\rightarrow \ket{m_e-1, m_e' + 1},
\end{equation}
\end{subequations}
as well as the opposite processes.

The above-mentioned processes cause hopping of two Cooper-pairs from one rung to the other. However, we also have to take into account the second order processes which bring back a single Cooper-pair to the rung where it was coming from. These processes are also mediated by the non-zero agiton states in Eq.~\eqref{eq:right_ag_states}.
The effect of these processes is to shift the energy of  $\ket{m_e, m_e' }$, but we will see that in a first approximation each state $\ket{m_e, m_e' }$ gets approximately the same shift and so we are essentially summing an operator proportional to the identity in the zero-agiton subspace, which is irrelevant.

The charging energy of an exciton state $\ket{m_e, m_e' }$ is 

\begin{equation}
E(m_e, m_e' ) = 4 E_{C,11}^{(e)} \bigl(m_e^2 + m_e'^2 \bigr) + 8 E_{C, 12}^{(e)} m_e m_e',
\end{equation}
while the intermediate non-zero agiton states $\ket{a_{\pm}; m_e, m_e}$ both have the same charging energy
\begin{multline}
\label{eq:ecadiff}
E(a; m_{e}, m_{e}') = 4 E_{C, 11}^{(e)} \biggl(m_{e} - \frac{1}{2} \biggr)^2 + 4 E_{C,11}^{(e)} \biggl(m_{e}' + \frac{1}{2} \biggr)^2 + 8 E_{C, 12}^{(e)}\biggl(m_{e} - \frac{1}{2} \biggr)\biggl(m_{e}' + \frac{1}{2} \biggr) \\
+2 \bigl ( E_{C,11}^{(a)} - E_{C,12}^{(a)} \bigr) \approx 2 \bigl ( E_{C,11}^{(a)} - E_{C,12}^{(a)} \bigr) = \frac{2 e^2}{C_g + 2  C_J} \equiv 2 E_{C, \mathrm{diff}}^{(a)}
\end{multline}
where the last approximation is valid when $m_e, m_e'$ are small.
Analogously, the states $\ket{a_{\pm}; m_e, m_e' }$ have energy

\begin{multline}
E(a; m_e, m_e') = 4 E_{C,11}^{(e)} \biggl(m_e + \frac{1}{2} \biggr)^2 + 4 E_{C,11}^{(e)} \biggl(m_e' - \frac{1}{2} \biggr)^2 \\ +
8 E_{C,12}^{(e)}\biggl(m_e +\frac{1}{2} \biggr)\biggl(m_e' - \frac{1}{2} \biggr) 
+2 \bigl ( E_{C,11}^{(a)} - E_{C,12}^{(a)} \bigr) \approx  2 E_{C, \mathrm{diff}}^{(a)}.
\end{multline}

In the limit of $C \gg C_g, C_J$ we can also make the approximation

\begin{equation}
\label{eq:deltae_approx}
E(a; m_e, m_e') - E(m_e, m_e') \approx 2 E_{C, \mathrm{diff}}^{(a)} .
\end{equation}
The matrix elements of the perturbation, i.e., the Josephson potential in Eq.~\eqref{eq:vpert} between the exciton states and the intermediate agiton states are given by

\begin{equation}
\label{eq:mat_elem_v}
\braket{m_e, m_e' | V |a_{\pm}; m_e, m_e'}  = -\frac{E_J}{2}.
\end{equation}
Importantly, these are the only non-zero matrix elements of $V$ between states in $\mathcal{H}_{a=0}$ and $\mathcal{H}_{a \neq 0}$.

We now proceed with the perturbative Schrieffer-Wolff analysis. Let $\mathcal{E}_{a =0}$ ($\mathcal{E}_{a \neq 0}$) be the set of eigenvalues of $H_0$ associated with eigenvectors in $\mathcal{H}_{a=0}$ ($\mathcal{H}_{a\neq0}$). We denote by $P_{a =0}$ the projector onto the zero-agiton subspace and by $P_{a \neq 0}$ the projector onto the subspace orthogonal to it.  We define the block-off-diagonal operator associated with the perturbation $V$ in Eq.~\eqref{eq:vpert} as

\begin{equation}
V_{\mathrm{od}} = P_{a=0} V P_{a \neq 0} + P_{a \neq 0} V P_{a=0} = P_{a=0} V P_{a \neq 0} + \mathrm{h.c.}
\end{equation}
The effective second order Schrieffer-Wolff Hamiltonian is given by

\begin{equation}
\label{eq:h_eff_gen}
H_{\mathrm{eff}} = P_{a=0} (H_0 + V) P_{a=0} + \frac{1}{2} P_{a=0} [S_1, V_{\mathrm{od}} ] P_{a=0}, 
\end{equation}
where the approximate generator of the Schrieffer-Wolff transformation is given by the anti-Hermitian operator
\begin{equation}
S_1 = \tilde{S}_1 - \mathrm{h.c.},
\end{equation}
with 
\begin{equation}
\tilde{S}_1 = \sum_{E \in \mathcal{E}_{a =0}} \sum_{E' \in \mathcal{E}_{a \neq 0}} \frac{\braket{E | V_{\mathrm{od}}|E'}}{E - E'} \ket{E} \bra{E'}.
\end{equation}
The first term on the right-hand side of Eq.~\eqref{eq:h_eff_gen} reads
\begin{equation}
P_{a=0} (H_0 + V) P_{a=0} = 4 E_{C,11}^{(e)} (\hat{\ell}_{0, L}^{(e)})^2 +  4 E_{C,11}^{(e)} (\hat{\ell}_{0,R}^{(e)})^2 + 8  E_{C,12}^{(e)} \hat{\ell}_{0,L}^{(e)} \hat{\ell}_{0, R}^{(e)},
\end{equation}
where the operators $\hat{\ell}_{0, L}$ and $\hat{\ell}_{0, R}^{(e)}$ are the projection of the exciton charge operators $\hat{\ell}_{L, R}^{(e)}$ onto the zero-agiton subspace, i.e., $\hat{\ell}_{0, L}^{(e)} = P_{a=0} \hat{\ell}_{L}^{(e)} P_{a=0}$ and $\hat{\ell}_{0, R}^{(e)} = P_{a=0} \hat{\ell}_{R}^{(e)} P_{a=0}$. The operators $\hat{\ell}_{0, L}^{(e)}$ and $\hat{\ell}_{0, R}^{(e)}$ have only integer eigenvalues. 

The explicit calculation of $\tilde{S}_1$ using Eq.~\eqref{eq:mat_elem_v} gives
\begin{equation}
\tilde{S}_1 = - \frac{1}{2}\sum_{m_e, m_e' \in \mathbb{Z}} \sum_{s =\pm}\frac{E_J}{E(m_e, m_e') - E(a;m_e, m_e')} \ket{m_e, m_e'} \bra{a_s;m_e, m_e'} 
-\mathrm{h.c.}.
\end{equation}
Straightforward algebra shows that the last term in Eq.~\eqref{eq:h_eff_gen} can be simply rewritten as
\begin{equation}
\label{eq:2ndSW}
\frac{1}{2} P_{a=0} [S_1, V_{\mathrm{od}}] P_{a=0} = \frac{1}{2}\tilde{S}_1 P_{a \neq 0} V P_{a = 0} + \mathrm{h.c.}.
\end{equation}
Writing explicitly $P_{a \neq 0} V P_{a = 0}$ gives
\begin{equation}
P_{a \neq 0} V P_{a =0} = - \frac{E_J}{2} \sum_{m_e, m_e' \in \mathbb{Z}} \sum_{s=\pm} \ket{a_s;m_e, m_e'} \bra{m_e, m_e'}. 
\end{equation}
Plugging into Eq.~\eqref{eq:2ndSW}, we obtain\footnote{Notice that all the $\sum_{\pm}$ sums account for a factor of 2 here.}

\begin{multline}
\frac{1}{2} P_{a=0} [S_1, V_{\mathrm{od}}] P_{a=0} \!= \!\frac{1}{4} \sum_{m_e, m_e' \in \mathbb{Z}} \frac{E_J^2}{E(m_e, m_e') -  E(a;m_e, m_e')} 
\ket{m_e, m_e'}\bra{m_e, m_e'} \\
+  \frac{1}{4} \sum_{m_e, m_e' \in \mathbb{Z}} \frac{E_J^2}{E(m_e, m_e') -  E(a;m_e, m_e')} \ket{m_e, m_e'}\bra{m_e-1, m_e' + 1}  + \mathrm{h.c.} \\ = 
\frac{1}{2} \sum_{m_e, m_e' \in \mathbb{Z}} \frac{E_J^2}{E(m_e, m_e') -  E(a;m_e, m_e')} \ket{m_e, m_e'}\bra{m_e, m_e'} \\ +  \frac{1}{4} \biggl( \sum_{m_e, m_e' \in \mathbb{Z}} \frac{E_J^2}{E(m_e, m_e') -  E(a;m_e, m_e')} \ket{m_e, m_e'}\bra{m_e-1, m_e'+1} + \mathrm{h.c.} \biggr).
\end{multline}
We obtain the effective Hamiltonian
\begin{multline}
H_{\mathrm{eff}} = 4 E_{C,11}^{(e)} (\hat{\ell}_{0, L}^{(e)})^2 +  4 E_{C, 11}^{(e)} (\hat{\ell}_{0, R}^{(e)})^2 + 8  E_{C,12}^{(e)} \hat{\ell}_{0, L}^{(e)} \hat{\ell}_{0, R}^{(e)} \\ + 
\frac{1}{2} \sum_{m_e, m_e' \in \mathbb{Z}}  \frac{E_J^2}{E(m_e, m_e') -  E(a ;m_e, m_e')} \ket{m_e, m_e'}\bra{m_e, m_e'} + 
\\ \frac{1}{4} \biggl( \sum_{m_e, m_e' \in \mathbb{Z}} \frac{E_J^2}{E(m_e, m_e') -  E(a;m_e, m_e')} \ket{m_e, m_e'}\bra{m_e-1, m_e'+1} + \mathrm{h.c.} \biggr).
\end{multline}
Using the approximation in Eq.~\eqref{eq:deltae_approx} we can write the effective Hamiltonian more compactly as

\begin{multline}
H_{\mathrm{eff}} = 4 E_{C,11}^{(e)} (\hat{\ell}_{0,L}^{(e)})^2 +  4 E_{C, 11}^{(e)} (\hat{\ell}_{0,R}^{(e)})^2 + 8  E_{C,12}^{(e)} \hat{\ell}_{0, L}^{(e)} \hat{\ell}_{0, R}^{(e)} \\ - 
\frac{E_J^2}{4 E_{C, \mathrm{diff}}^{(a)}} \biggl( \frac{1}{2} \sum_{m_e, m_e' \in \mathbb{Z}} \ket{m_e, m_e'}\bra{m_e-1, m_e'+1} + \mathrm{h.c.} \biggr) - \frac{E_J^2}{4 E_{C, \mathrm{diff}}^{(a)}} P_{a=0},
\end{multline}
where the last term proportional to $P_{a=0}$ can be neglected since it is the identity on the zero-agiton subspace. 
We identify the operators $e^{i \theta_{0,L}}$ and $e^{i \theta_{0,R}}$ as
\begin{equation}
e^{i \theta_{0,L}^{(e)}} =  \sum_{m_{e}\in \mathbb{Z}} \ket{m_{e}}\bra{m_{e}-1} \otimes I, \quad e^{i \theta_{0,R}^{(e)}} = I \otimes \sum_{m_{e}'\in \mathbb{Z}} \ket{m_{e}'}\bra{m_{e}'-1}
\end{equation}
where the subscript $0$ specifies that these are defined within the zero-agiton subspace.

In this way we rewrite the effective Hamiltonian as 
\begin{equation}
\label{eq:h_eff_4phase}
H_{\mathrm{eff}} = 4 E_{C,11}^{(e)} (\hat{\ell}_{0, L}^{(e)})^2 +  4 E_{C,11}^{(e)} (\hat{\ell}_{0, R}^{(e)})^2 + 8  E_{C,12}^{(e)} \hat{\ell}_{0, L}^{(e)} \hat{\ell}_{0, R}^{(e)} -E_{J, \mathrm{eff}}\cos(\hat{\theta}_{0, L}^{(e)} - \hat{\theta}_{0, R}^{(e)}  ),
\end{equation}
with the effective Josephson energy $E_{J, \mathrm{eff}}$ defined in Eq.~\eqref{eq:ejeff} of the main text.
Finally, in terms of the original node variables projected onto the zero-agiton subspace that we denote as $\hat{\theta}_{0, k}$ we get
\begin{equation}
\label{eq:cos_4phase}
 \cos(\hat{\theta}_{0, L}^{(e)} - \hat{\theta}_{0, R}^{(e)}  ) = \cos(\hat{\theta}_{0, 1} + \hat{\theta}_{0,4} - \hat{\theta}_{0, 3} -  \hat{\theta}_{0, 2}).
\end{equation}

\end{appendices}


\bibliography{rotor-codes}

\end{document}